\documentclass[11pt]{article}
\usepackage{amsmath}
\usepackage{enumerate}
\usepackage{natbib}
\usepackage{url} 

\newcommand{\blind}{0}

\usepackage{amsfonts}
\usepackage{amssymb}
\usepackage{amsthm}
\usepackage{enumitem}
\usepackage{tabularx}
\usepackage{multirow}
\usepackage{float}
\usepackage{graphicx}
\usepackage{epstopdf}
\usepackage{color}

\usepackage{subfig, caption, multicol}

\usepackage{booktabs}
\usepackage{pdflscape}
\usepackage{float}
\floatstyle{plaintop}
\restylefloat{table}

\restylefloat{table}

\theoremstyle{plain}
\newtheorem{theorem}{Theorem}[section]

\newtheorem{assumption}{Assumption}[section]

\newtheorem{corollary}{Corollary}[section]

\newtheorem{lemma}{Lemma}[section]
\newtheorem{proposition}{Proposition}[section]

\theoremstyle{definition}
\newtheorem{example}{Example}[section]

\theoremstyle{remark}
\newtheorem{remark}{Remark}[section]

\newcommand{\PP}{\mathbb{P}}
\newcommand{\N}{\mathbb{N}}

\newcommand{\R}{\mathbb{R}}
\newcommand{\Z}{\mathbb{Z}}

\newcommand{\E}{\mathbb{E}}

\def\bi{\begin{itemize}}
	\def\ei{\end{itemize}}

\allowdisplaybreaks
\graphicspath{{Figures/}}
\numberwithin{equation}{section}

\newif\ifi

\UseRawInputEncoding

\newcommand{\drift}{\ensuremath{\gamma}}
\newcommand{\lev}{\ensuremath{\eta}} 
\newcommand{\ind}{\ensuremath{\Bbb{I}}}

\begin{document}

	\def\spacingset#1{\renewcommand{\baselinestretch}%
		{#1}\small\normalsize} \spacingset{1}

	
	\if0\blind
	{
		\title{Inference and forecasting for continuous-time integer-valued trawl processes}
		
		\author{Mikkel Bennedsen\thanks{
				Department of Economics and Business Economics and CREATES, 
				Aarhus University, 
				Fuglesangs All\'e 4,
				8210 Aarhus V, Denmark.
				E-mail: {\tt mbennedsen@econ.au.dk}. 
			}, 
			Asger Lunde\thanks{
				Copenhagen Economics, Langebrogade 1B, 1411 Copenhagen K, Denmark,
				and CREATES, 
				Aarhus University, 
				Fuglesangs All\'e 4,
				8210 Aarhus V, Denmark.
				E-mail: {\tt alu@CopenhagenEconomics.com}. 
			}, Neil Shephard\thanks{
				Department of Economics and Department of Statistics, 
				Harvard University, 
				One Oxford Street,
				Cambridge, MA 02138, USA.
				E-mail: {\tt shephard@fas.harvard.edu}.
			},  Almut E. D. Veraart\thanks{
				Department of Mathematics, 
				Imperial College London, 
				South Kensington Campus,
				London SW7 2AZ, UK and CREATES, Aarhus University.
				E-mail: {\tt a.veraart@imperial.ac.uk}.
		}}
		
		\maketitle
	} \fi
	
	\if1\blind
	{
		\bigskip
		\bigskip
		\bigskip
		\begin{center}
			{\LARGE\bf Inference and forecasting for continuous-time integer-valued trawl processes and their use in financial economics}
		\end{center}
		\medskip
	} \fi
	
	\bigskip
	\begin{abstract}
		This paper develops likelihood-based methods for estimation, inference, model selection, and forecasting of continuous-time integer-valued trawl processes. The full likelihood of integer-valued trawl processes is, in general, highly intractable, motivating the use of composite likelihood methods, where we consider the pairwise likelihood in lieu of the full likelihood. Maximizing the pairwise likelihood of the data yields an estimator of the parameter vector of the model, and we prove consistency and, in the short memory case, asymptotic normality of this estimator. When the underlying trawl process has long memory, the asymptotic behaviour of the estimator is more involved; we present some partial results for this case. The pairwise approach further allows us to develop probabilistic forecasting methods, which can be used to construct the predictive distribution of integer-valued time series. In a simulation study, we document the good finite sample performance of the likelihood-based estimator and the associated model selection procedure. Lastly, the methods are illustrated in an application to modelling and forecasting financial bid-ask spread data, where we find that it is beneficial to carefully model both the marginal distribution and the autocorrelation structure of the data. 
	\end{abstract}
	
	\noindent%
	
	\vfill
	
	\newpage
	\spacingset{1.5} 
	
	\section{Introduction}
	
	
	In this paper, we develop likelihood-based methods for estimation, inference, model selection, and forecasting of continuous-time integer-valued trawl (IVT) processes. IVT processes,  introduced in \cite{BNLSV2014},  are a  flexible class of integer-valued, serially correlated, stationary, and infinitely divisible continuous-time stochastic processes. In general, however, IVT processes are not Markovian, which implies that the structure of the full likelihood of an IVT process is highly intractable \citep[][]{SY2016}. This is the impetus of the present paper, where we propose to use composite likelihood \citep[CL,][]{Lindsay88} methods for estimation and inference. Specifically, we propose to estimate the parameters of an IVT model by maximizing 
	the  pairwise likelihood of the data. CL methods in general, and the pairwise likelihood approach in particular, have been successfully used in many applications, such as statistical genetics \citep[][]{LF2011}, geostatistics \citep[][]{HO1994}, and finance \citep[][]{EPSS2020}. See \cite{VRF2011} for an excellent overview of CL methods. Although the theory behind CL estimation is quite well understood in the case of iid observations \citep[e.g.][]{CN2004,VV2005}, the time series case, which is what we consider here, generally requires separate treatment \citep[][p. 11]{VRF2011}. For instance, \cite{DY2011} develops the theory of CL estimators in the setting of linear Gaussian time series models, while \cite{CHW2016} and \cite{NJKL2011} consider CL methods for a hidden Markov model and a time series model with a latent autoregressive process, respectively. Also, \cite{S2019} develops a two-step CL estimation method for parameter-driven count time series models with covariates. Our paper adds to the literature on CL methods for time series models by deriving the theoretical properties (consistency, asymptotic normality) of a pairwise CL estimator applied to IVT models.

	
	A central feature of IVT processes is that they allow for specifying the correlation structure of the model separately from the marginal distribution of the model, making them flexible and well-suited for modelling count- or integer-valued data. In particular, the marginal distribution of an IVT process can be any integer-valued infinitely divisible distribution, while the correlation structure can be specified independently using a so-called trawl function. This setup allows for both short- and long-memory of the IVT process. So far, IVT processes have been applied to financial data \citep{BNLSV2014,SY2017,VERAART2019} and real-valued trawl processes to the modelling of extreme events in environmental time series \citep{NVG2018}. IVT processes are, under weak conditions, stationary and ergodic, which motivated  \cite{BNLSV2014} to suggest a method-of-moments-based estimator for the parameters of the IVT model. This method-of-moments-based estimator has been used in most applied work using IVT processes \cite[e.g.][]{BNLSV2014,SY2017,VERAART2019}. Exceptions are  \cite{SY2016} and \cite{NVG2018}. In  \cite{NVG2018}, a pairwise likelihood  was used for a hierarchical model involving a latent (Gamma-distributed) trawl process and the corresponding asymptotic theory was derived in   \cite{CV2020}.  However, the asymptotic theory for inference for integer-valued trawl processes which are observed directly is not covered by these earlier papers.
	In \cite{SY2016}, the authors derive a prediction decomposition of the likelihood function of a particularly simple IVT process, the so-called Poisson-Exponential IVT process, allowing them to conduct likelihood-based estimation and inference. Although the likelihood estimation method developed in \cite{SY2016} theoretically applies to more general IVT processes, the computational burden quickly becomes overwhelming in these scenarios, making estimation by classical maximum likelihood methods infeasible in practice. 
	
	The contributions of this paper can be summarized as follows. First, we derive the theoretical mixing properties of IVT processes. Using these, we prove consistency and, in the short memory case, asymptotic normality of the maximum composite likelihood (MCL) estimator of the parameter vector of an IVT model. We discuss the long memory case and, based on a result about the asymptotic behaviour of partial sums of IVT processes, conjecture that the MCL estimator has an $\alpha$-stable limit with infinite variance in this case.  For the purpose of conducting feasible inference and model selection, we propose two alternative estimators of the asymptotic variance of the MCL estimator in the short memory case: a kernel-based estimator, inspired by the heteroskedastic and autocorrelation consistent (HAC) estimator of \cite{NW1987}, and a simulation-based estimator.  Second, we use the same principle of considering the pairwise likelihood in lieu of the full likelihood, to derive the predictive distribution of an IVT model, conditional on the current value of the process; this allows us to use the IVT framework for forecasting integer-valued data. In a simulation study, we compare the MCL estimator to the standard method-of-moments-based estimator suggested in \cite{BNLSV2014} and find that the MCL estimator provides substantial improvements in most cases. Indeed, in a realistic simulation setup, we find that the MCL estimator can improve on the method-of-moments-based estimator by more than $50\%$, in terms of finite sample root median squared error. 
	Since the asymptotic theory for (G)MM estimation of trawl processes has not been worked out elsewhere, we also derive the asymptotic theory for GMM estimation and present the results for comparison purposes in the Supplementary Material, see Section \ref{sec:WeakDepGMM}.
	
	We apply the methods developed in the paper to a time series of the bid-ask spread of a financial asset. The time series behaviour of the bid-ask spread has been extensively studied in the literature on the theory of the microstructure of financial markets \citep[e.g.][]{HS1997,BSW2004}. The model selection procedure developed in the paper indicates that a model with Negative Binomial marginal distribution and slowly decaying autocorrelations most adequately describe the data. These findings are in line with those of \cite{GH2013}, who also found strong persistence in bid-ask spread time series. Then, in a pseudo-out-of-sample forecast exercise, we find that it is important to carefully model both the marginal distribution and the autocorrelation structure to get accurate forecasts of the future bid-ask spread. These findings highlight the strength of modelling using a framework where the choice of marginal distribution can be made independently of the choice of autocorrelation structure.

	The rest of the paper is structured as follows. Section \ref{sec.:setup} outlines the mathematical setup of IVT processes, while Section \ref{sec:estimation} contains details on the estimation and model selection procedures. Section \ref{sec:forec} presents the theory behind our proposed forecasting approach. Section \ref{sec:MC} summarises the results from our simulation study, investigating the finite sample properties of the estimation and model selection procedures. Section \ref{sec:emp} illustrates the use of the new methodology 
	in an empirical application to financial bid-ask spread data. Section \ref{sec.:concl} concludes.   The proofs of the main mathematical results are given in an Appendix.
	Practical details on the implementation of the asymptotic theory and additional derivations are given 
	in the Supplementary Material, which 
	also contains further simulation results and extensive details on various calculations used in the implementation of the methods.
	A software package for the implementation of simulation, estimation, inference, model selection, and forecasting of IVT processes is freely available in the MATLAB programming language.\if1\blind{\footnote{The software package can be found on GitHub.}} \fi \if0\blind{\footnote{The software package can be found at \url{https://github.com/mbennedsen/Likelihood-based-IVT}.}} \fi

	\section{Integer-valued trawl processes}\label{sec.:setup}
	
	Let $(\Omega, \mathcal{F}, \mathbb{P})$ denote a probability space, supporting a Poisson random measure $N$, defined on $\mathbb{Z} \times [0,1] \times \R$, with mean (intensity) measure $\eta \otimes Leb \otimes Leb$. Throughout $Leb$ denotes the Lebesgue measure and $\eta$ is a L\'evy measure. 
	A  \emph{L\'evy basis} $L$ is a homogeneous and independently scattered random measure on $[0,1] \times \R$, defined as
	\begin{align}\label{eq:LB0}
		L(dx,ds) := \int_{-\infty}^{\infty} y N(dy,dx,ds), \quad (x,s) \in [0,1] \times \R.
	\end{align}  
See, e.g., \cite{RR1989} and \cite{BN2011} for further details on L\'evy bases. Since we are only interested in integer-valued L\'evy bases, we will work under the following assumption.
	\begin{assumption}\label{ass:LB}
		The L\'evy basis $L$ is given by \eqref{eq:LB0} with L\'evy measure $\eta$, concentrated on the integers ($y \in \mathbb{Z}$), such that $\| \eta \| := \sum_{y = -\infty}^{\infty} y^2 \eta(y) < \infty$. 
	\end{assumption}
	The L\'evy basis $L$ is an infinitely divisible random measure with cumulant (log-characteristic) function 
	\[
	C_{L(dx,ds)}(\theta) := \log \mathbb{E}[ \exp (i\theta L(dx,ds)) ] = \int_{-\infty}^{\infty} \left(e^{i\theta y}-1\right) \eta(dy)dx ds, \quad (x,s) \in [0,1] \times \R. 
	\]
	An important random variable associated with the L\'evy basis $L,$ is the so-called \emph{L\'evy seed}, $L'$, which we define as 
	the random variable $L'$ satisfying $\mathbb{E}[\exp (i \theta L')] = \exp (C_{L'} (\theta)),$ with $C_{L'}(\theta)  = \sum_{y=-\infty}^{\infty} \left(e^{i\theta y}-1\right) \eta(y) $. 
	
	\begin{remark}\label{rem:Lseed}
		Because the distribution of a L\'evy process is entirely determined by its distribution at a particular time point, we can specify a L\'evy process $L_t'$ from a L\'evy seed $L'$, by requiring that $L_1' \sim L'$. 
	\end{remark}
	
	Using the L\'evy seed, we can rewrite the cumulant function of the L\'evy basis 
	as $C_{L(dx,ds)}(\theta) = C_{L'}(\theta) dx ds$, or, for a Borel set $B \in \mathcal{B}([0,1]\times \mathbb{R}),$
	\begin{align}\label{eq.:cumLB}
		C_{L(B)}(\theta) = C_{L'}(\theta) Leb(B).
	\end{align}
	From \eqref{eq.:cumLB} we have that $\kappa_j(L(B)) = \kappa_j(L') Leb(B)$, $j \geq 0$, where $\kappa_j(Z)$ denotes the $j$th cumulant of the random variable $Z$, when it exists.\footnote{Recall that the cumulants $\kappa_j(Z)$ of the random variable $Z$ are defined implicitly through the power series expansion of the cumulant function of $Z$, i.e., $C_Z(\theta) = \log \mathbb{E}[ \exp (i\theta Z) ]  = \sum_{j=1}^{\infty} \kappa_j(Z) (i \theta)^j/j!$.} In particular 
	$\mathbb{E}[L(B)] = \mathbb{E}[L'] Leb(B)$, and $Var(L(B)) = Var(L') Leb(B)$.
	The relationship \eqref{eq.:cumLB} implies that the distribution of the random variable $L(B)$ is entirely specified by the L\'evy seed $L'$ and the Lebesgue measure of the set $B$. In Section \ref{sec:marginal} below, we illustrate how this can be used to construct trawl processes with a given marginal distribution.

	The L\'evy basis $L$ acts on sets in $\mathcal{B}([0,1]\times \mathbb{R})$. We restrict attention to \emph{trawl sets} of the form
	\begin{align}\label{eq:At}
		A_t = A + (0,t), \quad A = \{(x,s) : s\leq 0, 0\leq x < d(s)\}, \quad t \geq 0,
	\end{align}
	where $d$
	is a \emph{trawl function},  determining the shape of the trawl set $A_t$. Section \ref{sec:corr} contains several parametric examples for the trawl function $d$. 
	We will impose the following assumption.
	\begin{assumption}\label{ass:trawl}
		The trawl set $A_t$ is given by \eqref{eq:At}, where the trawl function $d: \R_- \to [0,1]$ is continuous and monotonically increasing such that $Leb(A) = \int_{-\infty}^0 d(s) < \infty$. 
	\end{assumption}
	
	Intuitively, $A_t$ is obtained from the set $A$ by ``dragging'' it along in time. Note in particular that $Leb(A_t) = Leb(A)$ for all $t$. Finally, define the IVT process $X = (X_t)_{t \geq 0}$ as the L\'evy basis evaluated over the trawl set:
	\begin{align}\label{eq:defIVT}
		X_t := L(A_t), \quad t\geq 0.
	\end{align} 
	
	
	\subsection{Modelling the marginal distribution}\label{sec:marginal}
	
	For an IVT process $X$ as defined in \eqref{eq:defIVT}, we have 	$C_{X_t}(\theta) = C_{L(A_t)}(\theta) = Leb(A)C_{L'}(\theta) = C_{ L'_{Leb(A)}}(\theta)$, where  $L'_t$ is a L\'evy process with $L'_1 \sim L'$. Hence we observe 
	that the marginal distribution of the IVT process $X_t$ is entirely decided by the Lebesgue measure of the trawl set $A$ and the L\'evy seed $L'$ of the underlying L\'evy basis $L$. Indeed, by specifying a distribution for $L'$, we can build IVT processes with the corresponding marginal distribution. The following two examples illustrate how to do this; additional details can be found in the Supplementary Material. 
	
	\begin{example}[Poissonian L\'evy seed] \label{ex:Poisson}
		Let $L' \sim \textnormal{Poisson}(\nu)$, i.e.~$L'$ is distributed as a Poisson random variable with intensity $\nu>0$. It follows from standard properties of the Poisson distribution that $X_t  \sim \textnormal{Poisson}(\nu Leb(A))$. In other words, for all $t \geq 0$,
		$P\left( X_t = x \right) = (\nu Leb(A))^{x} e^{-\nu Leb(A)} / x!$,  $x = 0, 1, 2, \ldots$.
	\end{example}
	
	\begin{example}[Negative Binomial L\'evy seed]\label{ex:NB}
		Let $L' \sim \textnormal{NB}(m,p)$, i.e.~$L'$ is distributed as a Negative Binomial random variable with parameters $m>0$ and $p \in [0,1]$. It follows from standard properties of the Negative Binomial distribution that $X_t  \sim \textnormal{NB}(m Leb(A),p)$.   In other words, for all $t \geq 0$,
		$P(X_t = x) =  \frac{\Gamma(Leb(A)m+x)}{x!\Gamma(Leb(A)m)} (1-p)^{Leb(A)m} p^x$, $x= 0, 1, 2, \ldots$,
		where $\Gamma(z) = \int_0^{\infty} y^{z-1} e^{-y} dy$ for $z>0$ is the $\Gamma$-function.
	\end{example}

	\subsection{Modelling the correlation structure}\label{sec:corr}
	Recall that the shape of the trawl set $A_t$ is determined by the trawl function $d$, see Equation \eqref{eq:At}. A particularly tractable and flexible class of parametrically specified trawl functions are 
	the so-called \emph{superposition trawls} \citep{BNLSV2014,SY2017}. They are defined as
	$d(s) := \int_0^{\infty} e^{\lambda s} \pi (d\lambda)$, for $s \leq 0$,
	where $\pi$ is a probability measure on $\mathbb{R}_+.$ This construction essentially randomizes the decay parameter $\lambda$ in an otherwise exponential function. 
	
	The IVT process with a superposition trawl function is stationary. Hence, we get the autocorrelation function \citep{BNLSV2014}
	\begin{align}\label{eq.:corr}
		\rho(h) := Corr( L(A_{t+h}),L(A_t)) = \frac{ Leb( A_{h} \cap A)}{Leb(A)} = \frac{ \int_h^{\infty} d(-s) ds}{\int_0^{\infty} d(-s)ds},  \quad h>0.
	\end{align}
	
	\begin{example}[Exponential trawl function] \label{ex:Exp}
		For the case where the measure $\pi$ has an atom at $\lambda > 0,$  i.e.~$\pi(dx) =\delta_{\lambda}(dx),$ where $\delta_x(\cdot)$ is the Dirac delta function at $x \in \mathbb{R}_+,$ we get $d(s) = e^{\lambda s}$ for $s \leq 0$. Consequently, 
		$\rho(h) = 
		\exp(-\lambda h)$, for $h \geq 0$.
	\end{example}
	
	\begin{example}[Inverse Gaussian trawl function]\label{ex:IG}
		Letting $\pi$ be given by the inverse Gaussian distribution
		$\pi(dx) = \frac{(\gamma/\delta)^{1/2}}{2 K_{1/2}(\delta \gamma)} x^{-1/2} \exp \left( -\frac{1}{2} (\delta^2 x^{-1} + \gamma^2 x) \right) dx$, 
		where $K_{\nu}(\cdot)$ is the modified Bessel function of the third kind and  $\gamma, \delta \geq 0$ with both not zero simultaneously. It can be shown that the resulting trawl function is given by
		$d(s) =  \left( 1- \frac{2s}{\gamma^2}\right)^{-1/2} \exp\left( \delta \gamma \left( 1- \sqrt{ 1 - \frac{2s}{\gamma^2}} \right) \right)$, for $s \leq 0$, 
		and hence that the correlation function of the IVT process with inverse Gaussian trawl function becomes
		$\rho(h) = Corr(X_{t+h},X_t) = \exp\left(-\delta \gamma (\sqrt{1+2h/\gamma^2} - 1)\right)$, for $h \geq 0$.
		The details on these calculations can be found in the Supplementary Material.
	\end{example}
	
	\begin{example}[Gamma trawl function]\label{ex:GAM}
		Let $\pi$ have the $\Gamma(1+H,\alpha)$ density,
		$\pi(dx) = \frac{1}{\Gamma(1+H)} \alpha^{1+H} \lambda^{H} e^{-\lambda \alpha} dx$,
		where $\alpha > 0 $ and $H > 0.$ We can show that 
		$d(s) = \left( 1- \frac{s}{\alpha}\right)^{-(H+1)}$, $s \leq 0$,
		which implies the correlation function
		$\rho(h) = Corr( X_{t+h},X_t) = \frac{ Leb( A_{h} \cap A)}{Leb(A)} =  \left(1 + \frac{h}{\alpha}\right)^{-H}$.
		Note that in this case $\int_0^{\infty} \rho(h) dh=\infty$ for $H\in (0,1]$ and $\int_0^{\infty} \rho(h) dh=\alpha(H-1)^{-1}$ for $H>1$,
		from which we see that an IVT process with a Gamma trawl function enjoys the long memory property, in the sense of a non-integrable autocorrelation function, when  $H \in (0,1].$ The details on these calculations can be found in the  Supplementary Material.
	\end{example}
	
	\subsection{Modelling IVT processes}
	
	Using the above methods, we can build flexible continuous-time integer-valued processes with a marginal distribution determined by the underlying L\'evy basis, and independently specified correlation structure determined by the trawl function. In our main examples given above, we considered a L\'evy basis with  Poisson 
	or Negative Binomial 
	marginals, and various trawl functions, namely the Exponential trawl function, 
	the IG trawl function, 
	and the Gamma trawl function. 
	Other specifications for the underlying L\'evy basis and trawl function than those given here could of course be considered. In practice, these choices should be guided by the properties of the data being modelled. 
	
	The simplest IVT process we can construct in this way is the Poisson-Exponential IVT process, i.e., the case where $L' \sim \textnormal{Poisson}(\nu)$ and $d(s) = \exp(\lambda s)$, $s \leq 0$, see Examples \ref{ex:Poisson} and \ref{ex:Exp}. This special case results in a Markovian process, which is not in general true of IVT processes \citep[][]{BNLSV2014}. In fact, the model is similar to the popular Poissonian INAR(1) model, introduced in \cite{McKenzie1985} and \cite{AAOA1987}.  
	An illustration of the exponential trawl set, $A_t,$ dragged through time, together with a simulation of the resulting Poisson-Exponential IVT trawl process $X_t = L(A_t)$, is seen in Figure \ref{fig.:PRMtrawl}. The parameters used are $\lambda = 1$ and $\nu = 5$. At each time point $t$, the value of $X_t$ (bottom plot) is the number of points inside the trawl set $A_t$ (top plot).

	\begin{figure}[htb] 
		\centering 
		\includegraphics[scale=0.9]{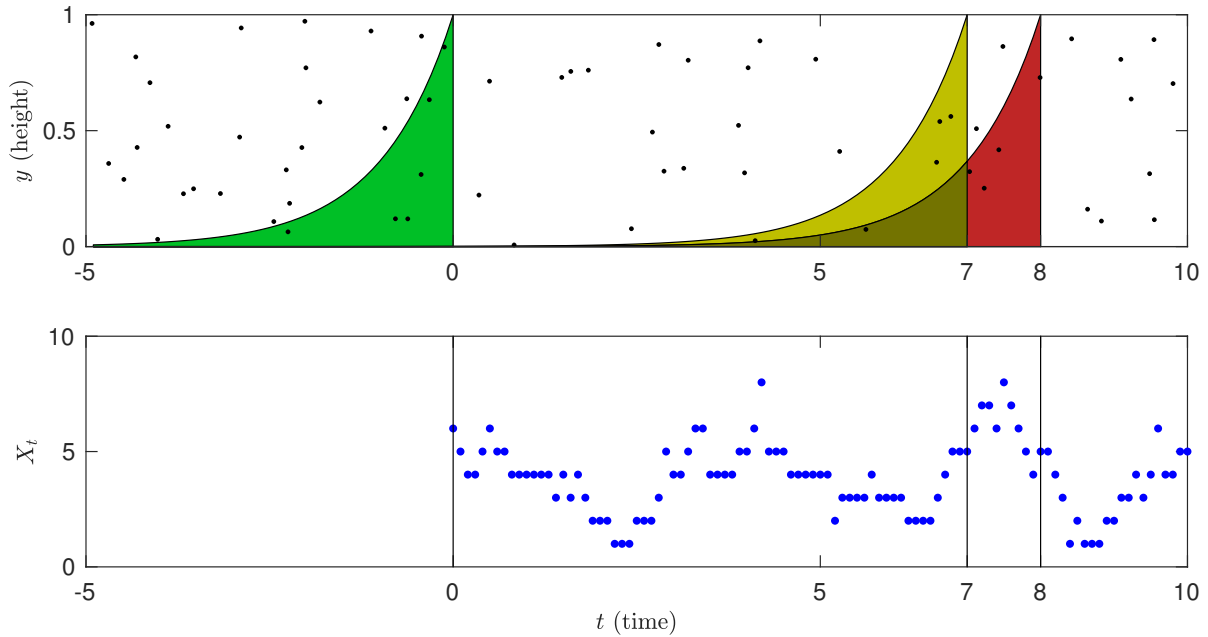}
		\caption{\it  Top: Simulation of a Poisson L\'evy basis on $\mathbb{R} \times  [0,1]$ (black dots) with an exponential trawl set $A_t$ (shaded) superimposed at three periods in time, $t \in \{0,7, 8\}$. Bottom: The associated trawl process $X_t = L(A_t)$, given by the number of `points' inside the trawl set $A_t$ at time $t$. The intensity of the Poisson random measure is $\|\eta\| = \eta(1) = \nu = 5$ and the parameter controlling the exponential trawl function, $d(s) = \exp(\lambda s)$, is $\lambda = 1.$}
		\label{fig.:PRMtrawl}
	\end{figure}



	\section{Estimation of integer-valued trawl processes}\label{sec:estimation}
	\cite{BNLSV2014} showed that the parameter vector $\theta$ of an IVT process can be consistently estimated using a generalized method of moments (GMM) procedure. 
	In Section \ref{sec:est_cst}, we propose a likelihood-based approach instead. 
	Both estimation procedures rely on the fact that the IVT process is stationary and mixing. The mixing property of IVT processes is obtained from results given in \cite{FS2013}, see \citet[][p. 699]{BNLSV2014}. Although mixing, in general, is sufficient for the consistency of the estimators, the central limit theorem for the likelihood-based estimator (Theorem \ref{th:CLT} below) relies on the stronger mixing concept of $\alpha$-mixing, where  the size (or rate) of mixing can also be established. 
	Let us  recall the definition of $\alpha$-mixing for a stationary process. Let $\mathcal{F}_{-\infty}^0 = \sigma( X_t; t\leq 0)$ and, for $m>0$, $\mathcal{F}_{m}^{\infty} = \sigma( X_t; t \geq m)$, and define the numbers
	$\alpha_m := \sup_{G \in \mathcal{F}_{-\infty}^0, H \in \mathcal{F}_{m}^{\infty}} | \mathbb{P}(H\cap G) - \mathbb{P}(H) \mathbb{P}(G)|$, for $m>0$.
	The process $X = (X_t)_{t\in \R}$ is $\alpha$-mixing if $\alpha_m \rightarrow 0$ as  $m \rightarrow  \infty$. It is $\alpha$-mixing of size $-\phi_0$ if $\alpha_m = O(m^{-\phi})$, as $m \rightarrow \infty$, for some $\phi>\phi_0$.
	
	We obtain the following important property for IVT processes. 
	\begin{theorem}\label{th:mixing}
		Let the IVT process  $X$ be given by \eqref{eq:defIVT} and let Assumptions \ref{ass:LB} and \ref{ass:trawl} hold. Now, $X$ is $\alpha$-mixing with $\alpha_m = O(\rho(m))$ as  $m \rightarrow  \infty$, where $\rho(m)$ is the autocorrelation function of $X$.
	\end{theorem}
	
	\begin{remark}\label{rem:mixing}
		The autocorrelation functions of the Exponential (Example \ref{ex:Exp}) and IG (Example \ref{ex:IG}) IVT models imply that these models are in fact $\alpha$-mixing with an exponential decay rate. The autocorrelation function of the Gamma (Example \ref{ex:GAM}) IVT model implies that it is $\alpha$-mixing of size $-(H-\epsilon)$ for all $\epsilon>0$.
	\end{remark}

\begin{remark}
	As an alternative to the proof of Theorem \ref{th:mixing} provided in Appendix \ref{app:proofs}, we could first show that trawl processes are $\theta$-weakly dependent, which we do in the Supplementary Material, see Section \ref{sec:WeakDepGMM}. Then, as pointed out in \citet[p.~324]{CuratoStelzer2019} and shown in the discrete-time case in \cite{Doukhanetal2012}, for integer-valued trawl processes, the fact that they are $\theta$-weakly dependent, implies that they are strongly mixing with the coefficient as stated in Theorem \ref{th:mixing}.
\end{remark}

	\subsection{Estimation by composite likelihoods}\label{sec:est_cst}
	Due to the non-Markovianity of the IVT process, we face computational difficulties when attempting to estimate the model by maximizing the full likelihood,
	hence we propose to use the CL method instead.  
	The main idea behind the CL approach is to consider a, possibly misspecified, 
likelihood function which captures the salient features of the data at hand; here this means capturing the features of the L\'evy basis, controlling the marginal distribution, and those of the trawl function, controlling the dependence structure. We focus on  
	pairwise CLs. 
	
	\subsubsection{Pairwise composite likelihood}\label{sec:pairwise}
	Suppose we have $n \in \N$ observations of the IVT process $X,$ $x_1, \ldots,x_n,$ on an equidistant grid of size $\Delta = T/n$, for some $T>0$. Define the following likelihood function using pairs of observations $k$ periods apart,
	\begin{align}\label{eq:CLh}
		CL^{(k)}(\theta;x) := \prod_{i=1}^{n-k} f(x_{i+k},x_{i};\theta), \quad k\geq 1,
	\end{align}
	where $f(x_{i+k},x_{i};\theta)$ is the joint probability mass function (PMF) of the observations $x_i$ and $x_{i+k}$, parametrized by the vector $\theta$. From \eqref{eq:CLh}, we construct  the composite likelihood function
	\begin{align}\label{eq:CLH}
		\mathcal{L}_{CL} (\theta;x) := \mathcal{L}_{CL}^{(K)} (\theta;x) := \prod_{k =1}^{K} CL^{(k)}(\theta;x) =  \prod_{k =1}^{K} \prod_{i=1}^{n-k} f(x_{i+k},x_{i};\theta),
	\end{align}
	where $K \in \N$ denotes the number of pairwise likelihoods to include in the calculation of the composite likelihood function. 
	
	%
	
	The maximum composite likelihood (MCL) estimator of $\theta$ is defined as
	\begin{align}\label{eq:clmax}
		\hat{\theta}^{CL}: = \arg \max_{\theta \in \Theta} l_{CL} (\theta;x),
	\end{align}
	where $\Theta$ is the parameter space and $l_{CL} (\theta;x) := \log \mathcal{L}_{CL} (\theta;x)$ is the log composite likelihood function. To apply this estimator in practice, we need to be able to calculate the PMFs $f(x_{i+k},x_i)$. Section \ref{app:f_derive} in the Supplementary Material contains a discussion on how to do this in the general integer-valued case. In the count-valued case, $f(x_{i+k},x_i)$ takes a particularly simple form which is convenient in implementations. Indeed, letting $\PP_{\theta}(B)$ denote the probability of the event $B$ given parameters $\theta$, we have the following.
	\begin{proposition}\label{prop:positiveLB}
		Let the IVT process  $X$ be given by \eqref{eq:defIVT} and let Assumptions \ref{ass:LB} and \ref{ass:trawl} hold. Suppose further, that the L\'evy basis $L$ is non-negative, i.e.~$\eta(y) = 0$ for $y<0$. The joint PMF of two observations $x_{i+k}$ and $x_i$ is
		\begin{align}
			f(x_{i+k},x_{i};\theta) = \sum_{c=0}^{ \min\{x_{i+k},x_i\}}  &\PP_{\theta}\left( L(A_{(i+k)\Delta} \setminus A_{i\Delta}) = x_{i+k} - c\right) \PP_{\theta}\left( L(A_{i\Delta} \setminus A_{(i+k)\Delta}) = x_{i}-c\right) \nonumber \\
			&\cdot \PP_{\theta}\left(L(A_{(i+k)\Delta} \cap A_{i\Delta}) = c\right). \label{eq:finsum}
		\end{align}
	\end{proposition}

	The probabilities $\PP_{\theta}(\cdot)$ in \eqref{eq:finsum} can be expressed as a function of the parameters of the L\'evy seed and the trawl function. Indeed, for a Borel set $B \in \mathcal{B}([0,1]\times \mathbb{R})$ we have
	$\PP_{\theta}(L(B) = x) = \PP_{\theta}(L'_{Leb(B)} = x)$, 
	where $L'_t$ is a L\'evy process with $L_1' \sim L'$, and $L'$ being the L\'evy seed associated to $X$, see Remark \ref{rem:Lseed}. 
	Also, $Leb( A_{(i+k)\Delta} \cap A_{i\Delta}) = \int_{-\infty}^{-k\Delta} d(s) ds$, 
	and
	$Leb( A_{(i+k)\Delta} \setminus A_{i\Delta}) = Leb( A_{i\Delta} \setminus A_{(i+k)\Delta})= Leb(A) - Leb( A_{(i+k)\Delta} \cap A_{i\Delta}) = \int_{-k \Delta}^0 d(s) ds$. 
	Plugging these into \eqref{eq:finsum} we obtain the pairwise likelihoods, $f(x_{i+k},x_i;\theta)$, and thus the CL function, $\mathcal{L}_{CL}(\theta;x)$, as a function of $\theta$. 
	
	\begin{example}[Poisson-Exponential IVT process]
		Let $L' \sim \textnormal{Poisson}(\nu)$ and $d(s) = \exp (\lambda s)$, $s \leq 0$, for some $\nu, \lambda > 0$. Since $L' \sim \textnormal{Poisson}(\nu)$ we have $L(B)  \sim \textnormal{Poisson}(Leb(B)\nu)$ for Borel sets $B$ and hence
		$\PP_{\theta}\left( L(B) = x \right) = (\nu Leb(B))^{x} e^{-\nu Leb(B)} / x!$, for $x \geq 0$.
		Further, it is not difficult to show that
		$Leb( A_{(i+k)\Delta} \cap A_{i\Delta}) = \lambda^{-1} e^{-\lambda k\Delta}$ and $Leb( A_{(i+k)\Delta} \setminus A_{i\Delta})  =\lambda^{-1}(1 - e^{-\lambda k \Delta})$.
		Using this, the probabilities in  \eqref{eq:finsum} can be expressed as functions of $\nu$ and $\lambda$ and hence the maximization \eqref{eq:clmax} can be carried out using standard numerical methods.
	\end{example}

	

	\subsubsection{Asymptotic theory}\label{sec:asym}
	Because we are only considering dependencies across pairs of observations and not their dependence with the remaining observations, the pairwise composite likelihood function  \eqref{eq:CLH} can be viewed as a misspecified likelihood. Nonetheless, since the individual PMFs $f(x_{i+k},x_i;\theta)$ in \eqref{eq:CLH} are proper bivariate PMFs, the \emph{composite score function} $\partial l_{CL}(\theta;x) / \partial \theta$ provides unbiased estimating equations and, under certain regularity assumptions, the usual asymptotic results will apply \citep{CN2004}. However, as pointed out in \cite{VRF2011}, formally proving the results in the time series case requires more rigorous treatment. The following two theorems provide the details on the asymptotic theory in the setup of this paper. 
	We will work under the following identification assumption.
	\begin{assumption}\label{ass:identification}
		For all $\theta \in \Theta$, it holds that
		\begin{align}\label{eq:identify}
			\theta \neq \theta_0  \Rightarrow  \sum_{k=1}^K f(x_1,x_2;\theta) \neq \sum_{k=1}^K f(x_1,x_2;\theta_0)
		\end{align}	
		for some $x_1,x_2\in \Z$.	 
	\end{assumption}
	First, we have a Law of Large Numbers.
	\begin{theorem}\label{th:LLN}
		Fix $K \in \N$, let the IVT process  $X$ be given by \eqref{eq:defIVT}, and let Assumptions \ref{ass:LB}--\ref{ass:trawl} and  \ref{ass:identification} hold. Then
		$\hat{\theta}_{CL}  \stackrel{\PP}{\rightarrow} \theta_0$, as $n\rightarrow \infty$.
	\end{theorem}
	\begin{remark}\label{rem:identify} 
		As is often the case, the identification condition in Assumption \ref{ass:identification} can be difficult to check in practice. 
		For the IVT processes considered in this paper and presented in the examples above, our numerical experiments indicate that requiring $K \geq \textnormal{dim}(\theta_d)$, where $ \textnormal{dim}(\theta_d)$ denotes the dimension of the parameters controlling the trawl function $d$, results in $\theta_0$ being identified. A similar requirement was suggested in \cite{DY2011}.
	\end{remark}
	
	We also impose a standard assumption on the parameter space.
	\begin{assumption}\label{ass:theta}
		The set $\Theta$ is compact such that the true parameter vector, $\theta_0$, lies in the interior of $\Theta$. 
	\end{assumption}
	
	It turns out that the asymptotic behaviour of the MCL estimator differs in the short- and long-memory cases. The former is captured by the following assumption.
	\begin{assumption}[Short memory]\label{ass:SM}
		The autocorrelation function of the IVT process satisfies $\lim_{n\rightarrow \infty}\rho(n)n=0$.
	\end{assumption}
	\begin{remark}
		Assumption \ref{ass:SM} is satisfied by IVT processes with the Exponential trawl (Example \ref{ex:Exp}), the Inverse Gaussian trawl (Example \ref{ex:IG}), and the Gamma trawl (Example \ref{ex:GAM}) with $H>1$.
	\end{remark}
	Under this assumption, the mixing property of IVT processes, presented in Theorem \ref{th:mixing}, implies that we can invoke a Central Limit Theorem for triangular arrays of mixing processes \cite[][Corollary 24.7]{Davidson1994} to get the following result.
	\begin{theorem}\label{th:CLT}
		Let the conditions from Theorem \ref{th:LLN} hold, together with Assumptions \ref{ass:theta}--\ref{ass:SM}. Then,
		\begin{align*}
			\sqrt{n} (\hat{\theta}^{CL} - \theta_0) \stackrel{d}{\rightarrow} N\left(0,G(\theta_0)^{-1}\right), \quad n\rightarrow \infty,
		\end{align*}
		where $G(\theta_0)$ is the \emph{Godambe information matrix} \citep{godambe60} matrix with inverse
		$G(\theta_0)^{-1} = H(\theta_0)^{-1} V(\theta_0) H(\theta_0)^{-1}$,
		where
		\begin{align*}
			H(\theta_0) =& -   \sum_{k =1}^K  \mathbb{E}\left[ \frac{\partial^2}{\partial \theta' \partial \theta}  \log f(X_{k\Delta},X_{0};\theta)|_{\theta = \theta_0} \right], \quad \mathrm{and}
			\\
			V(\theta_0) =& \sum_{k=1}^K Var \left( \frac{\partial}{\partial \theta}  \log f(X_{k\Delta},X_{0};\theta)|_{\theta = \theta_0} \right) \\
			&+ 2 \sum_{k=1}^K \sum_{k'=1}^{K}  \sum_{i=1}^{\infty} Cov\left( \frac{\partial}{\partial \theta} \log f(X_{k\Delta},X_{0};\theta)|_{\theta = \theta_0}, \frac{\partial}{\partial \theta'}\log f(X_{(i+k')\Delta},X_{i \Delta};\theta)|_{\theta = \theta_0} \right).
		\end{align*}
		Further, the infinite sum in the expression for $V(\theta_0)$ converges.
	\end{theorem}
	
	Theorem \ref{th:CLT} implies that feasible inference can be conducted using an estimate of the inverse of the Godambe information matrix
	$\hat G(\hat \theta^{CL})^{-1} = \hat H(\hat \theta^{CL})^{-1} \hat V(\hat \theta^{CL}) \hat H(\hat \theta^{CL})^{-1}$,
	where $\hat \theta^{CL}$ is the MCL estimate from \eqref{eq:clmax}. Note that while the straight-forward estimator $\hat H(\hat \theta^{CL} )= -n^{-1}\frac{\partial}{\partial \theta \partial\theta'}l_{CL}(\hat \theta^{CL};x)$ is consistent for $H(\theta)$ due to the stationarity and ergodicity of the IVT process, $\hat V(\hat \theta^{CL})$ is more difficult to obtain, since the obvious candidate $n^{-1} \frac{\partial}{\partial\theta} l_{CL}(\theta;x) \frac{\partial}{\partial\theta} l_{CL}(\theta;x)'$ vanishes at $\theta = \hat \theta^{CL}$, a fact also remarked in \cite{VV2005}. While it is possible to estimate $V(\theta_0)$ using a Newey-West-type kernel estimator \citep{NW1987}, we obtained more precise results using a simulation-based approach to estimating $V(\theta_0)$. The details of both approaches are provided in the Supplementary Material, Section \ref{app:stdErr}.\footnote{It is also possible to approximate the standard error of $\hat \theta^{CL}$ using a standard parametric bootstrap approach. However, as we discuss in Section \ref{app:B2} of the Supplementary Material, this solution is more computationally expensive than the two alternative approaches suggested here.}

	\subsubsection{Asymptotic theory in the long memory case}
	While the consistency result in Theorem \ref{th:LLN} applies for all IVT processes satisfying Assumptions \ref{ass:LB}--\ref{ass:trawl} and \ref{ass:identification}, Assumption \ref{ass:SM}, required in the central limit result in Theorem \ref{th:CLT}, excludes IVT processes with long memory, e.g. those with autocorrelation function adhering to $\rho(h) = O(h^{-H})$ for $H \in (0,1]$. As mentioned in Remark \ref{rem:mixing}, this is for instance the case for  the Gamma trawl function (Example \ref{ex:GAM}) with $H \in (0,1]$.

	Although a long memory CLT as such eludes us, we can say some things about the asymptotic behaviour of the MCL estimator $\hat{\theta}^{CL}$ in the long memory case. For instance, the convergence rate is likely slower than $\sqrt{n}$, as the following result suggests.
	\begin{theorem}\label{th:aH}
		Let the conditions from Theorem \ref{th:LLN} hold and assume that the autocorrelation function of the IVT process satisfies $\rho(h) = L_{\infty}(h) h^{-H}$ for some $H \in (0,1]$, where $L_{\infty}$ is a function which is slowly varying at infinity, i.e.~for all $a>0$ it holds that $\lim_{x \rightarrow \infty} \frac{ L_{\infty}(ax)}{ L_{\infty}(x)} = 1$. Then,
		\begin{enumerate}
			\item[(i)] For all $\epsilon > 0$,
			$n^{H/2-\epsilon} (\hat{\theta}^{CL} - \theta_0)  \stackrel{\PP}{\rightarrow} 0$, as $n\rightarrow \infty$.
			\item[(ii)] Let $J = \textnormal{dim}(\theta_0)$ be the dimension of $\theta_0$ and denote by $\hat{\theta}^{CL}_i$ and $\theta_{0,i}$ the $i$th component of the vectors $\hat{\theta}^{CL}$ and $\theta_0$, respectively. Then, for  $i = 1, 2, \ldots, J$, we have that for  all $\epsilon > 0$,
			$Var\left( n^{H/2 +\epsilon} (\hat{\theta}^{CL}_i - \theta_{0,i}) \right) \rightarrow \infty$, as 
			$n\rightarrow \infty$.
		\end{enumerate}
	\end{theorem}
	Theorem \ref{th:aH}(i) implies that the convergence rate of  $\hat{\theta}^{CL}$ cannot be slower than  $n^{H/2}$ for $H \in (0,1]$. Further, Theorem \ref{th:aH}(ii) implies that  if the convergence rate is faster than $n^{H/2}$ it must necessarily be the case that the limiting random variable has an infinite variance. We conjecture that $n^{H/(H+1)} (\hat{\theta}^{CL} - \theta_0)  \stackrel{(d)}{\rightarrow} M Y_{1+H}$ for a matrix $M$, where $Y_{\alpha}$ is an $\alpha$-stable random vector. Note that, for $H \in (0,1)$ it is the case $H/(1+H) \in (H/2,1/2)$, meaning that the conjectured convergence rate is faster than $n^{H/2}$, but slower that $\sqrt{n}$. Our reason for the conjecture has its roots in Theorem \ref{th:conjecture} below. First, we introduce a technical assumption on the trawl function $d$, ensuring that we are in the long memory case.
	\begin{assumption}[Long memory]\label{ass:LM}
		Assume that $H \in (0,1)$ and
		\begin{enumerate}
			\item
			$d(-x) = g_1(x) x^{-H-1}$, $x>0$, where $g_1$ is a function that is slowly varying at infinity.
			\item
			$d'(-x) = g_2(x) x^{-H-2}$, $x>0$, where $g_2$ is a function that is slowly varying at infinity.
		\end{enumerate}
	\end{assumption}
	\begin{remark}
		The key condition in Theorem \ref{th:aH}, namely  $\rho(h) = L_{\infty}(h) h^{-H}$ for some $H \in (0,1)$, is implied by Assumption \ref{ass:LM}.
	\end{remark}
	\begin{remark}
		The Gamma trawl function (Example \ref{ex:GAM}) fulfils Assumption \ref{ass:LM} with $g_1(x) = \left( x^{-1} + \alpha^{-1}\right)^{-H-1}$ and  $g_2(x) = \frac{H+1}{\alpha} \left( x^{-1} + \alpha^{-1}\right)^{-H-2}$.
	\end{remark}
	\begin{theorem}\label{th:conjecture}
		Suppose $L' \sim Poi(\nu)$ and that the parameters of the trawl function $d$ are known. Let the conditions from Theorem \ref{th:LLN} hold, along with Assumptions  \ref{ass:theta} and \ref{ass:LM}. 
		Then,
		\begin{align*}
			n^{H/(1+H)} (\hat{\nu}^{CL} - \nu_0 - R_n)  \stackrel{d}{\rightarrow} H(\nu_0)^{-1}\nu_0^{-1} Y_{1+H}, \quad n\rightarrow \infty,
		\end{align*}
		where $H(\nu_0)$ is given as in Theorem \ref{th:CLT},  $Y_\alpha$ is an $\alpha$-stable random variable with characteristic function
		\begin{align} \label{eq:chfct}
			\phi_{Y_\alpha}(u) := \mathbb{E}[\exp(i u Y_\alpha)] = \exp\left(c |u|^{\alpha}  \Gamma(2-\alpha)\ \left( \cos\left(\frac{\pi \alpha}{2}\right) - i \cdot \textnormal{sgn}(u) \sin\left(\frac{\pi \alpha}{2}\right) \right) \right), \quad u \in \R,
		\end{align}
		and where $R_n$ is given by
		\begin{align*}
			R_n = H(\nu_0)^{-1}\nu_0^{-1} n^{-1} S_n(U),
		\end{align*}
		with $S_n(U):= \sum_{i=1}^n \left(U_i -\mathbb{E}[U_i] \right)$ denoting the de-meaned partial sum of the sequence $U = \{U_i\}_{i=1}^n$, where $U_i := \sum_{k=1}^K g(X_{(i+k)\Delta},X_{i\Delta})$ and $g(X_{(i+k)\Delta},X_{i\Delta}) :=  \mathbb{E}[L(A_{(i+k)\Delta}\setminus A_{i\Delta})|X_{(i+k)\Delta},X_{i\Delta}]$.
	\end{theorem}
	\begin{remark}
		The asymptotic behaviour of the remainder term $R_n$ in Theorem \ref{th:conjecture} is decided by a quite general function $g$ of the pairs $(X_i,X_{i+k})$ and one can show that similar issues arise in the more general case where $L'$ is integer-valued and the parameters in the trawl function are estimated. The asymptotic behaviour of such general functions of the data could possibly be studied using mixing conditions for partial sums with $\alpha$-stable limits \citep[e.g.][]{Jakubowski1993} or by deriving Breuer-Major-like theorems  \citep[][]{BM1986,NPP2011} valid for IVT processes using Malliavin calculus for Poissonian spaces, see \cite{BPT2020} for a related approach. We believe that especially this latter route could be fruitful, but leave it for future work.
	\end{remark}
	
	The proof of Theorem \ref{th:conjecture} relies on a result about the partial sums of the IVT process $X$, which might be of independent interest. We, therefore, state it here.
	\begin{theorem}\label{th:partial}
		Suppose the L\'evy basis $L$ is non-negative, i.e.~$\eta(y) = 0$ for $y<0$, and let the conditions from Theorem \ref{th:LLN} hold, along with Assumptions  \ref{ass:theta} and \ref{ass:LM}.
		Then,
		\begin{align*}
			n^{-\frac{1}{1+H}} S_n(X)   \stackrel{d}{\rightarrow}  Y_{1+H}, \quad n\rightarrow \infty,
		\end{align*}
		where  $Y_\alpha$ is an $\alpha$-stable random variable with characteristic function \eqref{eq:chfct}.
	\end{theorem}
	
	\begin{remark}
		Closely related results about partial sums of trawl processes have previously been put forth in  \cite{DJLS2019} and \cite{PPSV2021} in the context of a discrete-time trawl process under a standard asymptotic scheme ($n\rightarrow \infty$) and in the context of a continuous-time trawl process under an infill asymptotic sampling scheme ($\delta \rightarrow 0$), respectively. In this paper, we consider $n$ observations of the continuous-time trawl process sampled on an equidistant $\delta$-grid, $X_{\delta}, X_{2\delta}, \ldots, X_{n\delta}$, where $\delta>0$ is fixed, and let $n \rightarrow \infty$. In this sense, our setup is closer to the one in  \cite{DJLS2019}. Indeed, it can be shown that the law of  $(X_{\delta}, X_{2\delta}, \ldots, X_{n\delta})$ is equal to the law of $(Y_1, Y_2, \ldots, Y_n)$, where $Y$ is an appropriately specified discrete-time trawl process in the sense of  \cite{DJLS2019}. Although this highlights a close connection between the long memory results of \cite{DJLS2019} and those presented in this section, the underlying assumptions in the two approaches are different. Firstly, the assumptions on the marginal distribution made in  \cite{DJLS2019} are different from ours. In particular, while  \cite{DJLS2019} are not restricting the marginal distribution of the process to be infinitely divisible, they do impose a restrictions on the size of the jumps of the L\'evy basis, see Equation (3.36) in  \cite{DJLS2019}. Secondly, our assumptions on the correlation structure of the process are slightly different than the assumptions made in   \cite{DJLS2019}. In particular, besides polynomial decay, we also allow for a slowly varying function to enter the correlation structure, compare Assumption \ref{ass:LM} with Equation (2.12) in   \cite{DJLS2019}. For these reasons, although our setting is closely related to that in   \cite{DJLS2019}, we cannot use their results directly. We can, however, follow similar lines of arguments as done in \cite{DJLS2019}, and this is what we do in the proof of Theorem \ref{th:partial}, given in the Appendix.		
	\end{remark}

	\subsection{Information criteria for model selection}\label{sec:IC}
	Takeuchi's Information Criterion \citep{Takeuchi1976} is an information criterion, which can be used for model selection in the case of misspecified likelihoods. \cite{VV2005} adapted the ideas of Takeuchi to the composite likelihood framework and provided arguments for using the composite likelihood information criterion (CLAIC)
	\[
	CLAIC = l_{LC}(\hat \theta^{CL};x) + \textnormal{tr}\left\{ \hat V(\hat \theta^{CL}) \hat H(\hat \theta^{CL})^{-1}  \right\}
	\]
	as a basis for model selection, where $\textnormal{tr}\{M\}$ is the trace of the matrix $M$. Specifically,  \cite{VV2005} suggest picking the model that maximizes $CLAIC$.
	
	Analogous to the usual Bayesian/Schwarz Information Criterion \citep[BIC,][]{BIC1978}, we also suggest the alternative composite likelihood information criterion \citep[][]{GS2010}
	\[
	CLBIC = l_{CL}(\hat \theta^{CL};x) + \frac{\log(n)}{2}  \textnormal{tr}\left\{ \hat V(\hat \theta^{CL}) \hat H(\hat \theta^{CL})^{-1}  \right\},
	\]
	where $n$ is the number of observations of the data series $x$. Note that the various models we consider are generally non-nested, whereas most research on model selection using the composite likelihood approach has considered nested model \citep[][]{NJ2014}. An analysis of the properties of $CLAIC$ and $CLBIC$ in the non-nested case in the spirit of, e.g., \cite{Vuong1989} would be very valuable but is beyond the scope of the present article.

	\section{Forecasting integer-valued trawl processes}\label{sec:forec}
	Let $\mathcal{F}_t = \sigma((X_s)_{s \leq t})$ be the sigma-algebra generated by the history of the IVT process $X$ up until time $t$ and let $h>0$ be a forecast horizon. 
	We are interested in the predictive distribution of the IVT process, i.e.~the distribution of $X_{t+h}|\mathcal{F}_t$. 
	However, since the IVT process $X$ is in general non-Markovian, the distribution of   $X_{t+h}|\mathcal{F}_t$ is highly intractable. This problem is similar to the one encountered when considering the likelihood of observations of $X$, cf. Section \ref{sec:est_cst}. For this reason, we propose to approximate the distribution of $X_{t+h}|\mathcal{F}_t$ by $X_{t+h}|X_t$, i.e.~instead of conditioning on the full information set, we only condition on the most recent observation. Thus, our proposed solution to the forecasting problem is akin to the proposed solution to the problem of the intractability of the full likelihood. That is, instead of considering the full distribution of $X_{t+h}|\mathcal{F}_t$, we use the conditional ``pairwise'' distribution implied by $X_{t+h}|X_t$.  
	
	To fix ideas, let $t \in \mathbb{R}$ and $h>0$, and consider the random variables $X_t=L(A_t)=L(A_t \cap A_{t+h})+L(A_t\setminus A_{t+h})$ and $X_{t+h}=L(A_{t+h})=L(A_t \cap A_{t+h})+L(A_{t+h}\setminus A_{t})$. The goal is to find the conditional distribution of $X_{t+h}$ given $X_t$. 
	Note  that $L(A_t \cap A_{t+h})$ and $L(A_{t+h}\setminus A_{t})$ are independent random variables. Further, since $L(A_{t+h}\setminus A_t)$ is independent of $X_t$ with a known distribution, we only need to determine the distribution of $L(A_t\cap A_{t+h})$ given $X_t$. The following lemma characterises the conditional distribution of $L(A_t\cap A_{t+h})$.
	
	\begin{lemma}\label{lem:fCond}
		Let $x \in \mathbb{N}\cup \{0\}$ and $l \in \{0, 1, \ldots, x\}$, then
		\begin{align*}
			&\PP(L(A_t \cap A_{t+h})=l|X_t=x)=\frac{\PP(L(A_t\setminus A_{t+h})=x-l)\PP(L(A_t \cap A_{t+h})=l)}{\PP(X_t=x)}.
		\end{align*}
	\end{lemma}
	
	\begin{example}\label{ex:f1} In the case when $L'\sim\mathrm{Poi}(\nu)$, we get the Binomial distribution:
		\begin{align*}
			L(A_t\cap A_{t+h})|X_t \sim \mathrm{Bin}\left(X_t,\frac{Leb(A_0\cap A_h)}{Leb(A_0)}\right),
		\end{align*}
		which implies that 
		$
		\E(L(A_t\cap A_{t+h})|X_t)=X_t Leb(A_0\cap A_h)(Leb(A_0))^{-1}.
		$
	\end{example}

	\begin{example}\label{ex:f2} In the case when $L'\sim\mathrm{NB}(m,p)$, we get the Dirichlet-multinomial distribution:
		\begin{align*}
			L(A_t\cap A_{t+h})|X_t\sim \mathrm{Dirichlet-multinomial}(X_t, \alpha_1, \alpha_2),
		\end{align*}
		where
		$\alpha_1=Leb(A_0\setminus A_h)m, \alpha_2=Leb(A_0\cap A_h)m)$ and $\alpha_1+\alpha_2=Leb(A_0)m$. For $x-l, x\in \{0, \ldots, X_t\}$, the corresponding probability mass function is given by 
		\begin{align*}
			&\PP(L(A_t\cap A_{t+h}) = l|X_t=x) \\
			&= {x \choose l} \frac{\Gamma(Leb(A_0\setminus A_h)m + x-l)}{\Gamma(Leb(A_0\setminus A_h)m)}  \frac{\Gamma(Leb(A_0 \cap A_h)m + l)}{\Gamma(Leb(A_0\cap A_h)m)} \frac{\Gamma(Leb(A_0)m)}{\Gamma(Leb(A_0)m + x)}, \quad x \geq l \geq 0,
		\end{align*}
		where $ {x \choose l}  = \frac{x!}{l!(x-l)!}$ is the binomial coefficient. This implies that, as before,  
		$
		\E(L(A_t\cap A_{t+h})|X_t)=X_t Leb(A_0\cap A_h)(Leb(A_0))^{-1}.
		$
	\end{example}
	
	Using Lemma \ref{lem:fCond}, we can derive the distribution of $X_{t+h}|X_t$, which can be used for probabilistic forecasting. The details for non-negative-valued L\'evy bases are given in the following proposition.
	
	\begin{proposition}\label{prop:fCond}
		Let the IVT process  $X$ be given by \eqref{eq:defIVT} and let Assumptions \ref{ass:LB} and \ref{ass:trawl} hold. Suppose further, that the L\'evy basis $L$ is non-negative, i.e.~$\eta(y) = 0$ for $y<0$.  Now, 
		\begin{align*}
			\PP(X_{t+h}=x_{t+h}|X_t=x_t)
			=\sum_{c=0}^{\min(x_t,x_{t+h})}
			\PP(L(A_{t+h}\setminus A_t)=x_{t+h}-c)
			\PP(L(A_t\cap A_{t+h})=c|X_t=x_t).
		\end{align*}
	\end{proposition}
	
	The following corollaries give the specific details for our two main specifications for the marginal distribution of $X_t$, studied in Examples \ref{ex:f1} and \ref{ex:f2} above.
	
	\begin{corollary}\label{cor:fpoi}
		If $L'\sim\mathrm{Poi}(\nu)$, then 
		\begin{multline*}
			\PP(X_{t+h}=x_{t+h}|X_t=x_t)
			\\
			=\sum_{c=0}^{\min(x_t,x_{t+h})} \frac{(\nu {Leb}(A_h\setminus A_0))^{x_{t+h}-c}}{(x_{t+h}-c)!}e^{-\nu {Leb}(A_h\setminus A_0)}
			{x_t \choose c} 
			\left(\frac{{Leb}(A_h \cap A_0)}{{Leb}(A_0)}\right)^c
			\left(1-\frac{{Leb}(A_h \cap A_0)}{{Leb}(A_0)}\right)^{x_t-c}.
		\end{multline*}
		
	\end{corollary}
	
	\begin{corollary}\label{cor:fnb}
		If $L'\sim\mathrm{NB}(m,p)$, then 
		\begin{align*}
			&\PP(X_{t+h}=x_{t+h}|X_t=x_t) \\
			&=\sum_{c=0}^{\min(x_t,x_{t+h})} (1-p)^{Leb(A_h\setminus A_0)m} p^{x_{t+h}-c}  {x_t \choose c} \frac{1}{(x_{t+h}-c)!} \\
			&\cdot \frac{\Gamma(Leb(A_h\setminus A_0)m + x_{t+h}-c)}{ \Gamma(Leb(A_h\setminus A_0)m)}
			\frac{\Gamma(Leb(A_h\setminus A_0)m + x_t-c)}{\Gamma(Leb(A_h\setminus A_0)m)}  \frac{\Gamma(Leb(A_h \cap A_0)m + c)}{\Gamma(Leb(A_h\cap A_0)m)} \frac{\Gamma(Leb(A_0)m)}{\Gamma(Leb(A_0)m + x_t)}.
		\end{align*}
		
	\end{corollary}
	
	If the parameters of an IVT process $X$ with Poisson or Negative Binomial marginal distribution are known, we can use Corollary \ref{cor:fpoi} or \ref{cor:fnb}, and the calculations for the Lebesgue measures of the trawl sets given in Section \ref{sec:corr}, 
	for computing the predictive PMFs and thus for forecasting. When the true parameter values are unknown, they can be estimated using the MCL estimator suggested above, and plugged into the formulas to arrive at estimates of the  predictive PMFs.

	\section{Monte Carlo simulation experiments}\label{sec:MC}
	Using simulations, we examine the finite sample properties of the composite likelihood-based estimation procedure
	and of the model selection procedure. Details are available in the Supplementary Material, see Section \ref{sec:SuppSim}. Here we summarise our main findings. 
	
	We consider six data-generating processes (DGPs): 
	the Poisson-Exponential (P-Exp), the Poisson-Inverse Gaussian (P-IG), the Poisson-Gamma (P-Gamma), the Negative Binomial-Exponential (NB-Exp), the Negative Binomial-Inverse Gaussian (NB-IG), and the Negative Binomial-Gamma (NB-Gamma)  IVT models. 
	The parameter choices used in the simulation study, see Table  \ref{tab:paramTab} in the Supplementary Material,  are motivated by the estimates of our empirical study. 

	We compare the finite sample properties of the MCL estimator with the GMM estimator, which has been used in the existing literature. 
	Figure \ref{fig:CLvsMM} plots the root median squared error (RMSE) of the MCL estimator of a given parameter divided by the RMSE of the GMM estimator of the same parameter for the six DGPs. 
	Thus, numbers smaller than one indicate that the MCL estimator has a lower RMSE than the GMM estimator and vice versa for numbers larger than one. We see that for most parameters in most of the DGPs, the MCL estimator outperforms the GMM estimator substantially; indeed, in many cases, the RMSE of the MCL estimator is around $50\%$ that of the GMM estimator. The exception seems to be the trawl parameters, i.e.~the parameters controlling the autocorrelation structure, in the case of the Gamma and IG trawls, where the GMM estimator occasionally performs on par with the MCL estimator. However, in most cases, it appears that the MCL estimator is able to provide large improvements over the GMM estimator.
	
	In the Supplementary Material (Section \ref{sec:fs_asym}), we also examine how close the finite sample distribution of the MCL estimator is to the true (Gaussian) asymptotic limit, as presented in Theorem \ref{th:CLT}. We find that the Gaussian approximation is very good for the case of the parameters governing the marginal distribution, as well as for the parameter $\lambda$ for the case of IVTs with exponential trawl functions. When there are two parameters in the trawl function ($\delta$ and $\gamma$ in the case of the IG trawl and $H$ and $\alpha$ in the case of the Gamma trawl), however, the Gaussian distribution can be a poor approximation to the finite sample distribution of the MCL estimator in case of the trawl parameters. This indicates that for constructing confidence intervals or testing hypotheses on these parameters, it might be useful to consider bootstrap approaches instead of relying on the Gaussian distribution.
	


	\begin{figure}[!t]
		\centering
		\includegraphics[scale=0.75]{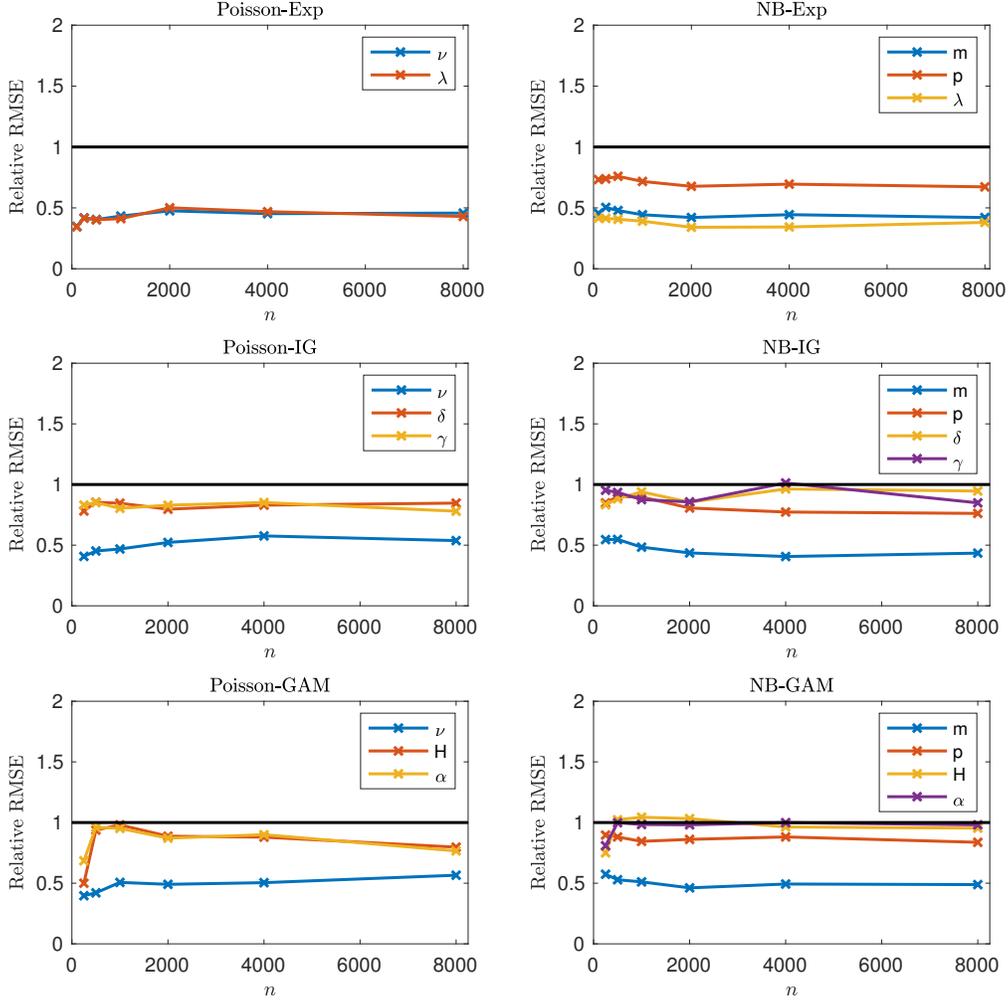} 
		\caption{\it  Root median square error (RMSE) of the MCL estimator \eqref{eq:clmax} divided by the RMSE of the GMM estimator. The underlying IVT process $X_t$ is simulated on the grid $t = \Delta, 2\Delta, \ldots, n\Delta$, with $\Delta = 0.10$, see Table \ref{tab:paramTab} for  the values of the parameters used in the simulations. For the Poisson-Exp and NB-Exp models we set $K = 1$ in Equations \eqref{eq:MM} and \eqref{eq:clmax}; for the other models we set $K = 10$.}
		\label{fig:CLvsMM}
	\end{figure}

	\section{Application to financial bid-ask spread data}\label{sec:emp}

	In this section, we apply the IVT modelling framework to the bid-ask spread of equity prices. The bid-ask spread has been extensively studied in the market microstructure literature, see, e.g., \cite{HS1997} and \cite{BSW2004}. An application similar to the one studied in this section was considered in \cite{BNLSV2014}. 
	To illustrate the use of the methods proposed in this paper, we study the time series of the bid-ask spread, measured in U.S. dollar cents, of the Agilent Technologies Inc. stock (ticker: A) on a single day, May $4$, $2020$. 
	We cleaned the data and sampled the data in 5s intervals, leading to $n = 3961$ observations, see Section \ref{sec:SuppEmp} in the Supplementary Material for details.

	\begin{figure}[!t]
		\centering
		\includegraphics[scale=0.75]{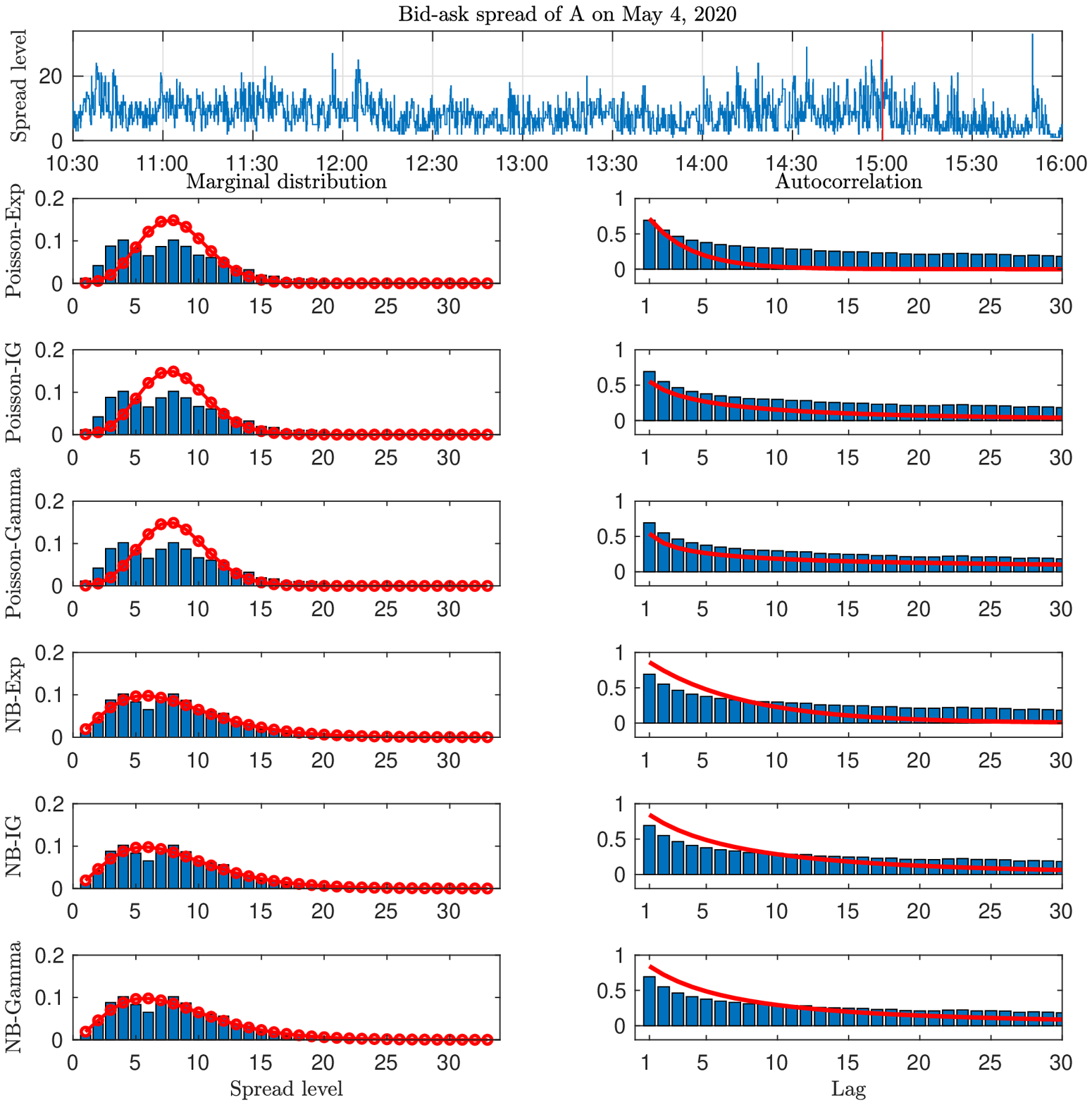} 
		\caption{\it Analysis of the A spread level on May 4, 2020. Top: The data (spread level in cents) from $10$:$30$AM to $4$:$00$PM sampled every $5$ seconds. The vertical red line separates the initial in-sample period (left) from the out-of-sample period (right), used in the forecasting exercise in Section \ref{sec:fSpr}. The bottom six rows show the empirical autocorrelation (left; blue bars) and the empirical marginal distribution of the spread level (right; blue bars) together with the fits from the six IVT models (red lines). The parameters of the models have been estimated using the MCL estimator \eqref{eq:clmax} with $K = 10$, see Table \ref{tab:paramEmp}.}
		\label{fig:nci1}
	\end{figure}
	
	Let $s_t$ be the  bid-ask spread level at time $t$, the time series of which is displayed in the top panel of Figure \ref{fig:nci1}. Since the minimum spread level in the data is one tick (one dollar cent), we work on this time series minus one, i.e on $x_t = s_t-1$. 
	We can now apply  our model selection method. 
	We first inspect the empirical autocorrelation of the data (shown in the right panels of Figure \ref{fig:nci1}) which shows evidence of a very persistent process; we, therefore, set $K$ to a moderately large value to accurately capture the dependence structure of the data. Here, we choose $K=10$ but the results are robust to other choices. 
	

	\begin{table}
\caption{\it Model selection results: bid-ask spread data}
\vspace{-0.8cm}
\begin{center}
\footnotesize
\begin{tabularx}{1.00\textwidth}{@{\extracolsep{\stretch{1}}}lcccccc@{}} 
\toprule
Model: & Poisson-Exponential & Poisson-IG & Poisson-Gamma & NB-Exponential & NB-IG & NB-Gamma  \\
\midrule
$CL$ &   $-244125.5 $ & $ -242885.2 $ & $-242835.8 $ & $-216363.9 $ & $-216318.1 $ & ${\bf -216313.5}$  \\
$CLAIC$ & $ -244125.7 $ & $-242899.1 $ & $-242855.8 $ & $-216364.0  $ & $ -216318.5$ & ${\bf -216313.9}$ \\ 
$CLBIC$ &  $-244126.4 $ & $-242942.8$ & $ -242918.8 $ & $-216364.1$ & $ -216319.7 $ & ${\bf-216315.1}$ \\
\bottomrule 
\end{tabularx}
\end{center}
{\footnotesize \it Composite likelihood and information criteria values for fitting the A bid-ask spread data on May 4, 2020, shown in the top plot of Figure \ref{fig:nci1}, calculated using six different models as given in the top row of the table using $K=10$. The maximum value for a given criteria (i.e. row-wise) is given in bold. The parameter estimates corresponding to the fits are given in Table \ref{tab:paramEmp}.} 
\label{tab:nci_ic}
\end{table}

	\begin{table}
\caption{\it Estimation results: bid-ask spread data}
\vspace{-0.8cm}
\begin{center}
\footnotesize
\begin{tabularx}{1.00\textwidth}{@{\extracolsep{\stretch{1}}}lcccccccc@{}} 
\toprule
DGP        &  $\nu$  & $m$           & $p$      & $\lambda$ & $\delta$ & $\gamma$ & $H$ & $\alpha$ \\
\cmidrule{2-9}
P-Exp      & $28.9319$      &                  &              & $4.0399$            &               & & & \\
  & $(0.6644)$   &                  &              & $(0.0904)$            &               & & & \\
P-IG        & $ 294.9102$      &                  &              &                   &  $1.5293$  & $0.0371$ & & \\
        & $(28.2715)$      &                  &              &                   &  $(0.0424)$  & $(0.0040)$ & & \\
P-Gamma  & $83.8197$      &                  &              &              & &     &  $0.6123$  & $0.0523$ \\
  & $(-)$      &                  &              &              & &     &  $(-)$  & $(-)$ \\
NB-Exp   &            & $6.4273$   & $0.6665$  & $1.7835$            &               & & & \\
   &            & $(0.9324)$   & $(0.0215)$  & $(0.1349)$            &               & & & \\
NB-IG     &            &   $7.7104$  & $0.6675$  &                   &  $1.7816$ & $0.8292$ & & \\
     &            &   $(1.0111)$  & $(0.0238)$  &                   &  $(0.4436)$ & $(0.2094)$ & & \\
NB-Gamma     &            &   $7.7336$  & $0.6675$  &          & &         &  $1.7020$ & $0.7897$ \\
     &            &   $(1.1316)$  & $(0.0260)$  &          & &         &  $(0.7365)$ & $(0.3363)$ \\
\bottomrule 
\end{tabularx}
\end{center}
{\footnotesize \it Parameter estimates (standard errors in parentheses) from the six different DGPs when applied to the bid-ask spread data of A on May 4, 2020 using the MCL estimator with $K = 10$. The standard errors have been obtained using the simulation-based approach to estimating the asymptotic covariance matrix of the MCL estimator, see Section \ref{app:stdErr} in the Supplementary Material. Since our asymptotic theory does not cover the long memory case, no standard deviations are reported for the P-Gamma model. See Figure \ref{fig:nci1} for the resulting fits of the models to the empirical distribution and autocorrelation.}
\label{tab:paramEmp}
\end{table}
	
	Using this setting, we calculated the maximized composite likelihood value, CL, and the two information criteria, $CLAIC$ and $CLBIC$, obtained for these data using the six models considered in Section \ref{sec:MC}. The results are shown in Table \ref{tab:nci_ic}. The table shows that the NB-Gamma model is the preferred model on all three criteria, while the second-best model is the NB-IG model. 
	
	To further examine the fit of the various models, the bottom six rows of Figure \ref{fig:nci1} contain the empirical autocorrelation (left; blue bars) and the empirical marginal distribution of the spread level (right; blue bars). Each respective row also shows the fit of one of the six models considered in Table \ref{tab:nci_ic}; the parameter estimates corresponding to the models are given in Table \ref{tab:paramEmp}.  The fit of the models shown in the  bottom six panels of Figure \ref{fig:nci1} and the selection criteria of Table  \ref{tab:nci_ic} indicate that the models based on the Negative Binomial distribution are preferred to the models based on the Poisson distribution. We conclude that, for this data series, the Poisson distribution is unable to accurately describe the marginal distribution of the spread level sampled every $5$ seconds. That the Gamma and IG trawl functions are preferred to the Exponential trawl function indicates that the Exponential autocorrelation function is not flexible enough to capture the correlation structure of the data. By both visual inspection of the autocorrelations in  Figure \ref{fig:nci1} and the selection criteria of Table  \ref{tab:nci_ic}, we conclude that the NB-Gamma model is the preferred model for these data. As shown in  Table \ref{tab:paramEmp}, this model has $\hat H = 1.70$ (s.e. $0.74$), implying that the model possesses a slowly decaying autocorrelation structure, albeit not the long memory property.



	\subsection{Forecasting the spread level}\label{sec:fSpr}
	This section illustrates the use of IVT models for forecasting, as outlined in Section \ref{sec:forec}. The aim is to forecast the future spread level of the A stock on May $4$, i.e.~the data studied above and plotted in the top panel of Figure \ref{fig:nci1}. We set aside the first $n_1 = 3221$ observations as an ``in-sample period'' for initial estimation of the parameters of the models, see the vertical red line in the top panel of Figure \ref{fig:nci1} for the placement of this split. We then forecast the spread level from $5$ seconds until $100$ seconds into the future, using the approach presented in Section \ref{sec:forec}. That is, we forecast $x_{n_1+1}, x_{n_1+2}, \ldots, x_{n_1+20}$ given the current value $x_{n_1}$. After this, we update the in-sample data set with one additional observation so that this sample now contains $n_2 = n_1 +1 = 3222$ observations. Then we again forecast the next $20$ observations, $x_{n_2+1}, x_{n_2+2}, \ldots, x_{n_2+20}$, given $x_{n_2}$. We repeat this procedure until the end of the sample, which yields $n_{oos} = n - n_1 - 20 = 720$ out-of-sample forecasts for each forecast horizon. To ease the computational burden, we only re-estimate the model every $24$ periods (i.e.~every $2$ minutes). 
	
	To evaluate the forecasts, we consider four different loss metrics. The first two, the mean absolute error (MAE) and the mean squared error (MSE), are often used in econometric forecasting studies of real-valued data \citep[e.g.][]{ET2016}. For a forecast horizon $h = 1, 2, \ldots, 20$, define the mean absolute forecast error,
	\[
	MAE(h) =  \frac{1}{n_{oos}} \sum_{i=n_1+h}^{n-(20-h)} | x_{i} - \hat{x}_{i|i-h}|,
	\]
	where $\hat{x}_{i|i-h}$ is the $h$-step ahead forecast of $x_i$, constructed using the information available up to observation $i-h$. That is, $\hat{x}_{i|i-h}$ is the point forecast of $x_i$ coming from a particular IVT model, such as the conditional mean, median, or mode. In what follows, we set $\hat x_{i|i-h}$ equal to the estimated conditional mean, i.e.~we set $\hat x_{i|i-h} = \sum_{k=0}^{M}  \widehat \PP(X_{i|i-h} = k) k$, where $M\geq 1$ is a large (cut-off) number and $\widehat \PP$ is the estimated predictive density of the IVT model.\footnote{Letting $\hat x_{i|i-h}$ be the conditional mode, instead of the conditional mean, produces results similar to those reported here. These results are reported in the Supplementary Material, Section \ref{sec:add_forec}.} Here, we set $M=60$ but the results are very robust to other choices. 
	Define also the mean squared forecast error
	\[
	MSE(h) = \frac{1}{n_{oos}} \sum_{i=n_1+h}^{n-(20-h)} ( x_{i} - \hat{x}_{i|i-h})^2.
	\]
	We consider two additional loss metrics, designed to directly evaluate the accuracy of the estimated predictive PMF $\widehat \PP$, which is arguably more relevant to the problem at hand than MAE and MSE. The first is the logarithmic score \citep[][p. 30]{ET2016},
	\[
	logS(h) =  \frac{1}{n_{oos}} \sum_{i=n_1+h}^{n-(20-h)}  -\log \widehat \PP(X_{i|i-h} = x_i),
	\]
	where $x_i$ is the realized outcome. The second is the ranked probability score \citep[RPS;][]{Epstein1969},
	\[
	RPS(h) =  \frac{1}{n_{oos}} \sum_{i=n_1+h}^{n-(20-h)}  \sum_{k=0}^M (\widehat F_{i|i-h}(k) - \mathbb{I}_{\{x_i \leq k\}})^2,
	\]
	where $\widehat F_{i|i-h}(k) = \sum_{j=0}^k \widehat \PP(X_{i|i-h} = j)$ is the estimated cumulative distribution function  of $X_i|x_{i-h}$  coming from a given model and $\mathbb{I}_{\{x_i \leq k\}}$ is the indicator function of the event $\{x_i \leq k\}$.

	Figure \ref{fig:ncif_all} shows the four different forecast loss metrics for the preferred NB-Gamma IVT model as a ratio of the forecasting loss of a given benchmark model in the out-of-sample forecasting exercise described above. The numbers plotted in the figure are $Loss(h)_{NB-Gamma}/Loss(h)_{benchmark}$, where ``$Loss$'' denotes one of the four loss metrics given above and $h = 1, 2, \ldots, 20$ denotes the forecasting horizon. Thus, numbers less than one favour the NB-Gamma model compared to the benchmark model and vice versa for numbers greater than one. Initially, we choose  the Poisson-Exponential IVT process as the benchmark model (Figure \ref{fig:ncif_all}, first column); as remarked above, this process is identical to the Poissonian INAR(1) model, which is often used for forecasting count-valued data \citep[e.g.][]{FM2004,MM2005,SPS2009}. It is evident from the figure that losses from the NB-Gamma model are smaller than those from the Poisson-Exponential model for practically all forecast horizons and loss metrics. 
	In the case of the  two most relevant loss functions for evaluating the predictive distribution, the logS and RPS, the reduction in losses are substantial for all forecast horizon, on the order of $20\%$. 
	
	\begin{figure}[!t] 
		\centering
		\includegraphics[scale=0.95]{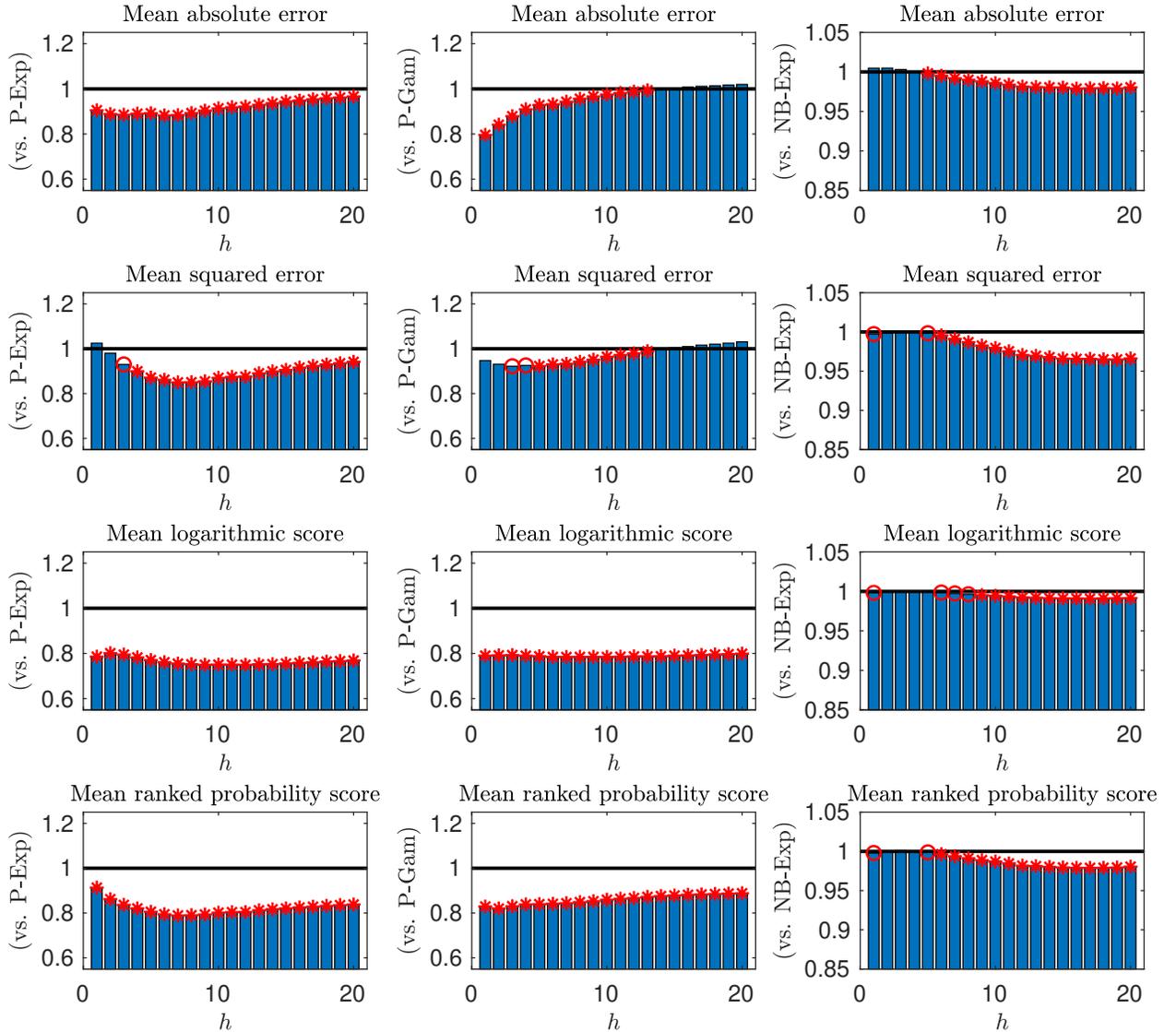} 
		\caption{\it Forecasting the spread level of the A stock on May 4, 2020. Four different loss metrics (row-wise, as indicated above each plot) and twenty forecast horizons, $h = 1, 2, \ldots, 20$. The numbers plotted are relative average losses of the NB-Gamma forecasting model, 
			compared with the Poisson-Exponential model (first column), the Poisson-Gamma model (second column), and the NB-Exponential model (third column), over $n_{oos} = 720$ out-of-sample forecasts. 
			A circle above the bars indicates rejection null of equal forecasting performance between the two models, against the alternative that the NB-Gamma model provides superior forecasts, using the Diebold-Mariano \citep{DM1995} test at a $5\%$ level; an asterisk denotes rejection at a $1\%$ level.}
		\label{fig:ncif_all}
	\end{figure}

	To assess whether these loss differences are also statistically significant, we perform the Diebold-Mariano test of superior predictive ability \citep[][]{DM1995}. The null hypothesis of the statistical test is that the two models have equal predictive power, while the alternative hypothesis is that the NB-Gamma model provides superior forecasts compared to the benchmark model. In Figure \ref{fig:ncif_all}, a circle (asterisk) denotes rejection of the Diebold-Mariano test at a $5\%$ ($1\%$) level. The test rejects the null hypothesis of equal predictive ability for almost all forecast horizons and loss metrics at a $1\%$ level. 
	
	To investigate whether the increased forecast performance of the NB-Gamma model comes from having a more flexible marginal distribution than the Poisson-Exponential benchmark model (Negative Binomial vs. Poisson marginals) or from having a more flexible correlation structure (polynomial decay vs. exponential decay) or both, we compare the forecasts from NB-Gamma model to those coming from a Poisson-Gamma model and from an NB-Exp model. These results are given in the second and third columns of Figure \ref{fig:ncif_all}, respectively. From  the second column, we see that the NB-Gamma model outperforms the Poisson-Gamma model considerably, especially for the shorter forecast horizons, indicating that it is important to use a model with Negative Binomial marginals when forecasting these data. 
	From  the third column of the figure, we see that for the shorter forecast horizons, the NB-Exp model performs on par with the NB-Gamma model, but for the longer forecast horizons, the NB-Gamma model is superior. Hence, when forecasting, it also appears to be important to specify a model with an accurate autocorrelation structure, especially for longer forecasting horizons.

	\section{Conclusion}\label{sec.:concl}
	
	This paper has developed likelihood-based methods for the estimation, inference, model selection, and forecasting of IVT processes. 
	We proved the consistency and asymptotic normality of the MCL estimator and provided the details on how to conduct feasible inference and model selection. We also developed a pairwise approach to approximating the conditional predictive PMF of the IVT process, which can be used for forecasting integer-valued data. All these methods are implemented in a freely available software package written in the MATLAB programming language.
	
	In a simulation exercise, we demonstrated the good properties of the MCL estimator compared to the often-used method-of-moments-based estimator. Indeed, the reduction in root median squared error of the MCL estimator was in many cases more than $50\%$ compared to the corresponding GMM estimator. 
	
	In an empirical application to financial bid-ask spread data, we illustrated the model selection procedure and found that the Negative Binomial-Gamma IVT model provided the best fit for the data.  Using the forecast tools developed in the paper, we saw that this model outperformed the simpler Poisson-Exponential IVT model considerably, resulting in a reduction in forecast loss on the order of $20\%$  for most forecast horizons. We demonstrated that most of the superior forecasting performance came from accurate modelling of the marginal distribution of the data; however, we also found that it was beneficial to carefully model the autocorrelation structure, especially for longer forecasting horizons. These findings highlight the strengths of the IVT modelling framework, where the marginal distribution and autocorrelation structure can be modelled independently in a 
	flexible fashion. 

	\bibliographystyle{agsm}
	\bibliography{mybib-v7}
	

	\appendix
	
	\section{Mathematical proofs}\label{app:proofs}
	
	We first give an alternative representation of the L\'evy basis $L$, underlying the IVT process X. From the  construction of the IVT process in Section \ref{sec.:setup} in the main article, it is clear that the distribution of $L$ is representable as a compound Poisson distribution. That is, for a Borel set $B$, we can write 
	\begin{align}
		\PP_{\theta}(L(B) = x) &= \sum_{q=0}^{\infty} \PP_{\theta}\left( \sum_{i=1}^q Y_i = x, \tilde N(B) = q\right) \nonumber \\
		&= \sum_{q=0}^{\infty} \PP_{\theta}\left( \sum_{i=1}^q Y_i = x\right) \PP_{\theta}( \tilde N(B) = q), \quad x \in \Z, \label{eq:Lalt}
	\end{align}
	where $Y_i$ are iid integer-valued random variables with probability mass function $\tilde \eta(y):= \frac{\eta(y)}{ \sum_{y=-\infty}^{\infty} \eta(y)}$, i.e.~$\PP_{\theta}(Y_i = y) = \tilde \eta(y)$, where $\eta$ is the L\'evy measure given in the construction of the IVT process in Equation. 
	Likewise, $\tilde N$ is a Poisson random measure, given by
	\begin{align*}
		\tilde N(dx,ds) = \int_{-\infty}^{\infty} N(dy,dx,ds),
	\end{align*}
	with an underlying intensity $\tilde \nu := \sum_{y=-\infty}^{\infty} \eta(y)$. The random variables $Y_i$ are independent of the random measure $\tilde N$. Intuitively, we have decomposed the event that the sum of the points in the set $B$ equals $x$ (i.e.~$\{L(B) = x\}$) into the intersection of the two events that there are $q$ individual points in $B$ (i.e.~$\{\tilde N(B) = q\}$) and the ``sizes'' of these $q$ points add up to $x$ (i.e.~$\{ \sum_{i=1}^q Y_i = x \}$). With this construction, we have
	\begin{align}\label{eq:NBq}
		\PP_{\theta}(\tilde N(B) = q) = \frac{(\tilde \nu Leb(B))^q}{q!} e^{-\tilde \nu Leb(B)}, \quad q=0,1,2,\ldots.
	\end{align}
	We will use this alternative representation of $L$ in our proofs below.

	\begin{proof}[Proof of Theorem \ref{th:mixing}]
		Let $m>0$ and define $N_m = \tilde N (A_m \cap A)$ as the Poisson random variable, which counts the number of `events' in the set $A_m \cap A_0$. From \eqref{eq:NBq}, we know that there exists a constant $\tilde \nu >0$ such that 
		\begin{align*}
			\mathbb{P}(N_m = 0) = e^{-\tilde{\nu} Leb(A_m \cap A_0)}
		\end{align*}
		and, therefore,
		\begin{align}\label{eq:Nm}
			\mathbb{P}(N_m \neq 0) = 1- \mathbb{P}(N_m = 0) =1- e^{-\tilde{\nu} Leb(A_m \cap A_0)} = \tilde{\nu} Leb(A_m \cap A_0) + o(Leb(A_m \cap A_0)),
		\end{align}
		as $m \rightarrow \infty$. 
		
		Let $G \in \mathcal{F}_{-\infty}^0$ and  $H \in \mathcal{F}_{m}^{\infty}$ be such that $\mathbb{P}(G),\mathbb{P}(H)>0$, and write, using the law of total probability,
		\begin{align*}
			|\mathbb{P}(H \cap G) - \mathbb{P}(H)\mathbb{P}(G)| &= |\mathbb{P}(H|G) - \mathbb{P}(H)| \cdot \mathbb{P}(G) \\
			&=  |\mathbb{P}(H|G,N_m=0) \mathbb{P}(N_m=0|G) + \mathbb{P}(H|G,N_m\neq 0) \mathbb{P}(N_m \neq 0|G) \\
			&-\mathbb{P}(H|N_m=0) \mathbb{P}(N_m=0) - \mathbb{P}(H|N_m\neq 0) \mathbb{P}(N_m \neq 0) | \cdot \mathbb{P}(G)\\
			&\leq (D_{1,m} + D_{2,m}) \cdot \mathbb{P}(G),
		\end{align*}
		where
		\begin{align*}
			D_{1,m}  &:= |\mathbb{P}(H|G,N_m=0) \mathbb{P}(N_m=0|G) -\mathbb{P}(H|N_m=0) \mathbb{P}(N_m=0) |, \\
			D_{2,m} &:= | \mathbb{P}(H|G,N_m\neq 0) \mathbb{P}(N_m \neq 0|G)- \mathbb{P}(H|N_m\neq 0) \mathbb{P}(N_m \neq 0) |.
		\end{align*}
		We seek to bound these expressions. For $D_{1,m}$, we use the fact that on the event $\{N_m = 0\}$ the two events $G$ and $H$ are independent. For both $D_{1,m}$ and $D_{2,m}$, we will use that the probability of the complementary event $\{N_m \neq 0\}$ is ``small enough'', cf. Equation \eqref{eq:Nm}. For the first of the terms, we get, using conditional independence of $H$ and $G$ and \eqref{eq:Nm},
		\begin{align*}
			D_{1,m} &=  \mathbb{P}(H|N_m=0) \cdot | \mathbb{P}(N_m=0|G) -  \mathbb{P}(N_m=0)| \\
			&\leq  | \mathbb{P}(N_m=0|G) -  \mathbb{P}(N_m=0)| \\
			& \leq | 1- \mathbb{P}(N_m=0|G) | + |1-  \mathbb{P}(N_m=0)| \\
			&= | 1- \mathbb{P}(N_m=0|G) |  + \tilde \nu Leb(A_m\cap A_0) + o(Leb(A_m\cap A_0)),
		\end{align*}
		and, using the Bayes formula and then the law of total probability,
		\begin{align*}
			| 1- \mathbb{P}(N_m=0|G) |  &= |1- \mathbb{P}(G|N_m=0) \mathbb{P}(N_m=0) \mathbb{P}(G)^{-1}| \\
			&=  \mathbb{P}(G)^{-1} |  \mathbb{P}(G) - \mathbb{P}(G|N_m=0) \mathbb{P}(N_m=0)| \\
			&= \mathbb{P}(G)^{-1} |  \mathbb{P}(G|N_m = 0)\mathbb{P}(N_m = 0) + \mathbb{P}(G|N_m \neq 0)\mathbb{P}(N_m \neq 0) \\
			&- \mathbb{P}(G|N_m=0) \mathbb{P}(N_m=0)| \\
			&= \mathbb{P}(G)^{-1} |  \mathbb{P}(G|N_m \neq 0)\mathbb{P}(N_m \neq 0)| \\
			&\leq \mathbb{P}(G)^{-1}\mathbb{P}(N_m \neq 0) \\
			&=  \mathbb{P}(G)^{-1}( \tilde \nu Leb(A_m \cap A_0) +  o(Leb(A_m \cap A_0) )).
		\end{align*}
		We conclude that 
		\begin{align*}
			D_{1,m} \cdot \mathbb{P}(G) \leq 2  \tilde \nu Leb(A_m \cap A_0) +  o(Leb(A_m \cap A_0) ).
		\end{align*}
		For the second term above, we get
		\begin{align*}
			D_{2,m} &=   |\mathbb{P}(H|G,N_m\neq 0) \mathbb{P}(N_m \neq 0|G) - \mathbb{P}(H|N_m\neq 0) \mathbb{P}(N_m \neq 0)| \\
			&= |\mathbb{P}(H|G,N_m\neq 0) \mathbb{P}(N_m \neq 0|G) - \mathbb{P}(H|G,N_m\neq 0) \mathbb{P}(N_m \neq 0) \\
			&+\mathbb{P}(H|G,N_m\neq 0) \mathbb{P}(N_m \neq 0) - \mathbb{P}(H|N_m\neq 0) \mathbb{P}(N_m \neq 0)| \\
			&\leq E_{1,m} + E_{2,m},
		\end{align*}
		where
		\begin{align*}
			E_{1,m}&:= |\mathbb{P}(H|G,N_m\neq 0) \mathbb{P}(N_m \neq 0|G) - \mathbb{P}(H|G,N_m\neq 0) \mathbb{P}(N_m \neq 0)|, \\
			E_{2,m}&:= |\mathbb{P}(H|G,N_m\neq 0) \mathbb{P}(N_m \neq 0) - \mathbb{P}(H|N_m\neq 0) \mathbb{P}(N_m \neq 0)|.
		\end{align*}
		Now, by \eqref{eq:Nm},
		\begin{align*}
			E_{2,m} &=  \mathbb{P}(N_m \neq 0) |\mathbb{P}(H|G,N_m\neq 0)  - \mathbb{P}(H|N_m\neq 0)| \\
			& \leq  \mathbb{P}(N_m \neq 0) \\
			& = \tilde \nu Leb(A_m \cap A_0) + o(Leb(A_m \cap A_0)).
		\end{align*}
		Also,
		\begin{align*}
			E_{1,m} &=  |\mathbb{P}(H|G,N_m\neq 0) \mathbb{P}(N_m \neq 0|G) - \mathbb{P}(H|G,N_m\neq 0) \mathbb{P}(N_m \neq 0)| \\
			&= \mathbb{P}(H|G,N_m\neq 0)  \cdot |\mathbb{P}(N_m \neq 0|G) - \mathbb{P}(N_m \neq 0)| \\
			& \leq   |\mathbb{P}(N_m \neq 0|G) - \mathbb{P}(N_m \neq 0)| \\
			& \leq \mathbb{P}(N_m \neq 0|G) + \mathbb{P}(N_m \neq 0).
		\end{align*}
		Using Bayes formula, we can write
		\begin{align*}
			\mathbb{P}(N_m \neq 0|G)  &= \mathbb{P}(G|N_m \neq 0)\mathbb{P}(N_m \neq 0)\mathbb{P}(G)^{-1} \\
			&\leq \mathbb{P}(N_m \neq 0)\mathbb{P}(G)^{-1}
		\end{align*}
		so that, from  \eqref{eq:Nm},
		\begin{align*}
			E_{1,m}  \leq (1+\mathbb{P}(G)^{-1})(\tilde \nu Leb(A_m \cap A_0) + o(Leb(A_m \cap A_0))).
		\end{align*}
		We conclude that 
		\begin{align*}
			D_{2,m} \cdot \mathbb{P}(G) \leq  3 \tilde \nu Leb(A_m \cap A_0) +  o(Leb(A_m \cap A_0) ).
		\end{align*}
		Taking it all together, we have that
		\begin{align*}
			|\mathbb{P}(H \cap G) - \mathbb{P}(H)\mathbb{P}(G)|  &\leq (D_{1,m} +D_{2,m}  ) \cdot \mathbb{P}(G) \\
			&\leq 5 \tilde \nu Leb(A_m \cap A_0) +  o(Leb(A_m \cap A_0) ),
		\end{align*}
		implying, since $G$ and $H$ were arbitrary, that (taking supremums)
		\begin{align*}
			\alpha_m \leq  5  \tilde \nu Leb(A_m \cap A_0) +  o(Leb(A_m \cap A_0) ),
		\end{align*}
		which implies that $\alpha_m \leq O( Leb(A_m \cap A_0))$.  
		
		To finish the proof, we show that we also have $\alpha_m \geq O( Leb(A_m \cap A_0))$. Letting $x_1,x_2 \in \Z$ and defining the events $\tilde H := \{ X_0 = x_1\} $ and $\tilde G := \{ X_{m} = x_2\}$,  the results and the arguments in the proof of Lemma \ref{lem:B}  imply that
		\begin{align*}
			|\mathbb{P}(\tilde H \cap \tilde G) - \mathbb{P}(\tilde H)\mathbb{P}(\tilde G)|   = O( Leb(A_m \cap A_0) ), \qquad m \rightarrow \infty.
		\end{align*}
		Since, clearly, $\tilde H \in \mathcal{F}_{-\infty}^0$ and $\tilde G \in  \mathcal{F}_{m}^{\infty}$, we conclude that $\alpha_m =  \sup_{H \in  \mathcal{F}_{-\infty}^0, G \in \mathcal{F}_{m}^{\infty} }| \mathbb{P}(H\cap G) - \mathbb{P}(H) \mathbb{P}(G)| \geq |\mathbb{P}(\tilde H \cap \tilde G) - \mathbb{P}(\tilde H)\mathbb{P}(\tilde G)|   = O( Leb(A_m \cap A_0) )$.

	\end{proof}
	
	\begin{proof}[Proof of Proposition \ref{prop:positiveLB}]
		When $\nu(y) = 0$ for $y<0$ we have $P_{\theta}(L(A_{(i+k)\Delta} \cap A_{i\Delta}) = c) = 0$ for $c < 0$. Further, since the maximal amount of events in $L(A_{(i+k)\Delta} \cap A_{i\Delta})$ is bounded by the number of events in $x_{t+k}$ and $x_t$ (no negative values in the trawl sets), we also have $P_{\theta}(L(A_{(i+k)\Delta} \cap A_{i\Delta}) = c) = 0$ for $c >  \min\{x_{t+k},x_t\}$. This, together with the discussion of the decomposition of trawl sets in Section \ref{sec:pairwise} (cf. Figure \ref{fig:trawl_decomp} in the Supplementary Material), 
		and the law of total probability implies that
		\begin{align*}
			f(x_{i+k},&x_{i};\theta) := \PP_{\theta}\left( X_{(i+k)\Delta} = x_{i+k}, X_{i\Delta} = x_{i}\right)  \\
			=& \sum_{c=-\infty}^{\infty}  \PP_{\theta}\left( X_{(i+k)\Delta} = x_{i+k}, X_{i\Delta} = x_{i}| L(A_{(i+k)\Delta} \cap A_{i\Delta}) = c \right)\cdot \PP_{\theta}\left(L(A_{(i+k)\Delta} \cap A_{i\Delta}) = c\right) \\
			=& \sum_{c=0}^{\max\{x_i,x_{i+k}\}}  \PP_{\theta}\left( L(A_{(i+k)\Delta} \setminus A_{i\Delta}) = x_{i+k} - c\right) \PP_{\theta}\left( L(A_{i\Delta} \setminus A_{(i+k)\Delta}) = x_{i}-c\right)  \\
			&\cdot \PP_{\theta}\left(L(A_{(i+k)\Delta} \cap A_{i\Delta}) = c\right),
		\end{align*}
		as we wanted to show.
	\end{proof}

	\begin{proof}[Proof of Theorem \ref{th:LLN}]
		Due to the stationarity and ergodicity of the IVT processes considered in this paper, the normalized log-composite likelihood function will converge in probability to its population counterpart, i.e.
		\begin{align*}
			Q_n( \theta) := \frac{1}{n} l_{CL} (\theta;X) =  \frac{1}{n} \sum_{k=1}^K \sum_{i=1}^{n-k} \log f(X_{(i+k)\Delta},X_{i\Delta};\theta) \stackrel{\mathbb{P}}{\rightarrow} \mathbb{E}\left[ \sum_{k=1}^K \log f(X_{k\Delta},X_{0};\theta) \right] =: Q(\theta), 
		\end{align*}
		as $n\rightarrow \infty$. By the identifiability condition \eqref{eq:identify}, and the fact that the pairwise likelihoods are indeed proper (bivariate) likelihoods, the information inequality implies that  $Q(\theta)$ is uniquely maximized at $\theta = \theta_0$ \citep[Lemma 2.2 in][p. 2124]{NM1994handbook}. The result now follows from Theorem 4.1 and Theorem 4.3 in \cite{Wooldridge1994handbook}.
	\end{proof}

	\begin{proof}[Proof of Theorem \ref{th:CLT}]
		Let 
		\begin{align*}
			s_n(\theta) := \frac{\partial}{\partial \theta} l_{CL} (\theta;X) =   \sum_{k =1}^K \sum_{i=1}^{n-k} \frac{\partial}{\partial \theta}  \log f(X_{(i+k)\Delta},X_{i\Delta};\theta)
		\end{align*}
		denote the score function and consider the estimating equation related to the MCL estimator $\hat \theta = \hat \theta^{CL}$, namely $s_n(\hat \theta) = 0$. Using this equation, we Taylor expand $s_n(\hat \theta)$ around the true parameter vector $\theta_0$ to get
		\begin{align*}
			s_n(\theta_0) + \frac{\partial}{\partial \theta'} s_n(\bar \theta) (\hat \theta - \theta_0) = 0,
		\end{align*}
		where $\bar \theta$ lies on the line segment between $\theta_0$ and $\hat \theta$ and $\frac{\partial}{\partial \theta'}s_n(\bar \theta)$ is shorthand for $\frac{\partial}{\partial \theta'}s_n( \theta)|_{\theta=\bar \theta}$. Rearranging this equation and multiplying through by $\sqrt{n}$, we get
		\begin{align*}
			\sqrt{n} (\hat \theta - \theta_0) = - \left( \frac{1}{n} \frac{\partial}{\partial \theta'} s_n(\bar \theta)  \right)^{-1}  n^{-1/2} s_n(\theta_0).
		\end{align*}
		Stationarity and ergodicity, along with consistency of $\hat \theta$ due to Theorem \ref{th:LLN}, implies that $-\frac{1}{n} \frac{\partial}{\partial \theta'}  s_n(\bar \theta)' \stackrel{\mathbb{P}}{\rightarrow} H(\theta_0)$ as $n \rightarrow \infty$. To prove the result, we thus need to show that $ n^{-1/2}  s_n(\theta_0) \stackrel{(d)}{\rightarrow} N(0,V(\theta_0))$ as $n\rightarrow \infty$. By the mixing properties of the IVT process $X$, given in Theorem \ref{th:mixing}, it is enough to show that $\mathbb{E}[s_n(\theta_0)] = 0$ and $Var\left(n^{-1/2} s_n(\theta_0)\right) \rightarrow V(\theta_0)$ as $n\rightarrow \infty$ \citep[e.g.][Corollary 24.7, p. 387]{Davidson1994}.\footnote{Note that the crucial condition (c') of Corollary 24.7 in \cite{Davidson1994} relies on the IVT process $X$ being mixing of size $\phi_0$ for some $\phi_0 > 1$. This rules out the long memory processes, as shown in Theorem \ref{th:mixing}.}
		
		
		To show this, we consider, for simplicity, the case where $\theta$ is a scalar. The vector case is similar, but with slightly more involved notation. First note that, clearly, 
		\begin{align*}
			\mathbb{E}\left[ s_n(\theta_0)  \right] = \mathbb{E}\left[ \frac{\partial}{\partial \theta}  l _{CL}(\theta;X) |_{\theta = \theta_0}   \right] = 0.
		\end{align*}
		Also
		\begin{align*}
			Var\left( s_n(\theta_0) \right)  &= Var\left( \frac{\partial}{\partial \theta}  l _{CL}(\theta;X) |_{\theta = \theta_0}   \right) \\
			&=Var\left( \sum_{k=1}^K \sum_{i=1}^{n-k} \frac{\partial}{\partial \theta} \log \left.f (X_{i\Delta},X_{(i+k)\Delta};\theta) \right|_{\theta=\theta_0} \right) \\
			&= \sum_{k=1}^K \sum_{i=1}^{n-k} Var\left(\frac{\partial}{\partial \theta} \log \left.f (X_{i\Delta},X_{(i+k)\Delta};\theta) \right|_{\theta=\theta_0} \right) \\
			&+ \sum_{k=1}^K\sum_{k'=1}^K \sum_{i \neq j}  Cov\left(\frac{\partial}{\partial \theta} \log \left.f (X_{i\Delta},X_{(i+k)\Delta};\theta) \right|_{\theta=\theta_0},\frac{\partial}{\partial \theta} \log \left.f (X_{j\Delta},X_{(j+k')\Delta};\theta) \right|_{\theta=\theta_0} \right).
		\end{align*}
		Due to stationarity, the first sum is $O(n)$ as $n\rightarrow \infty$. To prove the proposition, we, therefore, investigate the second sum. With slight abuse of notation, let $\frac{\partial}{\partial \theta} \log \left.f (X_{i\Delta},X_{(i+k)\Delta};\theta) \right|_{\theta=\theta_0}$ be denoted by $\frac{\partial}{\partial \theta} \log f (X_{i\Delta},X_{(i+k)\Delta};\theta_0)$. For $l,k,k'\geq 1$, define also the joint probability mass functions
		\begin{align*}
			f_l(x_1,x_2,x_3,x_4;k,k',\theta) := \mathbb{P}_{\theta}\left( X_0 = x_1, X_{k\Delta} = x_2, X_{l\Delta} = x_3, X_{(l+k')\Delta} = x_4   \right), \quad x_1,x_2,x_3,x_4 \in \mathbb{Z}.
		\end{align*}
		and
		\begin{align*}
			f_k(x_1,x_2) := \mathbb{P}_{\theta}\left( X_0 = x_1, X_{k\Delta} = x_2  \right), \quad x_1,x_2 \in \mathbb{Z}.
		\end{align*}
		Now, using that 
		\begin{align*}
			\mathbb{E}\left[\frac{\partial}{\partial \theta} \log f (X_{i\Delta},X_{(i+k)\Delta};\theta_0) \right] = 0,
		\end{align*}
		we have, for all $i, j, k, k'$, 
		\begin{align*}
			&Cov\left(\frac{\partial}{\partial \theta} \log f (X_{i\Delta},X_{(i+k)\Delta};\theta_0) ,\frac{\partial}{\partial \theta} \log f (X_{j\Delta},X_{(j+k')\Delta};\theta_0) \right) \\
			&= \mathbb{E}\left[\frac{\partial}{\partial \theta} \log f (X_{i\Delta},X_{(i+k)\Delta};\theta_0) \frac{\partial}{\partial \theta} \log f (X_{j\Delta},X_{(j+k')\Delta};\theta_0) \right] \\
			&= \sum_{x_1=-\infty}^{\infty} \sum_{x_{2}=-\infty}^{\infty} \sum_{x_{3}=-\infty}^{\infty} \sum_{x_{4}=-\infty}^{\infty} \frac{\partial}{\partial \theta}  f_k (x_1,x_{2};\theta_0)\frac{\partial}{\partial \theta} f_{k'} (x_3,x_{4};\theta_0)   \frac{f_{|i-j|}(x_1,x_{2},x_3,x_{4};k,k',\theta_0)}{f_k (x_1,x_{2};\theta_0)f_{k'} (x_3,x_{4};\theta_0)} \\
			&= \sum_{x_1=-\infty}^{\infty} \sum_{x_{2}=-\infty}^{\infty} \sum_{x_{3}=-\infty}^{\infty} \sum_{x_{4}=-\infty}^{\infty} \frac{\partial}{\partial \theta}  f_k (x_1,x_{2};\theta_0)\frac{\partial}{\partial \theta} f_{k'} (x_3,x_{4};\theta_0)   \left(\frac{f_{|i-j|}(x_1,x_{2},x_3,x_{4};k,k',\theta_0)}{f_k (x_1,x_{2};\theta_0)f_{k'} (x_3,x_{4};\theta_0)} - 1\right),
		\end{align*}
		where the last equality follows because e.g., 
		\begin{align*}
			\sum_{x_1=-\infty}^{\infty} \sum_{x_{2}=-\infty}^{\infty} \frac{\partial}{\partial \theta}  f_k (x_1,x_{2};\theta_0) =  \frac{\partial}{\partial \theta} \sum_{x_1=-\infty}^{\infty} \sum_{x_{2}=-\infty}^{\infty}   f_k (x_1,x_{2};\theta_0) = \frac{\partial}{\partial \theta}1 = 0. 
		\end{align*}
		Now, Lemma \ref{lem:B} below shows that
		\begin{align*}
			\left(\frac{f_{n}(x_1,x_{2},x_3,x_{4};k,k',\theta_0)}{f_k (x_1,x_{2};\theta_0)f_{k'} (x_3,x_{4};\theta_0)} - 1\right) = O(Leb(A_{n\Delta} \cap A_0)), \quad  n \rightarrow  \infty,
		\end{align*}
		from which we conclude, using Equation \eqref{eq.:corr}, i.e.~$Leb(A_{n\Delta} \cap A_0) = \rho(n\Delta) Leb(A)$, and the condition on $\rho$ imposed in the theorem, that the second sum in the expression for $Var\left( s_n(\theta_0) \right)$ is $O(n)$ as well. Indeed, taking it all together, we have that, as $n\rightarrow \infty$,
		\begin{align*}
			n^{-1} Var\left( s_n(\theta_0) \right) & \rightarrow  \sum_{k=1}^K Var\left(\frac{\partial}{\partial \theta} \log \left.f (X_{0},X_{k\Delta};\theta) \right|_{\theta=\theta_0} \right) \\
			&+ 2 \sum_{k=1}^K\sum_{k'=1}^K \sum_{i =1}^{\infty}  Cov\left(\frac{\partial}{\partial \theta} \log \left.f (X_{0},X_{k\Delta};\theta) \right|_{\theta=\theta_0},\frac{\partial}{\partial \theta} \log \left.f (X_{i\Delta},X_{(i+k')\Delta};\theta) \right|_{\theta=\theta_0} \right) \\
			&=: V(\theta_0),
		\end{align*}
		where the series converges. This finalizes the proof.  

	\end{proof}

	\begin{lemma}\label{lem:B}
		Fix $k,k'\geq 1$, let $X$ be an IVT process, let $f_k(\cdot, \cdot;\theta)$ be the  joint PMF of $(X_0,X_{k\Delta})$, and let $f_n(\cdot,\cdot,\cdot,\cdot;k,k',\theta)$ be the joint PMF of $(X_0,X_{k\Delta},X_{n\Delta},X_{(n+k')\Delta})$. That is
		\begin{align*}
			f_k(x_1,x_2;\theta) := \mathbb{P}_{\theta}\left( X_0 = x_1, X_{k\Delta} = x_2   \right), \quad x_1,x_2 \in \mathbb{Z},
		\end{align*}
		and
		\begin{align*}
			f_n(x_1,x_2,x_3,x_4;k,k',\theta) := \mathbb{P}_{\theta}\left( X_0 = x_1, X_{k\Delta} = x_2, X_{n\Delta} = x_3, X_{(n+k')\Delta} = x_4   \right), \quad x_1,x_2,x_3,x_4 \in \mathbb{Z}.
		\end{align*}
		Define the function
		\begin{align*}
			G_n(x_1,x_2,x_3,x_4;k,k',\theta) := \left( \frac{f_n(x_1,x_2,x_3,x_4;k,k',\theta)}{f_k(x_1,x_2;\theta) f_{k'}(x_3,x_4;\theta) } - 1\right).
		\end{align*}
		The following holds:
		\begin{align*}
			\frac{1}{Leb(A_{k\Delta} \cap A_{n\Delta})} G_n(x_1,x_2,x_3,x_4;k,k',\theta)  \rightarrow G(x_1,x_2,x_3,x_4;k,k',\theta), \quad n\rightarrow \infty,
		\end{align*}
		where $G$ is a function, given in Equation  \eqref{eq:app_G} below, that depends on the L\'evy basis and trawl function of $X$. (In  Remark \ref{rem:B_poi} below, we give the function $G$ in the special case where the L\'evy basis is Poissonian and the trawl function is the Gamma trawl.)
	\end{lemma}

	%
	%

	\begin{proof}[Proof of Lemma \ref{lem:B}]
		Letting $ f_n(x_1,x_2|x_3,x_4;k,k',\theta):= \mathbb{P}_{\theta}\left( X_0 = x_1, X_{k\Delta} = x_2| X_{n\Delta} = x_3, X_{(n+k')\Delta} = x_4   \right)$, we can write
		\begin{align*}
			G_n(x_1,x_2,x_3,x_4;k,k',\theta) &=  \frac{f_n(x_1,x_2,x_3,x_4;\theta)-f_k(x_1,x_2;\theta) f_{k'}(x_3,x_4;\theta)}{f_k(x_1,x_2;\theta) f_{k'}(x_3,x_4;\theta) } \\
			&=  \frac{f_n(x_1,x_2|x_3,x_4;k,k',\theta)f_{k'}(x_3,x_4;\theta)-f_k(x_1,x_2;\theta) f_{k'}(x_3,x_4;\theta)}{f_k(x_1,x_2;\theta) f_{k'}(x_3,x_4;\theta) } \\
			&= (f_n(x_1,x_2|x_3,x_4;k,k',\theta)-f_k(x_1,x_2;\theta)) f_k(x_1,x_2;\theta)^{-1}.
		\end{align*}
		To prove the lemma, we, therefore, study the asymptotic behaviour of $f_n(x_1,x_2|x_3,x_4;k,k',\theta)-f_k(x_1,x_2;\theta)$ as $n\rightarrow \infty$.
		
		Recall first the decomposition of the trawl sets into three disjoint sets which led to Proposition \ref{prop:positiveLB}, see Figure \ref{fig:trawl_decomp}. 
		In a similar manner, we can decompose the four trawl sets associated to $X_0 = L(A_0)$, $X_{k\Delta} = L(A_{k\Delta})$, $X_{n\Delta} = L(A_{n\Delta})$, and $X_{(n+k')\Delta} = L(A_{(n+k')\Delta})$, into $10$ disjoint sets as illustrated in Figure \ref{fig:trawl_decomp2} below. For ease of notation, we ignore the dependence on $n$, $k$, and $k'$ for a moment and write 
		\begin{align*}
			A_0 &= C_3 \cup C_4 \cup C_6 \cup C_7, \quad A_{k\Delta} = C_1 \cup C_2 \cup C_3 \cup C_4 \cup C_5\cup C_6, \\
			A_{n\Delta} &= C_1 \cup C_2 \cup C_3 \cup C_4 \cup D_2 \cup D_3, \quad A_{(n+k')\Delta} = C_1 \cup C_3 \cup D_1 \cup D_2,
		\end{align*}
		where the sets $C_1, C_2, \ldots, C_7, D_1, D_2, D_3$ are disjoint. 
		We will use below that $\lim_{n\rightarrow \infty} Leb(C_j) = 0$ for $j=1,2,3,4$, $\lim_{n\rightarrow \infty} Leb(C_5) = Leb(A_{k\Delta} \setminus A_0)$,  $\lim_{n\rightarrow \infty} Leb(C_6) = Leb(A_0 \cap A_{k\Delta})$, $\lim_{n\rightarrow \infty} Leb(C_7) = Leb(A_0 \setminus A_{k\Delta})$, 
		$\lim_{n\rightarrow \infty} Leb(D_1) = 
		Leb(A_{k'\Delta} \setminus A_{0})$, 
		$\lim_{n\rightarrow \infty} Leb(D_2) 
		=Leb(A_{0} \cap A_{k'\Delta})$, and $\lim_{n\rightarrow \infty} Leb(D_3) 
		= Leb(A_{0} \setminus A_{k'\Delta})$, cf. Figure \ref{fig:trawl_decomp2}. 
		
		Using this decomposition and the law of total probability, we may write
		\begin{align*}
			f(x_1,x_2;\theta) &=  \sum_{c_1,c_2,c_3,c_4 = -\infty}^{\infty}  \PP_{\theta}\left( L(C_6) + L(C_7) = x_1 - c_3 - c_4,L(C_5) + L(C_6) = x_2 - c_1 - c_2 - c_3 - c_4\right)  \\
			&\cdot \prod_{j=1}^4  \PP_{\theta}\left( L(C_j) = c_j \right) 
		\end{align*}
		and
		\begin{align*}
			&f(x_1,x_2|x_3,x_4;\theta) \\
			&=  \sum_{c_1,c_2,c_3,c_4 = -\infty}^{\infty} \PP_{\theta}\left( L(C_6) + L(C_7) = x_1 - c_3 - c_4,L(C_5) + L(C_6) = x_2 - c_1 - c_2 - c_3 - c_4\right)  \\
			&\cdot \PP_{\theta}\left( L(C_1) = c_1, L(C_2) = c_2, L(C_3) = c_3, L(C_4) = c_4 | X_{n\Delta} = x_3, X_{(n+k')\Delta} = x_4 \right).
		\end{align*}
		Taking these together, we get
		\begin{align*}
			&f_n(x_1,x_2|x_3,x_4;\theta)-f(x_1,x_2;\theta) \\
			&=   \sum_{c_1,c_2,c_3,c_4 = -\infty}^{\infty} \PP_{\theta}\left( L(C_6) + L(C_7) = x_1 - c_3 - c_4,L(C_5) + L(C_6) = x_2 - c_1 - c_2 - c_3 - c_4\right)  \\
			&\cdot \left(  \PP_{\theta}\left( L(C_1) = c_1, L(C_2) = c_2, L(C_3) = c_3, L(C_4) = c_4 | X_{n\Delta} = x_3, X_{(n+k')\Delta} = x_4 \right) - \prod_{j=1}^4  \PP_{\theta}\left( L(C_j) = c_j \right)  \right).
		\end{align*}
		Note that, for the first term in the  parenthesis,  the following holds
		\begin{align*}
			&\PP_{\theta}\left( L(C_1) = c_1, L(C_2) = c_2, L(C_3) = c_3, L(C_4) = c_4 | X_{n\Delta} = x_3, X_{(n+k')\Delta} = x_4 \right) \\
			&= f(x_3,x_4;\theta)^{-1}  \PP_{\theta}\left( L(C_1) = c_1, L(C_2) = c_2, L(C_3) = c_3, L(C_4) = c_4 , X_{n\Delta} = x_3, X_{(n+k')\Delta} = x_4 \right) \\
			&= f(x_3,x_4;\theta)^{-1}  \PP_{\theta}\left( L(D_2) + L(D_3) = x_3 -c_1 - c_2- c_3 - c_4, L(D_1) + L(D_2) = x_4 - c_1 - c_3\right)  \\
			&\cdot \prod_{j=1}^4  \PP_{\theta}\left( L(C_j) = c_j \right),
		\end{align*}
		which allows us to write
		\begin{align*}
			&f_n(x_1,x_2|x_3,x_4;\theta)-f(x_1,x_2;\theta) \\
			&=   \sum_{c_1,c_2,c_3,c_4 = -\infty}^{\infty} \PP_{\theta}\left( L(C_6) + L(C_7) = x_1 - c_3 - c_4,L(C_5) + L(C_6) = x_2 - c_1 - c_2 - c_3 - c_4\right)  \\
			&\cdot   \prod_{j=1}^4  \PP_{\theta}\left( L(C_j) = c_j \right) f(x_3,x_4;\theta)^{-1} \\
			&\cdot \left(  \PP_{\theta}\left( L(D_2) + L(D_3) = x_3 -c_1 - c_2- c_3 - c_4, L(D_1) + L(D_2) = x_4 - c_1 - c_3\right) - f(x_3,x_4)\right).
		\end{align*}

		
		Define the set $\mathcal{C}_0 := \{(c_1,c_2,c_3,c_4) \in \mathbb{Z}^4 : c_i \neq 0 \textnormal{ for at least one } i=1,2,3,4\}$. The above calculations imply that
		\begin{align}\label{eq:Fn0}
			f_n(x_1,x_2|x_3,x_4;\theta)-f(x_1,x_2;\theta) =  f(x_3,x_4)^{-1}\left( F_n^{(1)} + F_n^{(2)} \right),
		\end{align}
		where
		\begin{align*}
			F_n^{(1)} &:= \PP_{\theta}\left( L(C_6) + L(C_7) = x_1, L(C_5) + L(C_6) = x_2 \right) \prod_{j=1}^4  \PP_{\theta}\left( L(C_j) = 0 \right)  \\
			& \cdot  \left(\PP_{\theta}\left( L(D_2) + L(D_3) = x_3, L(D_1) + L(D_2) = x_4\right) -f(x_3,x_4;\theta)  \right)
		\end{align*}
		and
		\begin{align*}
			F_n^{(2)} &:= \sum_{(c_1,c_2,c_3,c_4) \in \mathcal{C}_0} \PP_{\theta}\left( L(C_6) + L(C_7) = x_1 - c_3 - c_4, L(C_5) + L(C_6) = x_2 - c_1 - c_2 - c_3 - c_4\right)  \\
			&\cdot \prod_{j=1}^4  \PP_{\theta}\left( L(C_j) = c_j \right)  \\
			& \cdot  \left(  \PP_{\theta}\left( L(D_2) + L(D_3) = x_3 -c_1 - c_2- c_3 - c_4,L(D_1) + L(D_2) = x_4 - c_1 - c_3\right) -f(x_3,x_4;\theta)  \right).
		\end{align*}
		We can think of $F_n^{(1)}$ as the part of $(f_n(x_1,x_2|x_3,x_4;\theta)-f(x_1,x_2;\theta)) f(x_3,x_4)$ where $c_1 = c_2 = c_3 = c_4 = 0$, while $F_n^{(2)}$ is the remainder.
		
		We study first the behavior of $F_n^{(1)}$ as $n\rightarrow \infty$. Considering the first two factors of this term, the continuity of the probability measure $\PP_{\theta}(\cdot)$ implies that
		\begin{align*}
			& \lim_{n\rightarrow\infty} \PP_{\theta}\left( L(C_6) + L(C_7) = x_1, L(C_5) + L(C_6) = x_2 \right) \prod_{j=1}^4  \PP_{\theta}\left( L(C_j) = 0 \right)  \\
			&=\PP_{\theta}\left( L(A_0 \cap A_{k\Delta}) + L(A_0\setminus A_{k\Delta} ) = x_1, L(A_0 \cap A_{k\Delta}) + L(A_{k\Delta} \setminus A_{0} ) = x_2 \right) \\
			&= f(x_1,x_2;\theta).
		\end{align*}
		The third term in $F_n^{(1)}$, i.e. 
		\begin{align*}
			\PP_{\theta}\left( L(D_2) + L(D_3) = x_3, L(D_1) + L(D_2) = x_4\right) -f(x_3,x_4;\theta) 
		\end{align*}
		will, by the same logic as above, converge to zero as $n\rightarrow \infty$. In fact, by decomposition of the trawl sets of $f(x_3,x_4;\theta)$ in the same manner as above, we get that
		\begin{align*}
			& \PP_{\theta}\left( L(D_2) + L(D_3) = x_3 , L(D_1) + L(D_2) = x_4\right) - f(x_3,x_4;\theta) \\
			&=  \PP_{\theta}\left( L(D_2) + L(D_3) = x_3 , L(D_1) + L(D_2) = x_4\right) \left( 1- \prod_{j=1}^4   \PP_{\theta}\left( L(C_j) = 0 \right) \right) \\
			&- \sum_{(c_1,c_2,c_3,c_4) \in \mathcal{C}_0}   \PP_{\theta}\left( L(D_2) + L(D_3) = x_3 - c_1 -c_2 - c_3 - c_4 , L(D_1) + L(D_2) = x_4 - c_1 -c_3\right) \\
			&\cdot  \prod_{j=1}^4   \PP_{\theta}\left( L(C_j) = c_j \right).
		\end{align*}
		In the first part of Lemma \ref{lem:B2} below, we show that there exists a constant $\tilde \nu > 0$, such that
		\begin{align*}
			1-\prod_{j=1}^4   \PP_{\theta}\left( L(C_j) = 0 \right)  = \tilde \nu Leb( A_{k\Delta} \cap A_{n\Delta}) + o\left(Leb( A_{k\Delta} \cap A_{n\Delta})\right),
		\end{align*}
		as $n \rightarrow \infty$. Similarly,  in the second part of Lemma \ref{lem:B2}, we  show that there exists a non-negative function $\tilde \eta$ concentrated on the integers, such that, for $c_j \neq 0$,
		\begin{align*}
			\PP_{\theta}\left( L(C_j) = c_j \right) = \tilde \eta(c_j) \tilde \nu Leb(C_j) e^{-\tilde \nu Leb(C_j)} + o\left( Leb(C_j) \right),
		\end{align*}
		as $n \rightarrow \infty$, while
		\begin{align*}
			\PP_{\theta}\left( L(C_j) = 0 \right) = e^{-\tilde \nu Leb(C_j)} + o\left( Leb(C_j) \right),
		\end{align*}
		as $n \rightarrow \infty$. This shows that for quadruplets $(c_1,c_2,c_3,c_4)$ where $c_i \neq 0$ for some $i=1,2,3,4$ and $c_j = 0$ for the remaining $j\neq i$ (i.e.~quadruplets of the form $(c_1,0,0,0)$, $(0,c_2,0,0)$, $(0,0,c_3,0)$ or $(0,0,0,c_4)$), we have
		\begin{align*}
			\prod_{j=1}^4 \PP_{\theta}\left( L(C_j) = c_j \right) &= \tilde \eta(c_i) \tilde \nu Leb(C_i) e^{-\tilde \nu \sum_{j=1}^4 Leb(C_j)} + o\left( \sum_{j=1}^4Leb(C_j)\right) \\
			&=  \tilde \eta(c_i) \tilde \nu Leb(C_i) e^{-\tilde \nu  Leb(  A_{k\Delta} \cap A_{n\Delta} )} + o\left( Leb( A_{k\Delta} \cap A_{n\Delta})\right),
		\end{align*}
		as $n \rightarrow \infty$. Conversely, for quadruplets $(c_1,c_2,c_3,c_4)$ where $c_i, c_j \neq 0$ for at least two distinct $i,j=1,2,3,4$, we have
		\begin{align*}
			\prod_{j=1}^4 \PP_{\theta}\left( L(C_j) = c_j \right) = o\left( Leb( A_{k\Delta} \cap A_{n\Delta})\right),
		\end{align*}
		as $n \rightarrow \infty$.
		
		Define the numbers $a_j := \lim_{n\rightarrow \infty} \frac{Leb(C_j)}{Leb(A_{n\Delta} \cap A_{k\Delta})} \geq 0$, $j=1,2,3,4$, which are such that $\sum_{j=1}^4 a_j = 1$ since $C_1 \cup C_2 \cup C_3 \cup C_4 = A_{n\Delta} \cap A_{k\Delta}$, cf. Figure \ref{fig:trawl_decomp2} below. Taking the above together, we may, after a little algebra, conclude that, as $n\rightarrow \infty$,
		\begin{align}
			\frac{F_n^{(1)}}{Leb( A_{k\Delta} \cap A_{n\Delta}) }  \rightarrow &\  \tilde  \nu f(x_1,x_2;\theta) f(x_3,x_4;\theta) \label{eq:Fn1} \\
			&-  \tilde \nu f(x_1,x_2)  \sum_{c\neq 0} \tilde \eta(c) \left(  (a_1 + a_3) f(x_3-c,x_4-c;\theta) + (a_2 + a_4) f(x_3- c ,x_4;\theta)  \right). \nonumber
		\end{align}
		
		Turning now to $F_n^{(2)}$, similar calculations yield that, as $n\rightarrow \infty$,
		\begin{align}
			\frac{F_n^{(2)}}{Leb( A_{k\Delta} \cap A_{n\Delta}) }  \rightarrow  & - \tilde \nu  f(x_3,x_4;\theta)  \sum_{c\neq 0} \tilde \eta(c)  \left( (a_1+a_2)f(x_1 ,x_2-c;\theta) +  (a_3 + a_4)f(x_1-c,x_2-c;\theta) \right)\nonumber \\
			&+ \tilde \nu a_1 \sum_{c\neq 0} \tilde \eta (c) f(x_1,x_2-c)f(x_3-c,x_4-c) \nonumber\\
			&+ \tilde \nu a_2 \sum_{c\neq 0} \tilde \eta (c) f(x_1,x_2-c)f(x_3-c,x_4) \nonumber\\
			&+ \tilde \nu a_3 \sum_{c\neq 0} \tilde \eta (c) f(x_1-c,x_2-c)f(x_3-c,x_4-c) \nonumber\\
			&+ \tilde \nu a_4 \sum_{c\neq 0} \tilde \eta (c) f(x_1-c,x_2-c)f(x_3-c,x_4). \label{eq:Fn2}
		\end{align}
		
		Finally, recalling Equation \eqref{eq:Fn0}, we can conclude that
		\begin{align}
			\lim_{n\rightarrow \infty} \frac{G_n(x_1,x_2,x_3,x_4;k,k',\theta)}{Leb( A_{k\Delta} \cap A_{n\Delta}) } &=  \lim_{n\rightarrow \infty}  \frac{F_n^{(1)} + F_n^{(2)} }{Leb( A_{k\Delta} \cap A_{n\Delta}) } f_k(x_1,x_2;\theta)^{-1}f_{k'}(x_3,x_4;\theta)^{-1} \nonumber \\
			&=: G(x_1,x_2,x_3,x_4;k,k',\theta), \label{eq:app_G}
		\end{align}
		where $\lim_{n\rightarrow \infty} \frac{1}{Leb( A_{k\Delta} \cap A_{n\Delta}) } F_n^{(i)}$ for $i = 1,2$ are given above in Equations \eqref{eq:Fn1}--\eqref{eq:Fn2}. (See the following Remark \ref{rem:B_poi} for how the expression for $G$ simplifies slightly in the case of an IVT process with Poisson L\'evy basis and Gamma trawl function.)
		
		\begin{figure}[!t]
			\centering
			\includegraphics[scale=0.95]{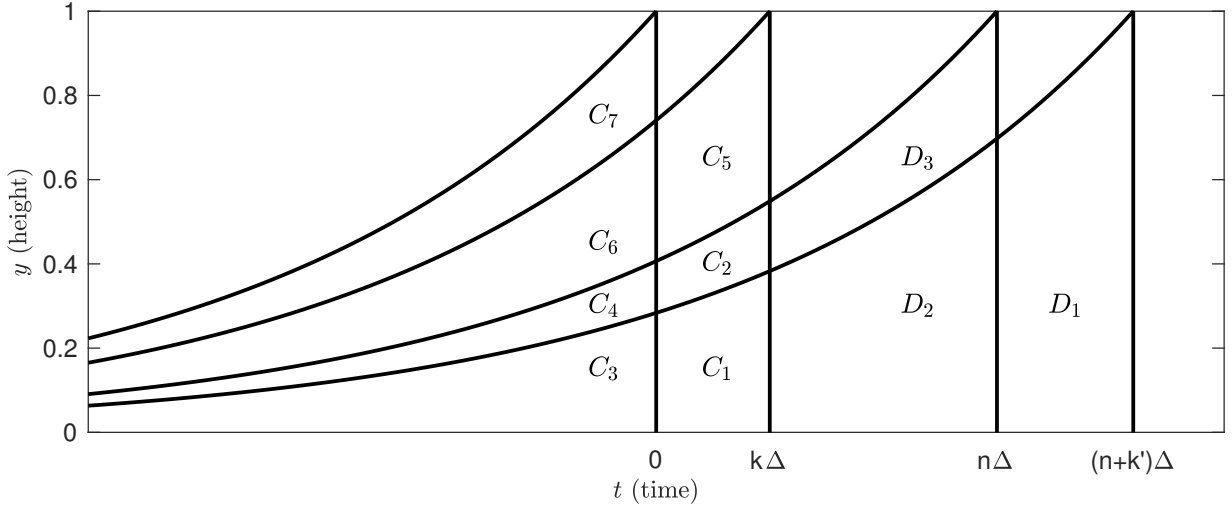} 
			\caption{\it  Illustration of the decomposition of the four trawl sets $A_0, A_{k\Delta}, A_{n\Delta}, A_{(n+k')\Delta}$.}
			\label{fig:trawl_decomp2}
		\end{figure}

	\end{proof}
	
	\begin{remark}\label{rem:B_poi} 
		Note that in the case of a Poisson L\'evy basis (Example \ref{ex:Poisson} in the main article),  
		we have $\tilde \nu = \nu$, $\tilde \eta(1) = 1$, and $\tilde \eta(c) = 0$ for $c\neq 1$. Further, in the case of $d$ being a Gamma trawl function (Example \ref{ex:GAM}), it is straightforward to show that $a_1 = a_2=a_4=0$ and $a_3=1$. For this specification, the limit in the proof of Lemma \ref{lem:B} simplifies somewhat. Indeed, in this case, Equation \eqref{eq:app_G} yields
		\begin{align*}
			G(x_1,x_2,x_3,x_4;k,k',\theta) = \nu \left( \frac{f_k(x_1-1,x_2-1;\theta)}{f_k(x_1,x_2;\theta)} - 1\right) \left(  \frac{f_{k'}(x_3-1,x_4-1;\theta)}{f_{k'}(x_3,x_4;\theta)} - 1\right).
		\end{align*}
		
	\end{remark}

	\begin{lemma}\label{lem:B2}
		In the setting of the proof of Lemma \ref{lem:B}, we have the following two-part result.
		
		(First part) There exists a constant $\tilde \nu>0$, such that 
		\begin{align*}
			1-\prod_{j=1}^4   \PP_{\theta}\left( L(C_j) = 0 \right) = \tilde \nu Leb( A_{k\Delta} \cap A_{n\Delta}) + o\left(Leb( A_{k\Delta} \cap A_{n\Delta})\right),
		\end{align*}
		as $n \rightarrow \infty$.
		
		(Second part) There exists a non-negative function $\tilde \eta$, concentrated on the integers, such that, for $c_j \neq 0$,
		\begin{align*}
			\PP_{\theta}\left( L(C_j) = c_j \right) = \tilde \eta(c_j) \tilde \nu Leb(C_j) e^{-\tilde \nu Leb(C_j)} + o\left( Leb(C_j) \right),
		\end{align*}
		as $n \rightarrow \infty$, while
		\begin{align*}
			\PP_{\theta}\left( L(C_j) = 0 \right) = e^{-\tilde \nu Leb(C_j)} + o\left( Leb(C_j) \right),
		\end{align*}
		as $n \rightarrow \infty$.

	\end{lemma}
	
	\begin{proof}[Proof of Lemma \ref{lem:B2}]
		The proof of the lemma relies on the alternative representation of the L\'evy basis $L$ given at the start of the Appendix, see Equation \eqref{eq:Lalt}.
		
		(Proof of second part) Note that since $\eta(0) = 0$, we have $\PP_{\theta}\left( Y_i = 0\right) = 0$. Using this in the setup of Lemma \ref{lem:B}, we get, from Equations \eqref{eq:Lalt} and \eqref{eq:NBq},
		\begin{align*}
			\PP_{\theta}\left( L(C_j) = 0 \right)  &= e^{-\tilde \nu Leb(C_j)}  + \sum_{q=2}^{\infty} \PP_{\theta}\left( \sum_{i=1}^q Y_i = x\right) \PP_{\theta}( \tilde N(C_j) = q) \\
			&=  e^{-\tilde \nu Leb(C_j) } + o\left( Leb(C_j) \right),
		\end{align*}
		as $n \rightarrow \infty$, while for $c_j \neq 0$,
		\begin{align*}
			\PP_{\theta}\left( L(C_j) = c_j \right)  &= \PP_{\theta}\left( Y_1 = c_j\right) \PP_{\theta}( \tilde N(C_j) = 1)  + \sum_{q=2}^{\infty} \PP_{\theta}\left( \sum_{i=1}^q Y_i = x\right) \PP_{\theta}( \tilde N(C_j) = q) \\
			&=  \tilde \eta(c_j) \tilde \nu Leb(C_j) e^{-\tilde \nu Leb(C_j)}  + o\left( Leb(C_j) \right),
		\end{align*}
		as $n\rightarrow \infty$. This proves the second part of the lemma. 
		
		(Proof of first part) As for the first part, use Equations \eqref{eq:Lalt} and \eqref{eq:NBq} to write
		\begin{align*}
			\prod_{j=1}^4   \PP_{\theta}\left( L(C_j) = 0 \right)  &= e^{-\tilde \nu \sum_{j=1}^4 Leb(C_j)} + o\left( \sum_{j=1}^4 Leb(C_j)\right) \\
			&=  e^{-\tilde \nu Leb(A_{n\Delta} \cap A_{k\Delta}) } + o\left( Leb(A_{n\Delta} \cap A_{k\Delta}) \right) \\
			&= 1 - \tilde  \nu Leb(A_{n\Delta} \cap A_{k\Delta}) + o\left( Leb(A_{n\Delta} \cap A_{k\Delta}) \right),
		\end{align*}
		as $n \rightarrow \infty$, where we in the last line Taylor expanded the exponential function. This proves the first part of the lemma.

	\end{proof}

	\begin{proof}[Proof of Theorem \ref{th:aH}]
		With similar calculations to those used in the proof of Theorem  \ref{th:CLT}, we can write
		\begin{align*}
			n^{H/2} (\hat \theta - \theta_0) = - \left( \frac{1}{n} \frac{\partial}{\partial \theta'} s_n(\bar \theta)  \right)^{-1}  n^{H/2-1} s_n(\theta_0).
		\end{align*}
		We again have that $-\frac{1}{n} \frac{\partial}{\partial \theta}  s_n(\bar \theta)' \stackrel{\mathbb{P}}{\rightarrow} H(\theta_0)$ as $n \rightarrow \infty$. As in the proof of Theorem \ref{th:CLT}, we can write
		\begin{align*}
			Var\left(s_n(\theta_0) \right) = A_{1,n} + A_{2,n},
		\end{align*}
		where $A_{1,n}$ is $O(n)$. Further, Lemma \ref{lem:B} implies that
		\begin{align*}
			A_{2,n} =  O(n^{2} Leb(A_{n\Delta} \cap A_0)),
		\end{align*}
		as $n\rightarrow \infty$. Equation \eqref{eq.:corr} along with the condition on $\rho$ imposed in the theorem thus yields
		\begin{align*}
			A_{2,n} =  O(n^{2} L_{\infty}(n\Delta) n^{-H}),
		\end{align*}
		as $n\rightarrow \infty$. Using this, we get that, for all $\epsilon > 0$,
		\begin{align}\label{eq:pH}
			n^{H - 2 \pm 2\epsilon} Var\left(s_n(\theta_0) \right) = O(n^{H-1 \pm 2 \epsilon}) +  O(L_{\infty}(n\Delta)n^{\pm 2\epsilon}),
		\end{align}
		as $n \rightarrow \infty$. Finally, we recall the so-called \emph{Potter bounds} for slowly varying functions: Since $L_{\infty}$ is a slowly varying function, for all $\epsilon > 0$ it holds that \citep[][Theorem 1.5.6(ii)]{BGT}
		\begin{align*}
			L_{\infty}(n \Delta) n^{\epsilon} \rightarrow \infty \quad \textnormal{ and } \quad L_{\infty}(n \Delta) n^{-\epsilon} \rightarrow 0,
		\end{align*}
		as $n\rightarrow \infty$. Combining the Potter bounds with \eqref{eq:pH} yields the required results.
	\end{proof}
	
	\begin{proof}[Proof of Theorem \ref{th:conjecture}]
		Set, for simplicity, $K=1$.  From the proof of Theorem \ref{th:CLT}, it is clear that the asymptotic behaviour of $\hat \nu^{CL}$ is governed by the asymptotic behaviour of the score function
		\begin{align}
			s_n(\nu_0) &=   \sum_{i=1}^n \left. \frac{\partial}{\partial \nu} \log f(X_{i\Delta}, X_{(i+1)\Delta}) \right|_{\nu = \nu_0}=   \sum_{i=1}^n \left.  \frac{1}{f(X_{i\Delta}, X_{(i+1)\Delta})} \frac{\partial}{\partial \nu} f(X_{i\Delta}, X_{(i+1)\Delta})\right|_{\nu = \nu_0}. \label{eq:S}
		\end{align}
		Let $B$ be a Borel set. Since $L' \sim Poi(\nu)$, we have
		\begin{align*}
			\mathbb{P}\left( L(B) = x \right) = \frac{[\nu Leb(B)]^x \exp(-\nu Leb(B))}{x!}, \quad x = 0, 1, \ldots.
		\end{align*}
		Hence,
		\begin{align*}
			\frac{\partial}{\partial \nu} \mathbb{P}\left( L(B) = x \right) = \left(x\nu^{-1} - Leb(B)\right) \mathbb{P}\left( L(B) = x \right), \quad x = 0, 1, \ldots.
		\end{align*}
		Using this, along with Proposition \ref{prop:positiveLB}, it is straightforward to show that
		\begin{align}
			\frac{1}{f(X_{(i+1)\Delta},X_{i\Delta})} \frac{\partial}{\partial \nu} f(X_{i\Delta}, X_{(i+1)\Delta})  &= \nu^{-1}(X_{(i+1)\Delta} + X_{i\Delta}) - C_{i} - D_i \label{eq:A}
		\end{align}
		where
		\begin{align}
			C_i &= \frac{\nu^{-1}}{f(X_{(i+1)\Delta},X_{i\Delta})} \sum_{c=0}^{\min (X_{i\Delta},X_{(i+1)\Delta})}  c \cdot P( L(A_{(i+1)\Delta}) = X_{(i+1)\Delta}, L(A_{i\Delta}) = X_{i\Delta} , (L(A_{i\Delta} \cap A_{(i+1)\Delta}) =  c) \nonumber \\
			&=  \nu^{-1} \sum_{c=0}^{\min (X_{i\Delta},X_{(i+1)\Delta})}  c \cdot P((L(A_{i\Delta} \cap A_{(i+1)\Delta}) =  c|L(A_{(i+1)\Delta}) = X_{(i+1)\Delta}, L(A_{i\Delta}) = X_{i\Delta})\nonumber \\
			&=\nu^{-1} \mathbb{E}[L (A_{i\Delta} \cap A_{(i+1)\Delta}) | X_{i\Delta}, X_{(i+1)\Delta}] \nonumber \\
			&=\nu^{-1} \mathbb{E}[ L(A_{i+k}) - L (A_{(i+1)\Delta} \setminus A_{i\Delta}) | X_{i\Delta}, X_{(i+1)\Delta}] \nonumber \\
			&=\nu^{-1} \mathbb{E}[ X_{(i+1)\Delta} - L (A_{(i+1)\Delta} \setminus A_{i\Delta}) | X_{i\Delta}, X_{(i+1)\Delta}] \nonumber \\
			&=\nu^{-1} \left( X_{(i+1)\Delta} -  \mathbb{E}[  L (A_{(i+1)\Delta} \setminus A_{i\Delta}) | X_{i\Delta}, X_{(i+1)\Delta}]  \right) \label{eq:C}
		\end{align}
		and
		\begin{align}
			D_i &= Leb(A_{(i+1)\Delta} \setminus A_{i\Delta}) +Leb(A_{i\Delta} \setminus A_{(i+1)\Delta}) +Leb(A_{i\Delta} \cap A_{(i+1)\Delta})\nonumber  \\
			&= Leb(A) + Leb(A_{(i+1)\Delta} \setminus A_{i\Delta}). \label{eq:D}
		\end{align}
		
		Define $U_i := \mathbb{E}[  L (A_{(i+1)\Delta} \setminus A_{i\Delta}) | X_{i\Delta}, X_{(i+1)\Delta}]$, $i = 1, 2, \ldots, n$. Recall that $\mathbb{E}[L(B)] = \nu Leb(B)$ from which we deduce that
		\begin{align*}
			\mathbb{E}[X_{i\Delta}] =  \nu Leb(A), \quad \textnormal{ and } \quad \mathbb{E}[U_i] =  \nu Leb(A_{(i+1)\Delta} \setminus A_{i\Delta}).
		\end{align*}
		Using the above and Equations \eqref{eq:A}--\eqref{eq:C} in Equation \eqref{eq:S}, we get 
		\begin{align*}
			s_n(\nu_0)  = \nu_0^{-1}  S_n(X) + \nu_0^{-1}  S_n(U),
		\end{align*}
		Using arguments similar to those in the proof of Theorem \ref{th:CLT}, the result now follows from Theorem \ref{th:partial}.

		%

	\end{proof}

	Before we provide the proof of  Theorem \ref{th:partial}, we will need the following two lemmas. For two sequences $a_n, b_n$, we will write ``$a_n \sim b_n$'' if $a_n/b_n$ tends to a non-zero constant as $n \rightarrow \infty$.
	\begin{lemma}\label{lem:LebB}
		Let Assumption \ref{ass:LM} hold and define the sets
		\begin{align}\label{eq:LebB}
			B_{i,k} = \{(x,s) : (i-1)\Delta \leq s \leq i\Delta , d(s-(k+1)\Delta) \leq x < d(s-k\Delta)\}, 
		\end{align}
		for $i = 1, 2, \ldots, n$ and $k \in \N$. 
		Then, for all $ i = 1, 2, \ldots, n$,
		\begin{align*}
			Leb(B_{i,k}) = \Delta^2d'(-(k-i+1+r_2)\Delta), 
		\end{align*}
		where $r_2=r_2(i,k) \in (-1,1)$. 
		In particular,
		\begin{align*}
			Leb(B_{1,N}) \sim g_2(N) N^{-H-2}, \quad N \rightarrow \infty,
		\end{align*}
		and
		\begin{align*}
			\sum_{k=1}^\infty k Leb(B_{1,k}) < \infty.
		\end{align*}
	\end{lemma}
	
	\begin{proof}[Proof of Lemma \ref{lem:LebB}]
		Fix $i \in \{1,2,\ldots, n\}$ and $k \geq 1$. Using the mean value theorem twice, we may write,
		\begin{align*}
			Leb(B_{i,k}) &= \int_{(i-1)\Delta}^{i\Delta}  \left( d(x-k\Delta) - d(x-(k+1)\Delta) \right)dx \\
			&=  \Delta\int_{(i-1)\Delta}^{i\Delta}   d'(x-(k+r_1)\Delta) dx \\
			&= \Delta[d(-(k+r_1-i)\Delta) - d(-(k+r_1-i+1)\Delta)] \\
			&= \Delta^2d'(-(k-i+1+r_2)\Delta), 
		\end{align*}
		where $r_1 \in (0,1)$ and $r_2 \in (-1,1)$, which proves the first part of the lemma. Note that $r_1$ and $r_2$ will in general depend on $k$ and $i$, but we will suppress this for ease of notation. The latter parts of the lemma now follow from Assumption \ref{ass:LM}. 
	\end{proof}

	\begin{lemma}\label{lem:LebB2}
		Let Assumption \ref{ass:LM} hold and define the sets
		\begin{align}\label{eq:LebB}
			\bar B_{1,k} = \cup_{j=k+1}^\infty B_{1,j}, \qquad k \in \N.
		\end{align}
		Then, 
		\begin{align*}
			Leb (\bar B_{1,N}) &= \int_0^\Delta d(s-(N+1)\Delta) ds \sim g_1(N) N^{-H-1},
		\end{align*}
		as $N \rightarrow \infty$.
		In particular, using also Lemma \ref{lem:LebB},
		\begin{align*}
			\sum_{k=1}^\infty k Leb(B_{1,k})^{1/2} Leb (\bar B_{1,k}) < \infty.
		\end{align*}
	\end{lemma}
	
	\begin{proof}[Proof of Lemma \ref{lem:LebB2}]
		Using a change-of-variables, Assumption \ref{ass:LM}, Karamata's theorem, the mean value theorem, and standard properties of slowly varying functions, we get
		\begin{align*}
			Leb (\bar B_{1,N}) &= \int_0^\Delta d(s-(N+1)\Delta) ds \\
			&= \int_{N\Delta}^{(N+1)\Delta} d(-y)dy \\
			&= \int_{N\Delta}^\infty g_1(y) y^{-H-1} dy - \int_{(N+1)\Delta}^\infty g_1(y) y^{-H-1} dy \\
			&\sim \frac{1}{H} g_1(N\Delta) (N\Delta)^{-H} -  \frac{1}{H} g_1((N+1)\Delta) ((N+1)\Delta)^{-H} \\
			&= \frac{\Delta^{-H}}{H} g_1(N\Delta) \left( N^{-H} -  \frac{g_1((N+1)\Delta)}{g_1(N\Delta)}  (N+1)^{-H} \right) \\
			&= \frac{\Delta^{-H}}{H} g_1(N\Delta) \left( N^{-H}-(N+1)^{-H} -  \left( \frac{g_1((N+1)\Delta)}{g_1(N\Delta)} -1\right)(N+1)^{-H} \right) \\
			&\sim \frac{\Delta^{-H}}{H} g_1(N\Delta) \left( N^{-H}-(N+1)^{-H} \right) \\
			&= \Delta^{-H} g_1(N\Delta) (N+r_1)^{-H-1} \\
			&\sim g_1(N\Delta) N^{-H-1}, \\
			&\sim g_1(N) N^{-H-1}, \\
		\end{align*}
		as $N \rightarrow \infty$, where $r_1 \in (0,1)$. The latter parts of the lemma now follow from Assumption \ref{ass:LM}. 
		%
	\end{proof}

	\begin{proof}[Proof of Theorem \ref{th:partial}]
		Suppose for simplicity that $L' \sim Poi(\nu)$ for some $\nu > 0$. The general case follows along similar lines to the proof in the Poisson case, using the compound Poisson representation of the IVT process as presented above.
		
		We follow the strategy of the proof outlined in Theorem 2 of \cite{DJLS2019}, adapted to our purposes. First, define the sequence of random variables
		\begin{align*}
			Z_i &=Z_{i\Delta}=L(A_{i\Delta} \setminus A_{(i-1)\Delta}) + \sum_{j=1}^{\infty} L((A_{i\Delta} \cap A_{(i-j)\Delta})\setminus A_{(i-j-1)\Delta} + (0,j\Delta)), \quad i = 1, 2, \ldots, n,
		\end{align*}
		where the sum converges almost surely by Kolmogorov's three-series theorem.
		Note that 
		\begin{align*}
			(A_{i\Delta} \cap A_{(i-j)\Delta})\setminus A_{(i-j-1)\Delta}=\{(x,s): (i-j-1)\Delta<s \leq (i-j)\Delta, 0\leq d(s-i\Delta)\},  
		\end{align*}
		and, hence, applying a time shift of $j\Delta$ leads to
		\begin{multline*}
			(A_{i\Delta} \cap A_{(i-j)\Delta})\setminus A_{(i-j-1)\Delta}+ (0,j\Delta)\\=\{(x,s): (i-1)\Delta<s \leq i\Delta, 0\leq d(s-(i+j)\Delta)\}. 
		\end{multline*}
		We observe that we can write $Z_i$ as a sum of independent variables,
		\begin{align*}
			Z_i &= \sum_{k=1}^{\infty} k L(B_{i,k}), 
		\end{align*}
		where $B_{i,k}$ are disjoint sets, given as in Equation \eqref{eq:LebB}. The construction of $Z_1$, and the decomposition into independent variables, are illustrated in Figure \ref{fig:Zillu}. Note furthermore that, by construction, the sequence $\{Z_i\}_{i=1}^{\infty}$ is iid. Let $S_n(\cdot)$ denote the de-meaned partial sums of a process, e.g. $S_n(X) = \sum_{i=1}^n(X_i - \mathbb{E}[X_i])$. Here we use the short-hand notation that $X_i=X_{i\Delta}$. Our aim is to show that
		\begin{enumerate}
			\item[(i)] $\mathbb{P}(Z_1 > y) \sim g(y) y^{-H-1}$, as $y \rightarrow \infty$, where $g$ is a slowly varying function.
			\item[(ii)] $\mathbb{E}[|S_n(X) - S_n(Z)|] = o(n^{1/(H+1)})$, as $n \rightarrow \infty$.
		\end{enumerate}
		Indeed, the result of the present theorem follows from (i) and (ii). To see this, note that (i) implies that \citep[][Theorem 2.6.7]{IL1971}
		\begin{align*}
			n^{-1/(H+1)} S_n(Z) \stackrel{(d)}{\rightarrow} L_{1+H}, \quad n \rightarrow \infty,
		\end{align*}
		where $L_{1+H}$ is an $(1+H)$-stable random variable with characteristic function as given in Equation \eqref{eq:chfct}. Secondly, (ii) now implies that
		\begin{align*}
			n^{-1/(H+1)} S_n(X) = n^{-1/(H+1)}  S_n(Z) + o_{\mathbb{P}}(1) \stackrel{(d)}{\rightarrow} L_{1+H}, \quad n \rightarrow \infty,
		\end{align*}
		which is what we wanted to show.
		
		We proceed to prove (i). First, let $Z^*$ be given by
		\begin{align*}
			Z^* = \sum_{k=1}^{\infty} k I(L(B_{1,k})> 0), 
		\end{align*}
		where $I(A)$ is the indicator of the set $A$. Note that $Z_1 \geq Z^* \geq 0$. Second, define $Z^{**} = Z_1-Z^*$. Let $g$ be a slowly varying function. As in  \cite{DJLS2019}, we prove that (a) $\mathbb{P}(Z^* > y) = g(y) y^{-H-1} + o( y^{-H-1})$ and (b) $\mathbb{P}(Z^{**} > y) = o(y^{-H-1})$ as $y \rightarrow \infty$, which allows us to deduce (i), i.e.~$\mathbb{P}(Z_1 > y) \sim g(y) y^{-H-1}$ as $y \rightarrow \infty$. 
		
		To prove (a), define the random variables
		\begin{align*}
			Z_1^{*} &:= \sum_{k=1}^{\infty} k I(L (B_{1,k})> 0)I(L (\bar B_{1,k}) =  0),  \\
			Z_2^{*} &:= \sum_{k=1}^{\infty} k I(L (B_{1,k})> 0)I(L (\bar B_{1,k}) =  1),
		\end{align*}
		where, as before,  
		\begin{align}\label{eq:LebB2}
			\bar B_{1,k} = \{(x,s) : 0 \leq s \leq \Delta , 0 \leq x < d(s-(k+1)\Delta)\} = \cup_{j=k+1}^{\infty} B_{1,j}.
		\end{align}
		The intuition is that $Z_1^{*}$ takes the value $k$ if there is an observation in $B_{1,k}$ but not in $B_{1,k+1}, B_{1,k+2}, \ldots$, while $Z_2^*$ takes the value $k$ if there is an observation in $B_{1,k}$ and precisely one observation in $B_{1,k+1}, B_{1,k+2}, \ldots$. Note that $Z^*  \geq  Z_1^{*} + Z_2^{*}  \geq 0$. Define, analogously to above $Z_3^{*} = Z^* - Z_1^* -Z_2^{*}$, which is $Z^*$ with the two largest ``points'' removed.\footnote{It can be shown that if $H \in (1/2,1)$, then we can set $Z_2^* = 0$, i.e., in this case, we only need to remove the largest point from $Z^*$ for the proof to go through.} We prove that (a') $\mathbb{P}(Z_1^{*} > y) \sim g(y) y^{-H-1}$ and (b') $\mathbb{P}(Z_2^{*} > y) = o(y^{-H-1})$ and  $\mathbb{P}(Z_3^{*} > y) = o(y^{-H-1})$ as $y \rightarrow \infty$, which allows us to deduce (a), i.e.~$\mathbb{P}(Z^* > y) \sim g(y) y^{-H-1}$ as $y \rightarrow \infty$. As noted in \cite{DJLS2019}, to prove (a'), we only need to prove  $\mathbb{P}(Z_1^{*} > N) \sim g(N) N^{-H-1} + o(N^{-H-1})$ for $N \rightarrow \infty$, where $N \in \mathbb{N}$. Let therefore $N \in \N$. We get, using Lemma \ref{lem:LebB2} and the fact that $L' \sim Poi(\nu)$,
		\begin{align*}
			\mathbb{P}(Z_1^{*} > N) &= \mathbb{P}\left(  \sum_{k=1}^{\infty} k I(L (B_{1,k})> 0)I(L (\bar B_{1,k}) =  0) > N \right) \\
			&=  \mathbb{P}\left( L (\bar B_{1,N}) >0 \right) \\
			&= 1- \mathbb{P}\left( L (\bar B_{1,N}) =0 \right) \\
			&= 1- e^{-\nu Leb (\bar B_{1,N})} \\
			&=  Leb (\bar B_{1,N}) + o(Leb (\bar B_{1,N})) \\
			&\sim g_1(N)N^{-H-1},
		\end{align*}
		as $ N \rightarrow \infty$.
		
		To prove (b'), note that, by Lemma \ref{lem:LebB2}, 
		\begin{align*}
			\mathbb{P}(Z_2^{*} > N) &= \mathbb{P}\left(  \sum_{k=1}^{\infty} k I(L (B_{1,k})> 0)I(L (\bar B_{1,k}) =  1) > N \right) \\
			&=  \mathbb{P}\left( L (\bar B_{1,N}) >1) \right) \\
			&= 1- \mathbb{P}\left( L (\bar B_{1,N}) =0 \right) -\mathbb{P}\left( L (\bar B_{1,N}) =1) \right) \\
			&= 1- e^{-\nu Leb (\bar B_{1,N})} - \nu Leb (\bar B_{1,N}) e^{-\nu Leb (\bar B_{1,N})} \\
			&\leq  \nu Leb (\bar B_{1,N})   - \nu Leb (\bar B_{1,N}) e^{-\nu Leb (\bar B_{1,N})} \\
			&=  \nu Leb (\bar B_{1,N}) (1  -e^{-\nu Leb (\bar B_{1,N})}) \\
			&\leq  \nu^2 Leb (\bar B_{1,N})^2 \\
			&\sim g_1(N)^2 N^{-2-2H} \\
			&= o(N^{-1-H}),
		\end{align*}
		as $N \rightarrow \infty$, by Lemma \ref{lem:LebB2}. To show that  $\mathbb{P}(Z_3^{*} > y) = o(y^{-H-1})$ as $y \rightarrow \infty$, we show that $Z_3^{*}$ is bounded in $L^2$. Let $\epsilon \in (0,H)$. Using Minkowski's inequality, the independence of $L (B_{1,k})$ and $L (\bar B_{1,k})$, and Lemmas \ref{lem:LebB} and \ref{lem:LebB2},  we get
		\begin{align*}
			\mathbb{E}[(Z_3^{*})^2]^{1/2} &= \mathbb{E}\left[ \left(Z^* -Z_1^{*}  -Z_2^{*} \right)^2 \right]^{1/2} \\
			&= \mathbb{E}\left[ \left( \sum_{k=1}^{\infty} k I(L (B_{1,k})> 0)\left(1 - I(L (\bar B_{1,k}) =  0) - I(L (\bar B_{1,k}) =  1) \right)\right)^2 \right]^{1/2} \\
			&\leq  \sum_{k=1}^{\infty} \mathbb{E}\left[  k^2 I(L (B_{1,k})> 0)^2\left(1 - I(L (\bar B_{1,k}) =  0) - I(L (\bar B_{1,k}) =  1) \right)^2 \right]^{1/2} \\
			&=  \sum_{k=1}^{\infty} k \mathbb{E}\left[  I(L (B_{1,k})> 0) I(L (\bar B_{1,k}) > 1) \right]^{1/2} \\
			&=  \sum_{k=1}^{\infty} k  \mathbb{P}(L (B_{1,k})> 0)^{1/2}  \mathbb{P}(L (\bar B_{1,k}) > 1)^{1/2} \\
			&=  \sum_{k=1}^{\infty} k  \left(1-\mathbb{P}(L (B_{1,k})= 0)\right)^{1/2}  \left(1-\mathbb{P}(L (\bar B_{1,k}) = 0) - \mathbb{P}(L (\bar B_{1,k}) = 1)\right)^{1/2} \\
			&=  \sum_{k=1}^{\infty} k  \left(1- \exp(-\nu Leb(B_{1,k}) \right)^{1/2}  \left(1-\exp(-\nu Leb(\bar B_{1,k}))- \nu Leb(\bar B_{1,k}) \exp(-\nu Leb(\bar B_{1,k}))\right)^{1/2} \\
			&\leq  \sum_{k=1}^{\infty} k  \nu^{1/2} Leb(B_{1,k})^{1/2}  \left(\nu Leb(\bar B_{1,k})- \nu Leb(\bar B_{1,k}) \exp(-\nu Leb(\bar B_{1,k}))\right)^{1/2} \\
			&=  \sum_{k=1}^{\infty} k \nu^{1/2} Leb(B_{1,k})^{1/2}  \nu^{1/2} Leb(\bar B_{1,k})^{1/2}  \left(1- \exp(-\nu Leb(\bar B_{1,k}))\right)^{1/2} \\
			&\leq    \nu^{3/2}  \sum_{k=1}^{\infty} k   Leb(B_{1,k})^{1/2}   Leb(\bar B_{1,k}) \\
			&<\infty,
		\end{align*}
		where we used, repeatedly, that $1-e^{-x} \leq x$.
		
		To prove (b), we show that $Z^{**}$ is bounded in $L^2$. We utilize the assumption that $L' \sim Poi(\nu)$, which implies that $L (B) = L(B)$ for all Borel sets $B$. This allows us to write
		\begin{align*}
			Z_1 	&=   \sum_{k=1}^{\infty} k L(B_{1,k}) \\
			&=  \sum_{k=1}^{\infty} k L(B_{1,k}) I(L(B_{1,k}) = 1)+  \sum_{k=1}^{\infty} k L(B_{1,k})I(L(B_{1,k}) > 1) \\
			&=  \sum_{k=1}^{\infty} k I(L(B_{1,k}) = 1)+  \sum_{k=1}^{\infty} k L(B_{1,k})I(L(B_{1,k}) > 1),
		\end{align*}
		and
		\begin{align*}
			Z_1^* &= \sum_{k=1}^{\infty} k I(L (B_{1,k})> 0) \\
			&= \sum_{k=1}^{\infty} k I(L (B_{1,k}) = 1) +  \sum_{k=1}^{\infty} k I(L (B_{1,k})> 1).
		\end{align*}
		We deduce that
		\begin{align*}
			Z_1^{**} 	&= Z_1-Z_1^* \\
			&=\sum_{k=1}^{\infty} k L(B_{1,k})I(L(B_{1,k}) > 1) -   \sum_{k=1}^{\infty} k I(L (B_{1,k})> 1) \\
			& = \sum_{k=1}^{\infty} kI(L(B_{1,k}) > 1) ( L(B_{1,k})-1)\\
			& \leq \sum_{k=1}^{\infty} kI(L(B_{1,k}) > 1) L(B_{1,k}).
		\end{align*}
		Using this and then Minkowski's inequality, we may write
		\begin{align*}
			\mathbb{E}[(Z_1^{**})^2]^{1/2} 	&\leq  \mathbb{E}[(\sum_{k=1}^{\infty} kI(L(B_{1,k}) > 1) L(B_{1,k}))^2]^{1/2} 	\\
			&\leq \sum_{k=1}^{\infty}  \mathbb{E}[k^2 I(L(B_{1,k}) > 1) L(B_{1,k})^2]^{1/2} \\
			&= \sum_{k=1}^{\infty} k \mathbb{E}[I(L(B_{1,k}) > 1) L(B_{1,k})^2]^{1/2} \\
			&= \sum_{k=1}^{\infty} k \mathbb{E}[ \sum_{j=2}^\infty (I(L(B_{1,k}) = j)) L(B_{1,k})^2]^{1/2} \\
			&= \sum_{k=1}^{\infty} k \left(\sum_{j=2}^\infty   \mathbb{E}[ I(L(B_{1,k}) = j) j^2] \right)^{1/2} \\
			&= \sum_{k=1}^{\infty} k \left(\sum_{j=2}^\infty   \mathbb{P}(I(L(B_{1,k}) = j) j^2 \right)^{1/2} \\
			&= \sum_{k=1}^{\infty} k \left(\sum_{j=2}^\infty  \frac{(\nu Leb(B_{1,k}))^j e^{-\nu Leb(B_{1,k})}}{j!} j^2 \right)^{1/2} \\
			&= \sum_{k=1}^{\infty} k \left(\nu^2 Leb(B_{1,k})^2 \sum_{j=2}^\infty  \frac{(\nu Leb(B_{1,k}))^{j-2} e^{-\nu Leb(B_{1,k})}}{j!} j^2 \right)^{1/2} \\
			&= \sum_{k=1}^{\infty} k \left(\nu^2  Leb(B_{1,k})^2 \sum_{j=0}^\infty  \frac{(\nu Leb(B_{1,k}))^{j} e^{-\nu Leb(B_{1,k})}}{(j+2)!} (j+2)^2 \right)^{1/2} \\
			&= \sum_{k=1}^{\infty} k \left(\nu^2  Leb(B_{1,k})^2 \sum_{j=0}^\infty  \frac{(\nu Leb(B_{1,k}))^{j} e^{-\nu Leb(B_{1,k})}}{j!} \frac{(j+2)}{(j+1)} \right)^{1/2} \\
			&\leq \sum_{k=1}^{\infty} k \left(\nu^2  Leb(B_{1,k})^2 \sum_{j=0}^\infty 2 \frac{(\nu Leb(B_{1,k}))^{j} e^{-\nu Leb(B_{1,k})}}{j!} \right)^{1/2} \\
			&=\sqrt{2} \nu \sum_{k=1}^{\infty} k  Leb(B_{1,k}) \\
			&<\infty,
		\end{align*}
		by Lemma \ref{lem:LebB}.
		
		We now prove (ii). Let $\kappa_1 = \mathbb{E}[L']$. Note first that 
		\begin{align*}
			\mathbb{E}[Z_1] &= \kappa_1\left(Leb(A_{i\Delta} \setminus A_{(i-1)\Delta}) + \sum_{j=1}^{\infty} Leb((A_{i\Delta} \cap A_{(i-j)\Delta})\setminus A_{(i-j-1)} + (0,j\Delta)) \right) \\
			&=  \kappa_1 Leb(A) \\
			&= \mathbb{E}[X_1],
		\end{align*}
		and
		\begin{align*}
			X_i =  L(A_{i\Delta} \setminus A_{(i-1)\Delta}) + \sum_{j=1}^{\infty} L(A_{i\Delta} \cap A_{(i-j)\Delta}), \quad i = 1,2,\ldots, n,
		\end{align*}
		see Figure \ref{fig:Zillu} for an illustration of these results. Now, following again \cite{DJLS2019}, we can use this to write
		\begin{align*}
			S_n(X) - S_n(Z) &= \sum_{i=1}^n X_i - \sum_{i=1}^n Z_i \\
			&=  \sum_{i=1}^n X_i - \sum_{i=1}^n Z_i \\
			&=  \sum_{i=1}^n \left( L(A_{i\Delta} \setminus A_{(i-1)\Delta}) + \sum_{j=1}^{\infty} L((A_{i\Delta} \cap A_{(i-j)\Delta})\setminus A_{(i-j-1)\Delta}) \right) \\
			&- \sum_{i=1}^n \left( L(A_{i\Delta} \setminus A_{(i-1)\Delta}) + \sum_{j=1}^{\infty} L((A_{i\Delta} \cap A_{(i-j)\Delta})\setminus A_{(i-j-1)\Delta} + (0,j\Delta)) \right) \\
			&=R_n' - R_n'',
		\end{align*}
		where 
		\begin{align*}
			R_n'  &=  \sum_{i=1}^n L(A_0 \cap A_{i\Delta}), 
		\end{align*}
		and
		\begin{align*}
			R_n'' &=  \sum_{i=1}^n \sum_{k=1}^\infty k L(B_{n+1-i,i+k}).
		\end{align*}
		Let  $\epsilon \in (0,H^2/2)$.  By Karamata's Theorem, Assumption \ref{ass:LM}, and properties of slowly varying functions, we have
		\begin{align*}
			\mathbb{E}[R_n'] &= \mathbb{E}[L'] \sum_{k=1}^n Leb(A_0 \cap A_{k\Delta}) \\
			&= \kappa_1  \sum_{k=1}^n \int_{k\Delta}^\infty d(-x) dx \\
			&\sim   \sum_{k=1}^n g_1(k)k^{-H} \\
			&\sim O(n^{1-H+\epsilon}),
		\end{align*}
		as $n \rightarrow \infty$. Note that since $\epsilon<\frac{H^2}{2}$, then $1-H+\epsilon < \frac{1}{H+1}$, proving that $\mathbb{E}[R_n']  = o(n^{1/(H+1)})$. Similarly, using Lemma \ref{lem:LebB2},
		\begin{align*}
			\mathbb{E}[R_n''] 	&= \kappa_1 \sum_{i=1}^n \sum_{k=1}^\infty k Leb(B_{n+1-i,i+k}) \\
			&= \kappa_1 \sum_{i=1}^n \sum_{k=1}^\infty k Leb(B_{1,i+k}) \\
			&\sim O(n^{1-H+\epsilon}),
		\end{align*}
		as $n \rightarrow \infty$. This concludes the proof.
	\end{proof}

	\begin{figure}[h!] 
		\centering 
		\includegraphics[trim={5cm 17cm 2cm 5cm},clip, width=0.95\columnwidth]{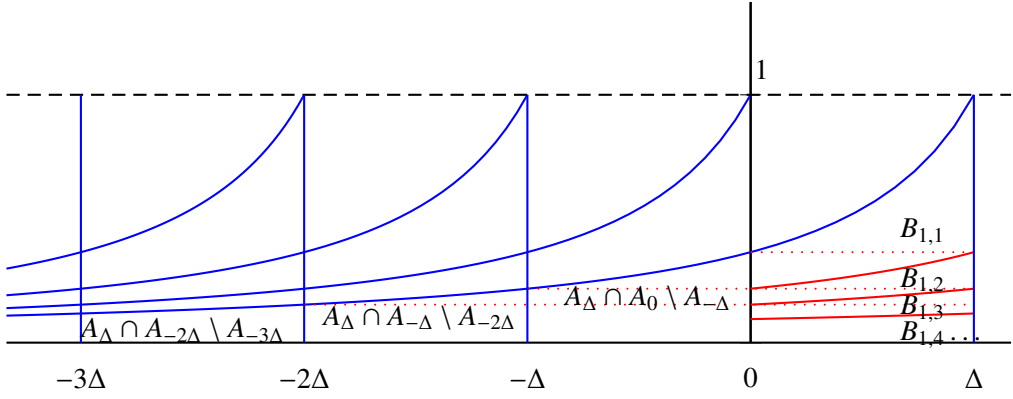} 
		\caption{\it The construction of $Z_1=L(A_{\Delta} \setminus A_0) + \sum_{j=1}^{\infty} L(A_{\Delta} \cap A_{(1-j)\Delta} \setminus A_{(-j) \Delta}+ (0,j\Delta))$, and the decomposition into independent variables, $Z_1 = \sum_{k=1}^{\infty} k L(B_{1,k})$.
			}
		\label{fig:Zillu}
	\end{figure}

	\begin{proof}[Proof of Lemma \ref{lem:fCond}]
		Using Bayes' theorem and the unconditional independence of $L(A_t \cap A_{t+h})$  and $L(A_t\setminus A_{t+h})$, we have for  $x \in \mathbb{N}\cup \{0\}$ and $l \in \{0, 1, \ldots, x\}$:
		\begin{align*}
			&\PP(L(A_t \cap A_{t+h})=l|X_t=x)=\frac{\PP(X_t=x|L(A_t \cap A_{t+h})=l)\PP(L(A_t \cap A_{t+h})=l)}{\PP(X_t=x)}
			\\
			&=\frac{\PP(L(A_t \cap A_{t+h})+L(A_t\setminus A_{t+h})=x|L(A_t \cap A_{t+h})=l)\PP(L(A_t \cap A_{t+h})=l)}{\PP(X_t=x)}\\
			&=\frac{\PP(L(A_t\setminus A_{t+h})=x-l|L(A_t \cap A_{t+h})=l)\PP(L(A_t \cap A_{t+h})=l)}{\PP(X_t=x)}
			\\
			&=\frac{\PP(L(A_t\setminus A_{t+h})=x-l)\PP(L(A_t \cap A_{t+h})=l)}{\PP(X_t=x)}.
		\end{align*}
	\end{proof}

	\begin{proof}[Proof of Proposition \ref{prop:fCond}]
		Using the conditional law of total probability we obtain the following convolution formula
		\begin{align*}
			&\PP(X_{t+h}=x_{t+h}|X_t=x_t)
			=\PP(L(A_t \cap A_{t+h})+L(A_{t+h}\setminus A_t)=x_{t+h}|X_t=x_t)
			\\
			&=\sum_{c=0}^{\min(x_t,x_{t+h})}
			\PP(L(A_t \cap A_{t+h})+L(A_{t+h}\setminus A_t)=x_{t+h}|X_t=x_t, L(A_t\cap A_{t+h})=c) \\
			&\cdot \PP(L(A_t\cap A_{t+h})=c|X_t=x_t)\\
			&=\sum_{c=0}^{\min(x_t,x_{t+h})}
			\PP(L(A_{t+h}\setminus A_t)=x_{t+h}-c|X_t=x_t, L(A_t\cap A_{t+h})=c)
			\PP(L(A_t\cap A_{t+h})=c|X_t=x_t)\\
			&=\sum_{c=0}^{\min(x_t,x_{t+h})}
			\PP(L(A_{t+h}\setminus A_t)=x_{t+h}-c)
			\PP(L(A_t\cap A_{t+h})=c|X_t=x_t).
		\end{align*}
		
	\end{proof}

	\newpage \clearpage
	
	\setcounter{section}{0}
	\renewcommand{\thesection}{S\arabic{section}}
	\renewcommand{\thesubsection}{S\arabic{section}.\arabic{subsection}}
	\setcounter{equation}{0}
	\setcounter{figure}{0}
	\setcounter{table}{0}
	\setcounter{page}{1}
	\makeatletter
	\numberwithin{equation}{subsection}
	
	\renewcommand{\theequation}{S\arabic{section}.\arabic{equation}}
	\renewcommand{\thefigure}{S\arabic{figure}}
	\renewcommand{\thetable}{S\arabic{table}}
	
	\begin{center}
		\textbf{\large Supplemental Materials: Inference and forecasting for continuous-time integer-valued trawl processes}
	\end{center}
	\section{Introduction}
	
	This document
	is structured as follows.
	\begin{itemize}
		\item Section \ref{sec:vis} presents a decomposition of trawl processes. 
		\item 
		
		Section \ref{sec:mm} provides the details of a method of moment-based estimation of integer-valued trawl processes so far used in the literature. 
		\item	Section \ref{app:stdErr} describes  how the composite likelihood estimator and the corresponding asymptotic covariance matrices can be computed in practice. 
		\item Section \ref{app:f_derive} presents an 
		expression for the pairwise likelihood for integer-valued processes (not restricted to count data) and discusses a simulation unbiased estimator for the composite likelihood.

		\item Section \ref{sec:SuppSim} contains additional details for the simulation study:  
		Subsection \ref{sec:detailssim} describes the simulation setup for the simulation study reported in the main article. 
		Subsections \ref{sec:emp_est} and \ref{sec:fs_ic} present the finite sample results of the MCL estimator and of the model selection procedure, respectively. 
		Section \ref{sec:setup2} repeats the simulation study using different parameter choices.  
		\item 
		Section \ref{sec:SuppEmp} contains additional details on the empirical study, including details on the  data pre-processing, see Subsection \ref{sec:EmpSetUp}, and additional forecasting results, see Subsection \ref{sec:add_forec}.
		
		\item Sections \ref{sec:levyBases} and  \ref{sec:trawls} present details on various parametric L\'evy bases and trawl functions, respectively. These structures might be used in the construction of IVT processes, as illustrated in the main paper. 
		\item Section \ref{sec:gradients} contains details on how to calculate the gradients for the log-composite likelihood functions implied by most of the parametric IVT processes. These calculations are straightforward to make (although somewhat tedious) and can rather easily be made for other IVT specifications than those considered here. The gradients can be used in the numerical optimization of the composite likelihood functions and are also crucial for implementing the asymptotic theory presented in the main paper. In particular, to estimate the asymptotic variance matrix $V(\theta)$, it is necessary to  evaluate the gradient at $\hat \theta^{CL}$. 
		\item  Section \ref{sec:add_calcs} contains some additional technical calculations. 
		\item Section \ref{sec:WeakDepGMM} shows that IVT processes are $\theta$-weakly dependent and presents the asymptotic theory for the  GMM estimation of the parameters. Section \ref{sec:AVARcomp} contains an analytical comparison of the asymptotic variance of the MCL and GMM estimators in the case of the Poisson-Exponential IVT process.
		\item Lastly, Section \ref{sec:code} contains brief details on the software packages accompanying the main paper. In particular, we supply software for simulation, estimation (including inference), model selection, and forecasting of IVT processes.
	\end{itemize}
	\clearpage
	
	\newpage
	
	\section{Trawl process decomposition}\label{sec:vis}
	%
	%
	
	When deriving theoretical results for trawl processes, we typically decompose trawl sets at different time points into a partition of disjoint sets. We note that, given a trawl function $d,$ it holds that
	\begin{align} \label{eq:rule1}
		Leb(A_t) = Leb(A) = \int_{-\infty}^0 d(s) ds.
	\end{align}
	It is also useful to note that for $t\geq s$,
	\begin{align}\label{eq:rule2}
		Leb(A_t \cap A_s) =  Leb(A_{t-s} \cap A)  = \int_{-\infty}^{-(t-s)} d(s) ds,
	\end{align}
	and
	\begin{align}\label{eq:rule3}
		Leb(A_t \setminus  A_s)  = Leb(A_s \setminus  A_t) = L(A) -  Leb(A_t \cap A_s).
	\end{align}
	Thus, given a trawl function $d$, it is straightforward to calculate $Leb(A), Leb(A_t \cap A_s)$, and $L(A_t \setminus A_s)$ for all $t\geq s\geq 0$.
	
	Also, we can write, for $0\leq s \leq t$, 
	\begin{align*}
		X_s&=L(A_s)=L(A_s \cap A_{t})+L(A_s\setminus A_{t}),\\ X_{t}&=L(A_{t})=L(A_s \cap A_{t})+L(A_{t}\setminus A_{s}),
	\end{align*}
	where the three random variables $L(A_s \cap A_{t}), L(A_s\setminus A_{t}), L(A_{t}\setminus A_{s})$ are independent since the corresponding sets are disjoint. 
	In Figure \ref{fig:trawl_decomp}, we illustrate such a decomposition for $s=3, t=4$. 
	
	\begin{figure}[htb]
		\centering
		\includegraphics[scale=0.75]{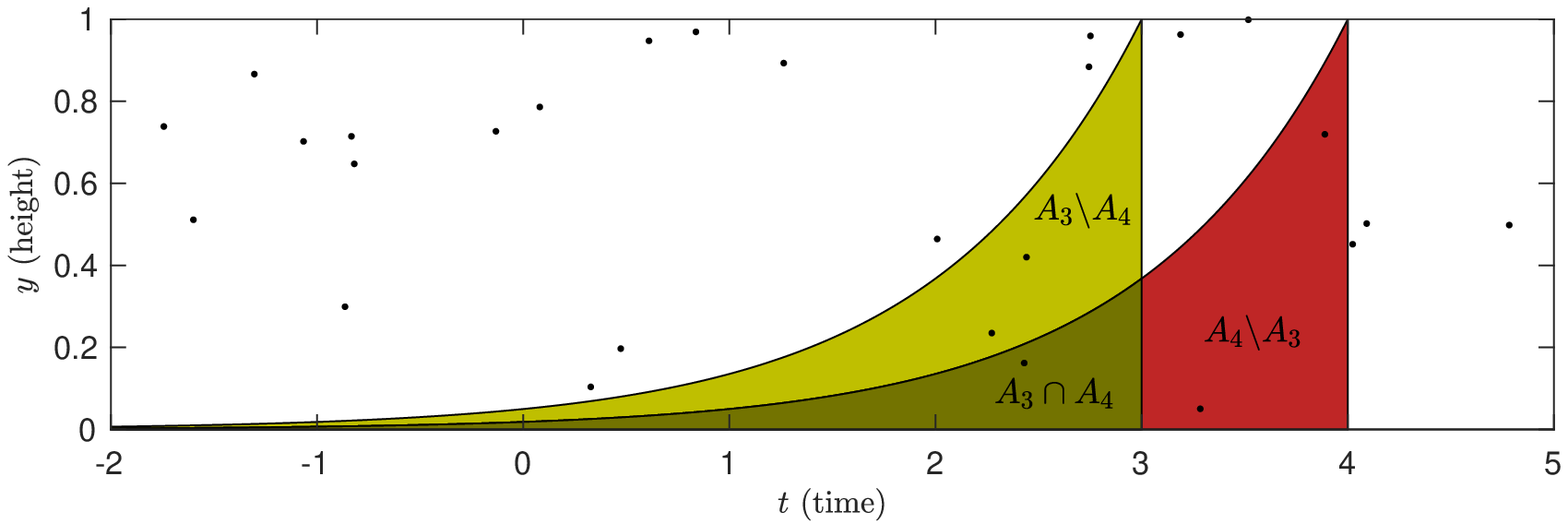} 
		\caption{\it  Decomposition of exponential trawl sets at times $s=3, t=4$, super-imposed on a Poisson L\'evy basis, with $\nu = 5$, on $\mathbb{R} \times [0,1]$. The exponential trawl parameter is set to $\lambda = 1$.}
		\label{fig:trawl_decomp}
	\end{figure}
	
	\clearpage \newpage
	
	\section{Method-of-moments-based estimation of IVT processes}\label{sec:mm}
	For a parametric IVT model, let  the parameter vector $\theta$ of the model be given by $\theta = (\theta_d, \theta_L)$, where $\theta_d$ contains the parameters governing the trawl function $d$ and $\theta_L$ contains the parameters governing the marginal distribution of the IVT process, as specified by the underlying L\'evy seed $L'$. For instance, in the case of the Poisson-Exponential IVT process considered above, cf. Figure \ref{fig.:PRMtrawl}, we would have $\theta_d = \lambda$ and $\theta_L = \nu$.  This section discusses how $\theta_d$ and $\theta_L$ can be estimated in a two-step procedure using a method-of-moments procedure. It is this procedure that has been used so far in most applied work on IVT processes. Note that we develop the asymptotic theory for the full (one-step) GMM estimation of $\theta$ in Section \ref{sec:WeakDepGMM} below.
	
	Because the correlation structure  of an IVT process is decoupled from its marginal distribution, the theoretical autocorrelation function of the process will not depend on $\theta_L$, and we can thus estimate $\theta_d$ in the first step, using the empirical autocorrelations of the data. To be precise, let $\rho_{\theta_d}(h)$ be the parametric autocorrelation function as implied by the trawl function $d$, see Equation \eqref{eq.:corr}, and let $\hat{\rho}(k)$ be the estimate of the empirical autocorrelation of the data at lag $k$. The GMM estimator of $\theta_d$ is 
	\begin{align}\label{eq:MM}
		\hat{\theta}_d^{GMM} := \arg \min_{\theta_d \in \Theta_d} \sum_{k=1}^{K} \left( \rho_{\theta_d}(k) - \hat{\rho}(k) \right)^2,
	\end{align}
	where $K \geq 1$ denotes the number of lags to include in the estimation and $\Theta_d$ is the parameter space of the trawl parameters in $\theta_d$.\footnote{\label{foot:lam}In the case of the IVT process with an exponential trawl function $d(s) = \exp(\lambda s)$, $s\leq 0$, we use a closed-form estimator of the $\lambda$ parameter using only the autocorrelation function calculated at the first lag, that is $\hat \lambda^{GMM} = - \log \hat \rho(1)/\Delta$, where $\Delta>0$ is the equidistant time between observations.}
	
	For the estimation of the parameters $\theta_L$ governing the marginal distribution of the IVT process, recall that the $j$th cumulant $\kappa_j$ of $X_t$ is given by
	\begin{align*}
		\kappa_j = Leb(A) \cdot \kappa_j', \quad j = 1, 2, \ldots,
	\end{align*}
	where $\kappa_j'$ is the $j$th cumulant of the L\'evy seed $L'$.  Using the estimates $\hat{\theta}_d^{GMM}$ from the first step, we can estimate the Lebesgue measure of the trawl set as
	\begin{align}\label{eq:estLebA}
		\widehat{Leb(A)} = \int_0^{\infty} \hat{d}(-s) ds, 
	\end{align}
	where $\hat{d}(\cdot)$ denotes the estimate of the trawl function implied by the estimated trawl parameters  $\hat{\theta}_d^{GMM}$. The parameters governing the marginal distribution, $\theta_L$, can now be estimated as follows. Let $r$ be the number of elements in $\theta_L$ and denote by $\widehat{Leb(A)}$ the estimate of the Lebesgue measure obtained from \eqref{eq:estLebA}. Estimates of the cumulants, $\hat{\kappa}_j$, can be obtained straightforwardly  by calculating the empirical cumulants of the data. Let $I_r$ be a set of $r$ distinct natural numbers (e.g., the numbers from $1$ to $r$). Now
	\begin{align*}
		\hat{\kappa}_j = \widehat{Leb(A)} \cdot \kappa_j', \quad j \in I_r,
	\end{align*}
	defines $r$ equations in the $r$ unknowns $\theta_L$. GMM  estimates of the elements in $\theta_L$, $\hat{\theta}_L^{GMM}$, can be obtained by solving these $r$ equations. Finally, set $\hat{\theta}^{GMM} := (\hat{\theta}_d^{GMM},\hat{\theta}_L^{GMM})$, which is the method-of-moments-based estimator of $\theta$.

	\clearpage \newpage

	\section{Practical details on feasible inference using the MCL estimator}\label{app:stdErr}
	As shown in Theorem \ref{th:CLT} of Section \ref{sec:asym}, the asymptotic variance of the maximum composite likelihood estimator, $\hat \theta_{CL}$ is given by the inverse Godambe information matrix,
	\begin{align*}
		G(\theta_0)^{-1} = H(\theta_0)^{-1} V(\theta_0) H(\theta_0)^{-1}.
	\end{align*}
	As mentioned, the matrices $H(\theta_0)$ and $V(\theta_0)$ can be consistently estimated by
	\begin{align*}
		\hat{H}(\hat{\theta}^{CL} ) &=  -\frac{1}{n} \frac{\partial}{\partial \theta \partial\theta'}l_{CL}(\hat \theta^{CL};x) \\
		&=  -\frac{1}{n}  \sum_{k=1}^K  \sum_{i=1}^{n-k} \frac{\partial^2}{\partial \theta \theta'} \log f(x_{(i+k)\Delta},x_{i\Delta};\theta)|_{\theta=\hat{\theta}^{CL}}  , \\
		\hat{V}_{HAC}(\hat{\theta}^{CL} ) &= \hat{\Sigma}_0 + \sum_{j=1}^q \left(1 - \frac{j}{q+1}\right)(\hat{\Sigma}_j + \hat{\Sigma}_j'),
	\end{align*}
	where
	\begin{align*}
		\hat{\Sigma}_j := \frac{1}{n}  \sum_{k=1}^K\sum_{k'=1}^K \sum_{i = 1}^{n-j-k'} \frac{\partial}{\partial \theta} \log f(x_{(i+k)\Delta},x_{i\Delta};\theta)|_{\theta=\hat{\theta}^{CL} } \frac{\partial}{\partial \theta'}\log f(x_{(i+j+k')\Delta},x_{(i+j) \Delta};\theta))|_{\theta=\hat{\theta}^{CL} },  
	\end{align*}
	and $q \in \mathbb{N}$ is the number of autocorrelation terms to take into account in the HAC estimator. The Hessian, $H(\theta_0)$, is straightforwardly estimated by the above expression. Indeed, a numerical approximation of this matrix is often directly available as output from the software maximizing the composite likelihood function. We have found that while this estimator $\hat{H}(\hat{\theta}^{CL} )$ is quite precise, the  HAC estimator $\hat{V}_{HAC}(\hat{\theta}^{CL} )$ can be rather imprecise. In practice, we therefore recommend estimating $V(\theta_0)$ using  simulation-based approach; the details are given in the following  
	\ref{app:B1}.
	
	\subsection{Simulation-based approach to estimating the asymptotic covariance matrix}\label{app:B1}
	To obtain a simulation-based estimator of $V(\theta_0)$, let $B$ denote a positive integer (e.g. $B = 500$) and suppose that $\hat \theta^{CL}$ is the maximum composite likelihood estimate of $\theta$ from \eqref{eq:clmax} when applied to the original data. To estimate $V(\theta_0)$, do as follows:
	\begin{enumerate}
		\item \label{st:1}
		For $b = 1, 2, \ldots, B$, simulate $N$ observations of a trawl process $X^{(b)} = \{X^{(b)}_i\}_{i=1}^N$ with underlying parameters $\hat \theta^{CL}$.
		\item
		For $b = 1, 2, \ldots, B$, use the simulated data $X^{(b)}$ to calculate $s^{(b)}(\hat \theta^{CL}) = N^{-1/2} \frac{\partial}{\partial\theta} l_{CL}( \hat \theta^{CL};X^{(b)})$. The gradient can either be calculated numerically or analytically.\footnote{The Supplementary Material contains analytical expressions for the gradients implied by the various parametric specifications considered in this paper.} Note that the bootstrap data is used to calculate the gradient, but the parameter vector $\hat \theta^{CL}$ is the original estimator obtained from the initial (real) data set. 
		\item
		Estimate $V(\hat \theta^{CL})$ as the sample covariance matrix of the simulated scores  $\left\{ s^{(b)}(\hat \theta^{CL})\right\}_{b=1}^{B}$.
	\end{enumerate}
	Note that the number of simulated observations, $N$, in Step \ref{st:1} does not need to equal the number of observations in the original data set, $n$. When $n$ is large, setting $N=n$ can be computationally costly; we found that setting $N = 500$ or even $N = 100$ provided good results. In our simulation study and in the empirical application we have set $B = N = 500$.
	
	\subsection{Feasible inference using parametric bootstrap methods}\label{app:B2}
	It is, of course, also possible to side-step the estimation of $G(\theta)^{-1}$ entirely, by considering a ``standard'' parametric bootstrap approach where the asymptotic variance of the MCL estimator is approximated directly by applying the MCL estimator to $B$ bootstrap samples of IVT processes simulated with parameters $\hat \theta^{CL}$. While mechanically simpler to implement than the simulation-based procedure suggested above in Section \ref{app:B1}, such a parametric bootstrap approach will often be more computationally demanding, since one needs to apply the numerical optimization of the composite likelihood function in each bootstrap replication, whereas the approach suggested in Section \ref{app:B1} only requires the evaluation of the composite likelihood once for each $b = 1, 2, \ldots, B$. Further, the ``standard'' bootstrap only delivers standard errors of $\hat \theta^{CL}$ and can not be used for calculating information criteria (Section \ref{sec:IC}). For these reasons, we do not consider the standard parametric bootstrap approach in this paper. The interested reader should have no trouble implementing it, however.
	
	
	\clearpage \newpage

	\section{Pairwise likelihood for IVT processes and simulation-based likelihood}\label{app:f_derive}

	Proposition \ref{prop:positiveLB} presented a simple expression for the pairwise PMFs $f(x_{i+k},x_{i};\theta)$ in the case where the underlying L\'evy basis is positive. In the general case, i.e.~where the L\'evy basis is integer-valued, we have, by the law of total probability, that 
	\begin{align}
		f(x_{i+k},&x_{i};\theta) := \PP_{\theta}\left( X_{(i+k)\Delta} = x_{i+k}, X_{i\Delta} = x_{i}\right) \nonumber \\
		=& \sum_{c=-\infty}^{\infty}  \PP_{\theta}\left( X_{(i+k)\Delta} = x_{i+k}, X_{i\Delta} = x_{i}| L(A_{(i+k)\Delta} \cap A_{i\Delta}) = c \right)\cdot \PP_{\theta}\left(L(A_{(i+k)\Delta} \cap A_{i\Delta}) = c\right) \nonumber\\
		=& \sum_{c=-\infty}^{\infty}  \PP_{\theta}\left( L(A_{(i+k)\Delta} \setminus A_{i\Delta}) = x_{i+k} - c\right) \PP_{\theta}\left( L(A_{i\Delta} \setminus A_{(i+k)\Delta}) = x_{i}-c\right) \label{eq:jointDens} \\
		&\cdot \PP_{\theta}\left(L(A_{(i+k)\Delta} \cap A_{i\Delta}) = c\right), \nonumber
	\end{align}
	which can be used for calculating the pairwise likelihood as a function of the parameter vector $\theta$, in the general integer-valued case. When implementing this result in practice, one could truncate the sum in \eqref{eq:jointDens} according to some criterion in order to approximate the joint PMF. Truncation can be avoided by resorting to a simulation-based approach, however. The following proposition shows that a simulation unbiased version of the joint PMF exists and that the simulation is, in fact, easy to perform.
	\begin{proposition}\label{prop:unbiasedSims}
		Let $t, s\geq 0$, choose $M \in \N$ and let $C^{(j)} \sim L(A_t \cap A_s)$, $j = 1, 2, \ldots, M$, be an iid sample. Then
		\begin{align*}
			\hat{f}(x_{t},x_{s}; \theta) = \frac{1}{M}\sum_{j=1}^M \PP_{\theta}(L(A_t \setminus A_s) = x_t - C^{(j)})\PP_{\theta}(L(A_s \setminus A_t) = x_s - C^{(j)})
		\end{align*}
		is a simulation-based unbiased estimator of $f(x_t,x_s;\theta)$. We further note that the simulation error
		\begin{align*}
			\hat{f}(x_{t},x_{s}; \theta)  - f(x_{t},x_{s}; \theta) 
		\end{align*}
		is, conditional on $x$, stochastically independent for different values of $t$ and $s$. Also, this error converges to zero at rate $\sqrt{M}$ as long as
		\begin{align*}
			\sum_{c=-\infty}^{\infty} f(x_t| c;\theta)^2 f(x_s| c;\theta)^2 f(x;\theta) <  \infty,
		\end{align*}
		where $f(c;\theta) = \PP_{\theta}(C^{(1)} = c)$ denotes the PMF of $C^{(j)}$, $j = 1, 2, \ldots, M$, and $f(x_t|c;\theta) = \PP_{\theta}(X_t = x_t|L(A_t \cap A_s) = c)$ denotes the conditional PMF of $X_t$. 
	\end{proposition}
	\begin{proof}[Proof of Proposition \ref{prop:unbiasedSims}]
		Ignoring the dependence on $\theta$ we have
		\begin{align*}
			f(x_t, x_s) = \sum_{c=-\infty}^{\infty} f(x_t,x_s|c) f(c) = \sum_{c=-\infty}^{\infty} f(x_t|x_s,c) f(x_s|c) f(c).
		\end{align*}
		Since, conditionally on $L(A_t \cap A_s)=c$, $X_t$ and $X_s$ are independent we have $f(x_t|x_s,c) = f(x_t|c)$ and thus
		\begin{align*}
			f(x_t, x_s)  = \sum_{c=-\infty}^{\infty} f(x_t|c) f(x_s|c) f(c),
		\end{align*}
		which shows that sampling from $f(c;\theta)$ delivers the quantity we need. The rest of the proposition is obvious.
	\end{proof}

	Proposition \ref{prop:unbiasedSims} shows that the simulated CL function is an unbiased estimator of the true CL function. In other words, if we let $U$ denote the vector of uniform random variables behind the simulation of $\{C^{(j)}\}_{j=1}^M$ and define
	\begin{align*}
		\log \mathcal{L}_U(\theta;x,u)  = \log \mathcal{L}^{(K)}_U(\theta;x,u) := \sum_{k = 1}^K  \sum_{i=1}^{n-k} \log \hat{f}(x_{i+k},x_i;\theta),
	\end{align*}
	then $\mathcal{L}_U(\theta;x,u)$ is a simulation unbiased estimator for the composite likelihood $\mathcal{L}_{CL}(\theta;x)$. That is
	\begin{align*}
		\mathcal{L}_{CL}(\theta;x) = \int \mathcal{L}_U(\theta;x,u) f_U(u)du,
	\end{align*}
	where $f_U(u) \propto 1$ is the joint density of the uniform random numbers behind all the simulations. It is well known that numerically optimizing a simulated likelihood function \citep[the so-called simulated maximum likelihood approach, see e.g.][]{LermanManski81} suffers a number of drawbacks and can be fragile in practice \citep[e.g.][]{NeilFlury11}. However, as a result of the seminal \cite{PMCMC}, it is feasible to do Markov Chain Monte Carlo (MCMC) when one can unbiasedly simulate the likelihood. As a consequence, it is feasible to perform simulation-based estimation through MCMC, instead of relying on numerical optimization. From an estimation viewpoint, this can be an attractive approach \citep{NeilFlury11}.

	\clearpage
	\newpage
	\section{Simulation study}\label{sec:SuppSim}
	In this section, we provide additional information on the simulation study. We first describe the simulation set-up used for the study which we summarise in the main article in Subsection \ref{sec:detailssim}.
	\subsection{Additional details on the simulation study reported in the main article}\label{sec:detailssim}

	In a simulation study, we examine the finite sample properties of the composite likelihood-based estimation procedure
	and the model selection procedure. 
	
	The IVT framework is very flexible and there are many possible choices of data-generating processes (DGPs) to use in the simulation studies. Here, we will consider
	the six combinations of the two marginal distributions, given in Examples \ref{ex:Poisson} and \ref{ex:NB}, and three correlations structures, given in Examples \ref{ex:Exp},  \ref{ex:IG}, and \ref{ex:GAM}. In other words, we consider 
	the Poisson-Exponential (P-Exp), the Poisson-Inverse Gaussian (P-IG), the Poisson-Gamma (P-Gamma), the Negative Binomial-Exponential (NB-Exp), the Negative Binomial-Inverse Gaussian (NB-IG), and the Negative Binomial-Gamma (NB-Gamma)  IVT models. Note that the first model contains two free parameters, the second, third, and fourth models three free parameters, and the fifth and sixth models four free parameters. Since the L\'evy bases considered here are non-negative-valued, we will use Proposition \ref{prop:positiveLB} for the calculation of the pairwise likelihoods. 
	
	The parameter values used in the simulation studies are given in Table \ref{tab:paramTab} and the implied marginal distributions and autocorrelation structures are shown in Figure \ref{fig:simEx}. The figure illustrates the difference between the six DGPs: those based on the Poisson L\'evy basis have a more concentrated marginal distribution compared to those based on the Negative Binomial L\'evy basis; those based on the exponential trawl function have smaller degrees of autocorrelation (memory) than those based on the Inverse Gaussian trawl function, and the Gamma trawl function can exhibit still greater autocorrelation.

	The choice of parameter values used in the simulation studies below and given in Table  \ref{tab:paramTab} are based on the estimates obtained in the empirical study in Section \ref{sec:emp}. We have found the finite sample properties of the methods proposed in this paper to be relatively robust to the exact choice of parameter values.

	\begin{table}
\caption{\it Parameter values used in simulation studies}
\begin{center}
\footnotesize
\begin{tabularx}{1.00\textwidth}{@{\extracolsep{\stretch{1}}}lcccccccc@{}} 
\toprule
DGP        &  $\nu$  & $m$           & $p$      & $\lambda$ & $\delta$ & $\gamma$ & $H$ & $\alpha$ \\
\cmidrule{2-9}
P-Exp      & $17.50$      &                  &              & $1.80$            &               & & & \\
P-IG        & $17.50$      &                  &              &                   &  $1.80$  & $0.80$ & & \\
P-Gamma  & $17.50$      &                  &              &              & &     &  $1.70$  & $0.80$ \\
NB-Exp   &            & $7.50$   & $0.70$  & $1.80$            &               & & & \\
NB-IG     &            &   $7.50$  & $0.70$  &                   &  $1.80$ & $0.80$ & & \\
NB-Gamma     &            &   $7.50$  & $0.70$  &          & &         &  $1.70$ & $0.80$ \\
\bottomrule 
\end{tabularx}
\end{center}
{\footnotesize \it Parameter values for the six different DGPs used in the simulation studies of Section \ref{sec:MC}. See Examples   \ref{ex:Poisson}, \ref{ex:NB}, \ref{ex:Exp}, \ref{ex:IG}, and \ref{ex:GAM} for details. The value $\nu = mp/(1-p)$  with $m = 7.5$ and $p = 0.70$ is chosen such that the first moment of the Poisson and Negative Binomial L\'evy bases are matched. Marginal distributions and autocorrelation functions implied by these parameter values are shown in Figure \ref{fig:simEx}.}
\label{tab:paramTab}
\end{table}
	
	\begin{figure}[!t]
		\centering
		\includegraphics[scale=0.75]{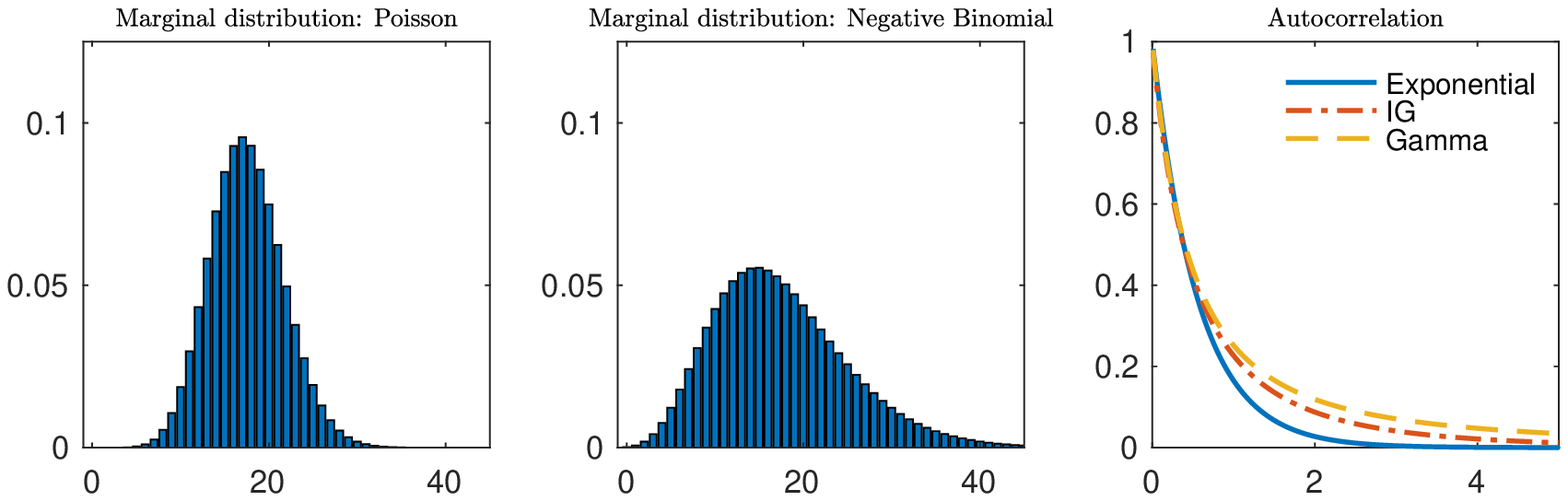} 
		\caption{\it  Marginal distributions of the L\'evy bases and autocorrelations of the DGPs used in the simulation studies of Section \ref{sec:MC}. The marginal distribution and autocorrelation structure of IVT processes can be specified independently, resulting in six different DGPs in this setup (P-Exp, P-IG, P-Gamma, NB-Exp, NB-IG, NB-Gamma). See Examples  \ref{ex:Poisson}, \ref{ex:NB}, \ref{ex:Exp}, \ref{ex:IG}, and  \ref{ex:GAM}  for details. The parameter values used to produce the plots are given in Table \ref{tab:paramTab}.}
		\label{fig:simEx}
	\end{figure}
	
	\clearpage
	
	\subsection{Finite sample properties of the MCL estimator}\label{sec:emp_est}
	Consider $n$ equidistant observations of an IVT process on an equidistant grid of size $0.10$, i.e.~$X_{\Delta}, X_{2\Delta}, \ldots, X_{n\Delta}$ with $\Delta = 0.10$. We simulate $500$ Monte Carlo replications of such time series and in each iteration estimate the parameters of the model using the MCL approach of Equation \eqref{eq:clmax}. For the IVT models based on the exponential trawl function (Example \ref{ex:Exp}), we set $K=1$, while we set $K = 10$ for the remaining IVT models.\footnote{Our analyses have shown that $K=1$ will deliver good estimation results for the IVT models with exponential trawl functions, but poor estimation results for the models with more other trawl functions. This is not surprising since the correlation structure for an IVT model with an exponential trawl is very simple, while it is more complicated for other IVT processes. The upshot is that choosing $K=1$ is sufficient for the simple exponential trawl-based IVT models, while it is necessary to choose $K>1$ to obtain good results for IVT models constructed using other trawl functions. This is analogous to the situation for the GMM estimator, where the estimator of the $\lambda$ parameter in the exponential trawl function has a closed-form solution using only the  autocorrelation function calculated at the first lag, cf. Section \ref{sec:mm}.} In extensive simulation experiments (not reported here), we verified that the results are robust to the choice of $K$. 

	As mentioned, previous applied work using IVTs has mainly relied on the moment-based estimator. We, therefore, compare the GMM estimation procedure, laid forth in Section \ref{sec:mm}, with the MCL estimator suggested in this paper. Figure \ref{fig:CLvsMM} plots the RMSE of the MCL estimator of a given parameter divided by the RMSE of the GMM estimator of the same parameter for the six DGPs of Table \ref{tab:paramTab}. Thus, numbers smaller than one indicate that the MCL estimator has a lower RMSE than the GMM estimator and vice versa for numbers larger than one. We see that for most parameters in most of the DGPs, the MCL estimator outperforms the GMM estimator substantially; indeed, in many cases, the RMSE of the MCL estimator is around $50\%$ that of the GMM estimator. The exception seems to be the trawl parameters, i.e.~the parameters controlling the autocorrelation structure, in the case of the Gamma and IG trawls, where the GMM estimator occasionally performs on par with the MCL estimator. However, in most cases, it appears that the MCL estimator is able to provide large improvements over the GMM estimator.

	
	\subsubsection{Simulation results supplementing those from the main paper}\label{sec:supp_sims}
	
	The simulation results, for various values of $n$, are shown in Tables \ref{tab:poi_exp}--\ref{tab:nb_gam} for the six DGPs of Table \ref{tab:paramTab}. We report the median, the median bias, and the root median squared error (RMSE) of the estimator, calculated over the $500$ Monte Carlo replications. The reason for reporting the median, instead of the mean, is that we found that when the number of observations, $n$, is small, the estimation approach will occasionally result in large outliers in few of the Monte Carlo runs, thus skewing the results (this was the case for both the MCL and GMM estimators). 
	
	From the tables, we see evidence of the MCL estimator being consistent, i.e.~the bias converges towards zero as the number of observations, $n$, grows. As expected, the estimator is most precise for the simpler models, e.g. the Poisson-Exp IVT (Table \ref{tab:poi_exp}) and somewhat less precise for the more complex models, e.g. the NB-IG model (Table \ref{tab:nb_ig}) and NB-Gamma model (Table \ref{tab:nb_gam}). 
	
	\begin{table}
\caption{\it MCL estimation results: Poisson trawl process with exponential trawl function}
\begin{center}
\footnotesize
\begin{tabularx}{1.00\textwidth}{@{\extracolsep{\stretch{1}}}rcccccc@{}} 
\toprule
& \multicolumn{3}{c}{$\hat{\nu}$ ($\nu = 17.5$)} & \multicolumn{3}{c}{$\hat{\lambda}$ ($\lambda = 1.8$)}  \\  
\cmidrule{2-4} \cmidrule{5-7}  \\  
$n$ & Med. & Bias & RMSE & Med. & Bias & RMSE \\ 
\midrule
  $    100 $ & $ 17.3213 $ & $ -0.1787 $ & $ 1.8616 $ & $ 1.8351 $ & $ 0.0351 $ & $ 0.2035 $\\
  $    250 $ & $ 17.4204 $ & $ -0.0796 $ & $ 1.1761 $ & $ 1.8009 $ & $ 0.0009 $ & $ 0.1279 $\\
  $    500 $ & $ 17.5179 $ & $ 0.0179 $ & $ 0.8444 $ & $ 1.8012 $ & $ 0.0012 $ & $ 0.0883 $\\
  $   1000 $ & $ 17.6023 $ & $ 0.1023 $ & $ 0.6164 $ & $ 1.8049 $ & $ 0.0049 $ & $ 0.0626 $\\
  $   2000 $ & $ 17.5540 $ & $ 0.0540 $ & $ 0.4538 $ & $ 1.8026 $ & $ 0.0026 $ & $ 0.0476 $\\
  $   4000 $ & $ 17.5099 $ & $ 0.0099 $ & $ 0.3038 $ & $ 1.8000 $ & $ 0.0000 $ & $ 0.0327 $\\
  $   8000 $ & $ 17.5302 $ & $ 0.0302 $ & $ 0.2197 $ & $ 1.8036 $ & $ 0.0036 $ & $ 0.0222 $\\
\bottomrule 
\end{tabularx}
\end{center}
{\footnotesize \it Median (Med.), median bias (Bias) and root median squared error (RMSE) of the MCL estimator with $K=1$. DGP: Poisson-Exponential IVT process. The IVT process $X_t$ is simulated on the grid $t = \Delta, 2\Delta, \ldots, n\Delta$, with $\Delta = 0.10$, see Table \ref{tab:paramTab} for  the values of the parameters used in the simulations. Number of Monte Carlo simulations: $500$.} 
\label{tab:poi_exp}
\end{table}

	\begin{table}
\caption{\it CL estimation results: Poisson trawl process with IG trawl function}
\begin{center}
\footnotesize
\begin{tabularx}{1.00\textwidth}{@{\extracolsep{\stretch{1}}}rccccccccc@{}} 
\toprule
& \multicolumn{3}{c}{$\hat{\nu}$ ($\nu = 17.5$)} & \multicolumn{3}{c}{$\hat{\delta}$ ($\delta = 1.8$)} & \multicolumn{3}{c}{$\hat{\gamma}$ ($\gamma = 0.8$)} \\  
\cmidrule{2-4} \cmidrule{5-7} \cmidrule{8-10}  \\  
$n$ & Med. & Bias & RMSE & Med. & Bias & RMSE & Med. & Bias & RMSE \\ 
\midrule
  $    250 $ & $ 17.7984 $ & $ 0.2984 $ & $ 1.8764 $ & $ 2.4177 $ & $ 0.6177 $ & $ 0.9923 $ & $ 1.0985 $ & $ 0.2985 $ & $ 0.5365 $\\
  $    500 $ & $ 17.7749 $ & $ 0.2749 $ & $ 1.4921 $ & $ 2.0205 $ & $ 0.2205 $ & $ 0.6317 $ & $ 0.9306 $ & $ 0.1306 $ & $ 0.3352 $\\
  $   1000 $ & $ 17.6514 $ & $ 0.1514 $ & $ 1.0344 $ & $ 1.9141 $ & $ 0.1141 $ & $ 0.4896 $ & $ 0.8592 $ & $ 0.0592 $ & $ 0.2422 $\\
  $   2000 $ & $ 17.5182 $ & $ 0.0182 $ & $ 0.8090 $ & $ 1.8977 $ & $ 0.0977 $ & $ 0.3413 $ & $ 0.8333 $ & $ 0.0333 $ & $ 0.1864 $\\
  $   4000 $ & $ 17.5273 $ & $ 0.0273 $ & $ 0.5931 $ & $ 1.8120 $ & $ 0.0120 $ & $ 0.2426 $ & $ 0.8051 $ & $ 0.0051 $ & $ 0.1347 $\\
  $   8000 $ & $ 17.4966 $ & $ -0.0034 $ & $ 0.4125 $ & $ 1.8072 $ & $ 0.0072 $ & $ 0.1692 $ & $ 0.8030 $ & $ 0.0030 $ & $ 0.0842 $\\
\bottomrule 
\end{tabularx}
\end{center}
{\footnotesize \it Median (Med.), median bias (Bias) and root median squared error (RMSE) of the MCL estimator with $K = 10$. DGP: Poisson-IG IVT. The IVT process $X_t$ is simulated on the grid $t = \Delta, 2\Delta, \ldots, n\Delta$, with $\Delta = 0.10$, see Table \ref{tab:paramTab} for  the values of the parameters used in the simulations. Number of Monte Carlo simulations: $500$.} 
\label{tab:poi_ig}
\end{table}

	\begin{table}
\caption{\it CL estimation results: Poisson trawl process with $\Gamma$ trawl function}
\begin{center}
\footnotesize
\begin{tabularx}{1.00\textwidth}{@{\extracolsep{\stretch{1}}}rccccccccc@{}} 
\toprule
& \multicolumn{3}{c}{$\hat{\nu}$ ($\nu = 17.5$)} & \multicolumn{3}{c}{$\hat{H}$ ($H = 1.7$)} & \multicolumn{3}{c}{$\hat{\alpha}$ ($\alpha = 0.8$)} \\  
\cmidrule{2-4} \cmidrule{5-7} \cmidrule{8-10}  \\  
$n$ & Med. & Bias & RMSE & Med. & Bias & RMSE & Med. & Bias & RMSE \\ 
\midrule
  $    250 $ & $ 17.8207 $ & $ 0.3207 $ & $ 1.8470 $ & $ 3.0719 $ & $ 1.3719 $ & $ 1.4020 $ & $ 1.5159 $ & $ 0.7159 $ & $ 0.7302 $\\
  $    500 $ & $ 17.5946 $ & $ 0.0946 $ & $ 1.2180 $ & $ 2.1831 $ & $ 0.4831 $ & $ 0.9537 $ & $ 1.0463 $ & $ 0.2463 $ & $ 0.4816 $\\
  $   1000 $ & $ 17.4396 $ & $ -0.0604 $ & $ 0.9877 $ & $ 1.8563 $ & $ 0.1563 $ & $ 0.7227 $ & $ 0.8914 $ & $ 0.0914 $ & $ 0.3636 $\\
  $   2000 $ & $ 17.4540 $ & $ -0.0460 $ & $ 0.6642 $ & $ 1.7807 $ & $ 0.0807 $ & $ 0.5371 $ & $ 0.8406 $ & $ 0.0406 $ & $ 0.2728 $\\
  $   4000 $ & $ 17.5048 $ & $ 0.0048 $ & $ 0.4821 $ & $ 1.6902 $ & $ -0.0098 $ & $ 0.3851 $ & $ 0.7784 $ & $ -0.0216 $ & $ 0.2004 $\\
  $   8000 $ & $ 17.5612 $ & $ 0.0612 $ & $ 0.3925 $ & $ 1.6433 $ & $ -0.0567 $ & $ 0.2464 $ & $ 0.7607 $ & $ -0.0393 $ & $ 0.1189 $\\
\bottomrule 
\end{tabularx}
\end{center}
{\footnotesize \it Median (Med.), median bias (Bias) and root median squared error (RMSE) of the MCL estimator with $K = 10$. DGP: Poisson-Gamma IVT. The IVT process $X_t$ is simulated on the grid $t = \Delta, 2\Delta, \ldots, n\Delta$, with $\Delta = 0.10$, see Table \ref{tab:paramTab} for  the values of the parameters used in the simulations. Number of Monte Carlo simulations: $500$.} 
\label{tab:poi_gam}
\end{table}

	\begin{table}
\caption{\it CL estimation results: NB trawl process with exponential trawl function}
\begin{center}
\footnotesize
\begin{tabularx}{1.00\textwidth}{@{\extracolsep{\stretch{1}}}rccccccccc@{}} 
\toprule
& \multicolumn{3}{c}{$\hat{m}$ ($m = 7.5$)} & \multicolumn{3}{c}{$\hat{p}$ ($p = 0.7$)}  & \multicolumn{3}{c}{$\hat{\lambda}$ ($\lambda = 1.8$)}  \\  
\cmidrule{2-4} \cmidrule{5-7} \cmidrule{8-10}  \\  
$n$ & Med. & Bias & RMSE & Med. & Bias & RMSE & Med. & Bias & RMSE \\ 
\midrule
  $    100 $ & $ 8.9761 $ & $ 1.4761 $ & $ 2.3850 $ & $ 0.6649 $ & $ -0.0351 $ & $ 0.0623 $ & $ 1.9197 $ & $ 0.1197 $ & $ 0.2595 $\\
  $    250 $ & $ 7.9680 $ & $ 0.4680 $ & $ 1.3264 $ & $ 0.6872 $ & $ -0.0128 $ & $ 0.0353 $ & $ 1.8311 $ & $ 0.0311 $ & $ 0.1431 $\\
  $    500 $ & $ 7.8443 $ & $ 0.3443 $ & $ 0.9956 $ & $ 0.6896 $ & $ -0.0104 $ & $ 0.0287 $ & $ 1.8202 $ & $ 0.0202 $ & $ 0.1010 $\\
  $   1000 $ & $ 7.6891 $ & $ 0.1891 $ & $ 0.6964 $ & $ 0.6970 $ & $ -0.0030 $ & $ 0.0202 $ & $ 1.8106 $ & $ 0.0106 $ & $ 0.0728 $\\
  $   2000 $ & $ 7.5855 $ & $ 0.0855 $ & $ 0.5010 $ & $ 0.6972 $ & $ -0.0028 $ & $ 0.0143 $ & $ 1.8001 $ & $ 0.0001 $ & $ 0.0470 $\\
  $   4000 $ & $ 7.5362 $ & $ 0.0362 $ & $ 0.3484 $ & $ 0.6993 $ & $ -0.0007 $ & $ 0.0095 $ & $ 1.8013 $ & $ 0.0013 $ & $ 0.0316 $\\
  $   8000 $ & $ 7.5265 $ & $ 0.0265 $ & $ 0.2365 $ & $ 0.6993 $ & $ -0.0007 $ & $ 0.0065 $ & $ 1.7994 $ & $ -0.0006 $ & $ 0.0251 $\\
\bottomrule 
\end{tabularx}
\end{center}
{\footnotesize \it Median (Med.), median bias (Bias) and root median squared error (RMSE) of the MCL estimator with $K = 1$. DGP: Negative Binomial-Exponential IVT process. The IVT process $X_t$ is simulated on the grid $t = \Delta, 2\Delta, \ldots, n\Delta$, with $\Delta = 0.10$, see Table \ref{tab:paramTab} for  the values of the parameters used in the simulations. Number of Monte Carlo simulations: $500$.} 
\label{tab:nb_exp}
\end{table}

	\begin{table}
\caption{\it CL estimation results: NB trawl process with IG trawl function}
\begin{center}
\scriptsize
\begin{tabularx}{1.00\textwidth}{@{\extracolsep{\stretch{1}}}rcccccccccccc@{}} 
\toprule
& \multicolumn{3}{c}{$\hat{m}$ ($m = 7.5$)} & \multicolumn{3}{c}{$\hat{p}$ ($p = 0.7$)} & \multicolumn{3}{c}{$\hat{\delta}$ ($\delta = 1.8$)} & \multicolumn{3}{c}{$\hat{\gamma}$ ($\gamma = 0.8$)} \\  
\cmidrule{2-4} \cmidrule{5-7} \cmidrule{8-10}  \cmidrule{11-13}  \\  
$n$ & Med. & Bias & RMSE & Med. & Bias & RMSE & Med. & Bias & RMSE & Med. & Bias & RMSE \\ 
\midrule
  $    250 $ & $ 9.4842 $ & $ 1.9842 $ & $ 2.2867 $ & $ 0.6695 $ & $ -0.0305 $ & $ 0.0501 $ & $ 3.1478 $ & $ 1.3478 $ & $ 1.3540 $ & $ 1.3513 $ & $ 0.5513 $ & $ 0.6440 $\\
  $    500 $ & $ 8.1866 $ & $ 0.6866 $ & $ 1.4351 $ & $ 0.6885 $ & $ -0.0115 $ & $ 0.0335 $ & $ 2.3665 $ & $ 0.5665 $ & $ 0.8803 $ & $ 1.0506 $ & $ 0.2506 $ & $ 0.4657 $\\
  $   1000 $ & $ 7.7323 $ & $ 0.2323 $ & $ 0.8607 $ & $ 0.6941 $ & $ -0.0059 $ & $ 0.0256 $ & $ 1.9682 $ & $ 0.1682 $ & $ 0.6219 $ & $ 0.8834 $ & $ 0.0834 $ & $ 0.3099 $\\
  $   2000 $ & $ 7.6150 $ & $ 0.1150 $ & $ 0.5755 $ & $ 0.6970 $ & $ -0.0030 $ & $ 0.0176 $ & $ 1.9398 $ & $ 0.1398 $ & $ 0.4398 $ & $ 0.8588 $ & $ 0.0588 $ & $ 0.2116 $\\
  $   4000 $ & $ 7.5561 $ & $ 0.0561 $ & $ 0.4348 $ & $ 0.6994 $ & $ -0.0006 $ & $ 0.0129 $ & $ 1.8866 $ & $ 0.0866 $ & $ 0.3547 $ & $ 0.8351 $ & $ 0.0351 $ & $ 0.1741 $\\
  $   8000 $ & $ 7.5317 $ & $ 0.0317 $ & $ 0.3103 $ & $ 0.6990 $ & $ -0.0010 $ & $ 0.0088 $ & $ 1.8470 $ & $ 0.0470 $ & $ 0.2266 $ & $ 0.8247 $ & $ 0.0247 $ & $ 0.1063 $\\
\bottomrule 
\end{tabularx}
\end{center}
{\footnotesize \it Median (Med.), median bias (Bias) and root median squared error (RMSE) of the MCL estimator with  $K = 10$. DGP: Negative Binomial-IG IVT. The IVT process $X_t$ is simulated on the grid $t = \Delta, 2\Delta, \ldots, n\Delta$, with $\Delta = 0.10$, see Table \ref{tab:paramTab} for  the values of the parameters used in the simulations. Number of Monte Carlo simulations: $500$.} 
\label{tab:nb_ig}
\end{table}

	\begin{table}
\caption{\it CL estimation results: NB trawl process with $\Gamma$ trawl function}
\begin{center}
\scriptsize
\begin{tabularx}{1.00\textwidth}{@{\extracolsep{\stretch{1}}}rcccccccccccc@{}} 
\toprule
& \multicolumn{3}{c}{$\hat{m}$ ($m = 7.5$)} & \multicolumn{3}{c}{$\hat{p}$ ($p = 0.7$)} & \multicolumn{3}{c}{$\hat{H}$ ($H = 1.7$)} & \multicolumn{3}{c}{$\hat{\alpha}$ ($\alpha = 0.8$)} \\  
\cmidrule{2-4} \cmidrule{5-7} \cmidrule{8-10}  \cmidrule{11-13}  \\  
$n$ & Med. & Bias & RMSE & Med. & Bias & RMSE & Med. & Bias & RMSE & Med. & Bias & RMSE \\ 
\midrule
  $    250 $ & $ 9.2542 $ & $ 1.7542 $ & $ 2.3233 $ & $ 0.6678 $ & $ -0.0322 $ & $ 0.0532 $ & $ 4.0968 $ & $ 2.3968 $ & $ 2.3968 $ & $ 1.9812 $ & $ 1.1812 $ & $ 1.1814 $\\
  $    500 $ & $ 8.3098 $ & $ 0.8098 $ & $ 1.5008 $ & $ 0.6799 $ & $ -0.0201 $ & $ 0.0384 $ & $ 2.6049 $ & $ 0.9049 $ & $ 1.2064 $ & $ 1.2160 $ & $ 0.4160 $ & $ 0.5768 $\\
  $   1000 $ & $ 7.9214 $ & $ 0.4214 $ & $ 0.9728 $ & $ 0.6890 $ & $ -0.0110 $ & $ 0.0257 $ & $ 2.2420 $ & $ 0.5420 $ & $ 0.9221 $ & $ 1.0470 $ & $ 0.2470 $ & $ 0.4404 $\\
  $   2000 $ & $ 7.7270 $ & $ 0.2270 $ & $ 0.6394 $ & $ 0.6938 $ & $ -0.0062 $ & $ 0.0179 $ & $ 1.9538 $ & $ 0.2538 $ & $ 0.7091 $ & $ 0.9244 $ & $ 0.1244 $ & $ 0.3413 $\\
  $   4000 $ & $ 7.5876 $ & $ 0.0876 $ & $ 0.4594 $ & $ 0.6977 $ & $ -0.0023 $ & $ 0.0130 $ & $ 1.8234 $ & $ 0.1234 $ & $ 0.5016 $ & $ 0.8545 $ & $ 0.0545 $ & $ 0.2418 $\\
  $   8000 $ & $ 7.5249 $ & $ 0.0249 $ & $ 0.3355 $ & $ 0.6990 $ & $ -0.0010 $ & $ 0.0091 $ & $ 1.7745 $ & $ 0.0745 $ & $ 0.3466 $ & $ 0.8436 $ & $ 0.0436 $ & $ 0.1671 $\\
\bottomrule 
\end{tabularx}
\end{center}
{\footnotesize \it Median (Med.), median bias (Bias) and root median squared error (RMSE) of the MCL estimator with $K = 10$. DGP: Negative Binomial-Gamma IVT process. The IVT process $X_t$ is simulated on the grid $t = \Delta, 2\Delta, \ldots, n\Delta$, with $\Delta = 0.10$, see Table \ref{tab:paramTab} for  the values of the parameters used in the simulations. Number of Monte Carlo simulations: $500$.} 
\label{tab:nb_gam}
\end{table}

	\clearpage
	
	\newpage
	
	\subsection{Finite sample approximation of the asymptotic distribution}\label{sec:fs_asym}
	This section investigates how close the finite sample distribution of the MCL estimator is to the true (Gaussian) asymptotic limit, as presented in Theorem \ref{th:CLT}.
	
	We simulate $M = 1,000$ Monte Carlo replications of the IVT processes considered in the previous section, i.e.~with parameter values given in Table \ref{tab:paramTab}. For each replication, we estimate the asymptotic covariance matrix $G(\theta_0)^{-1}$ of Theorem \ref{th:CLT} using the simulation-based approach, described in Section \ref{app:B1}. We then construct the standardized version of the estimated parameters, 
	\begin{align}\label{eq:z_std}
	z_{i,m} = \sqrt{n} \frac{\hat\theta_{i}^{CL} - \theta_{i,0}}{ \sqrt{ \hat G(\hat \theta^{CL})^{-1}_{i,i}} } ,\quad m = 1, 2\ldots, M,
	\end{align}
	where the $i$ denote the $i$th entrance in the parameter vector $\theta$, and $\hat G(\hat \theta^{CL})^{-1}_{i,i}$ is the $i$th diagonal entrance of the matrix $\hat G(\hat \theta^{CL})^{-1}$. 
	
	According to Theorem \ref{th:CLT}, we would expect that $z_{i,m}$ is  distributed approximately as a standard normal random variable. Figures \ref{fig:QQ1}--\ref{fig:QQ6} contain QQ plots of $\{z_{i,m}\}_{m=1}^{M}$ for various sample sizes ($n$), various parameters (denoted by $i$ here), and for the six different DGPs, respectively.
	
	From the figures, we observe a general tendency: The finite sample distribution of the estimators of the parameters governing the marginal distribution ($\nu$ for the Poisson distribution, $m$ and $p$ for the NB distribution) is close to the standard normal distribution. The same holds for the trawl parameter in the case of the IVTs with exponential autocorrelation function (Exponential trawl with parameter $\lambda$). The picture changes when there are two parameters in the trawl function ($\delta$ and $\gamma$ in the case of the IG trawl and $H$ and $\alpha$ in the case of the Gamma trawl): Here, the convergence to the Gaussian distribution appears to be quite slow. Indeed, even for $n=8000$, there are deviations from the Gaussian distribution. 
	
			\begin{figure}[!t]
		\centering
		\includegraphics[scale=0.80]{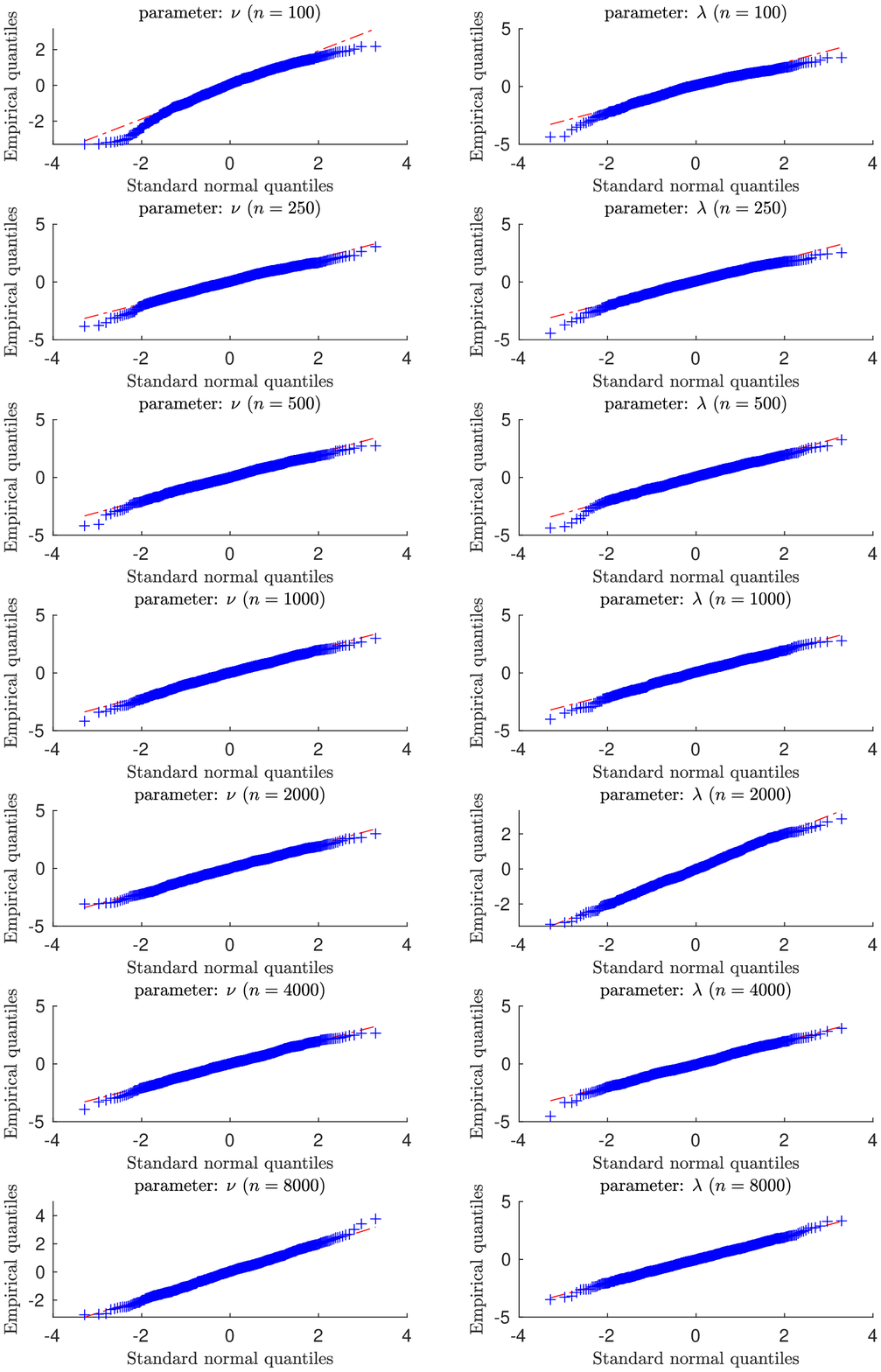} 
		\caption{\it QQ plot of $\{z_{i,m}\}_{m=1}^{M}$ of Equation \eqref{eq:z_std} for the Poisson-Exp DGP.}
		\label{fig:QQ1}
	\end{figure}

			\begin{figure}[!t]
		\centering
		\includegraphics[scale=0.80]{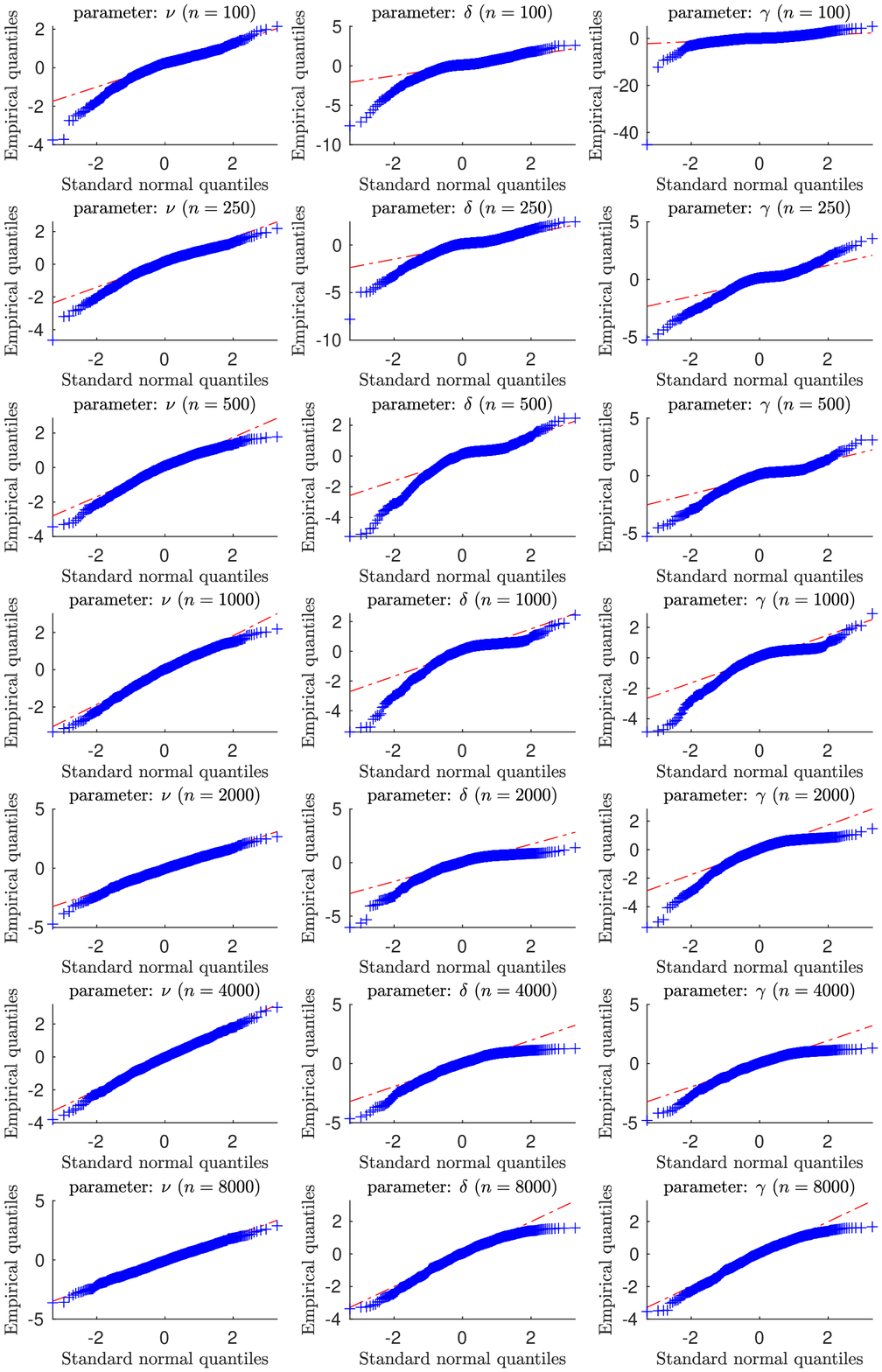} 
		\caption{\it QQ plot of $\{z_{i,m}\}_{m=1}^{M}$ of Equation \eqref{eq:z_std} for the  Poisson-IG DGP.}
		\label{fig:QQ2}
	\end{figure}

			\begin{figure}[!t]
		\centering
		\includegraphics[scale=0.80]{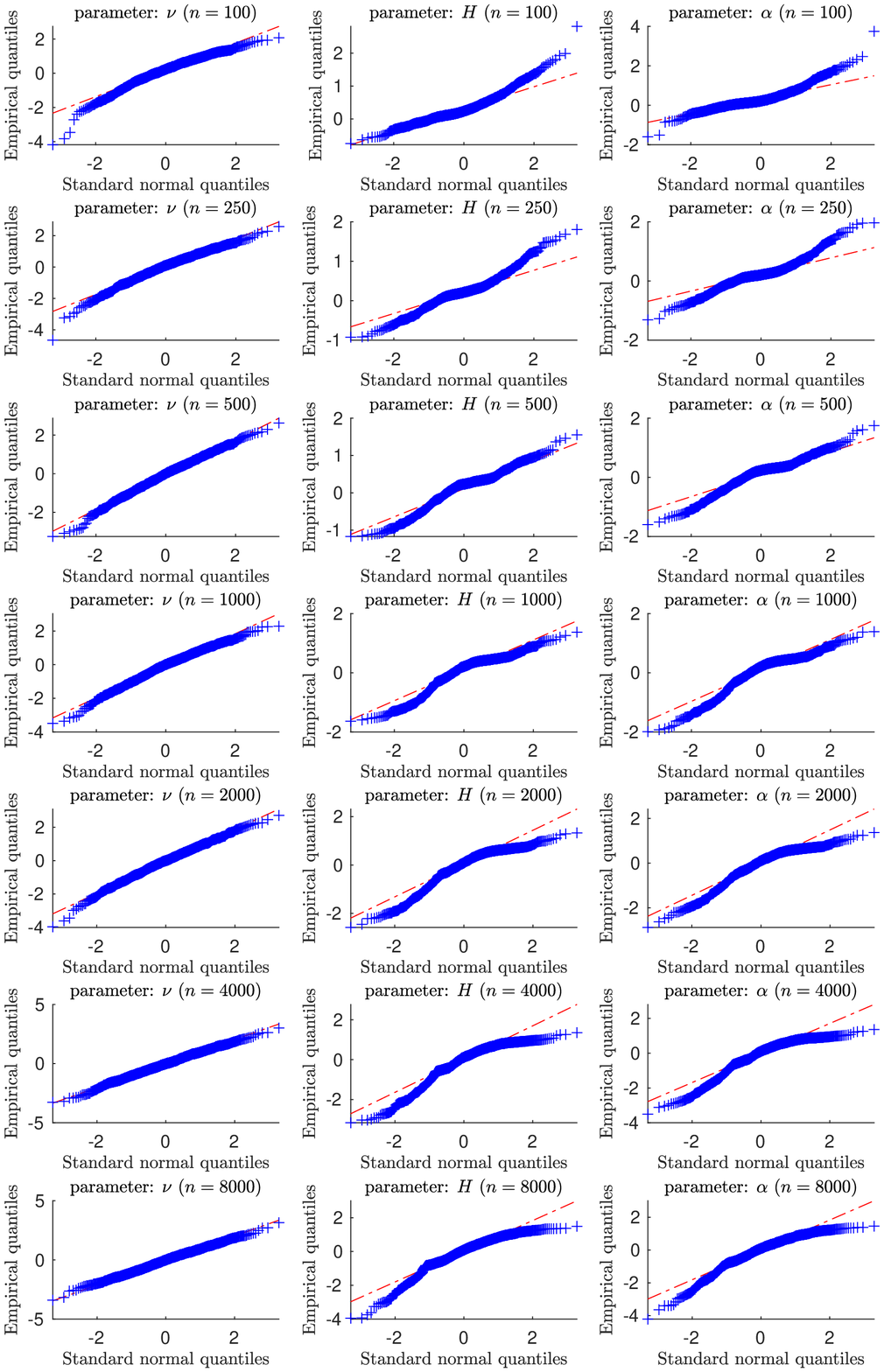} 
		\caption{\it QQ plot of $\{z_{i,m}\}_{m=1}^{M}$ of Equation \eqref{eq:z_std} for the  Poisson-Gamma DGP.}
		\label{fig:QQ3}
	\end{figure}
	
				\begin{figure}[!t]
		\centering
		\includegraphics[scale=0.80]{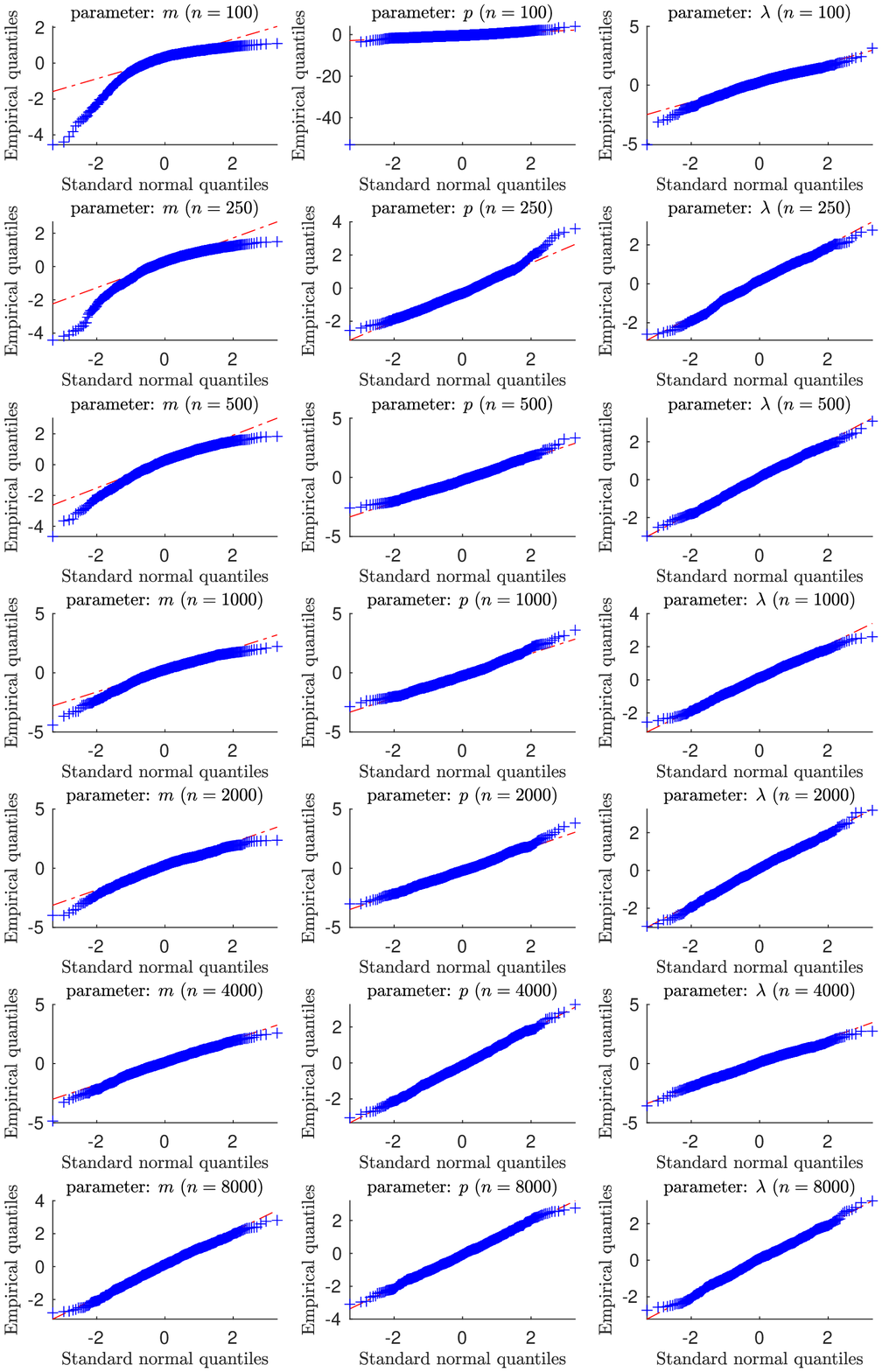} 
		\caption{\it QQ plot of $\{z_{i,m}\}_{m=1}^{M}$ of Equation \eqref{eq:z_std} for the NB-Exp DGP.}
		\label{fig:QQ4}
	\end{figure}

			\begin{figure}[!t]
		\centering
		\includegraphics[scale=0.80]{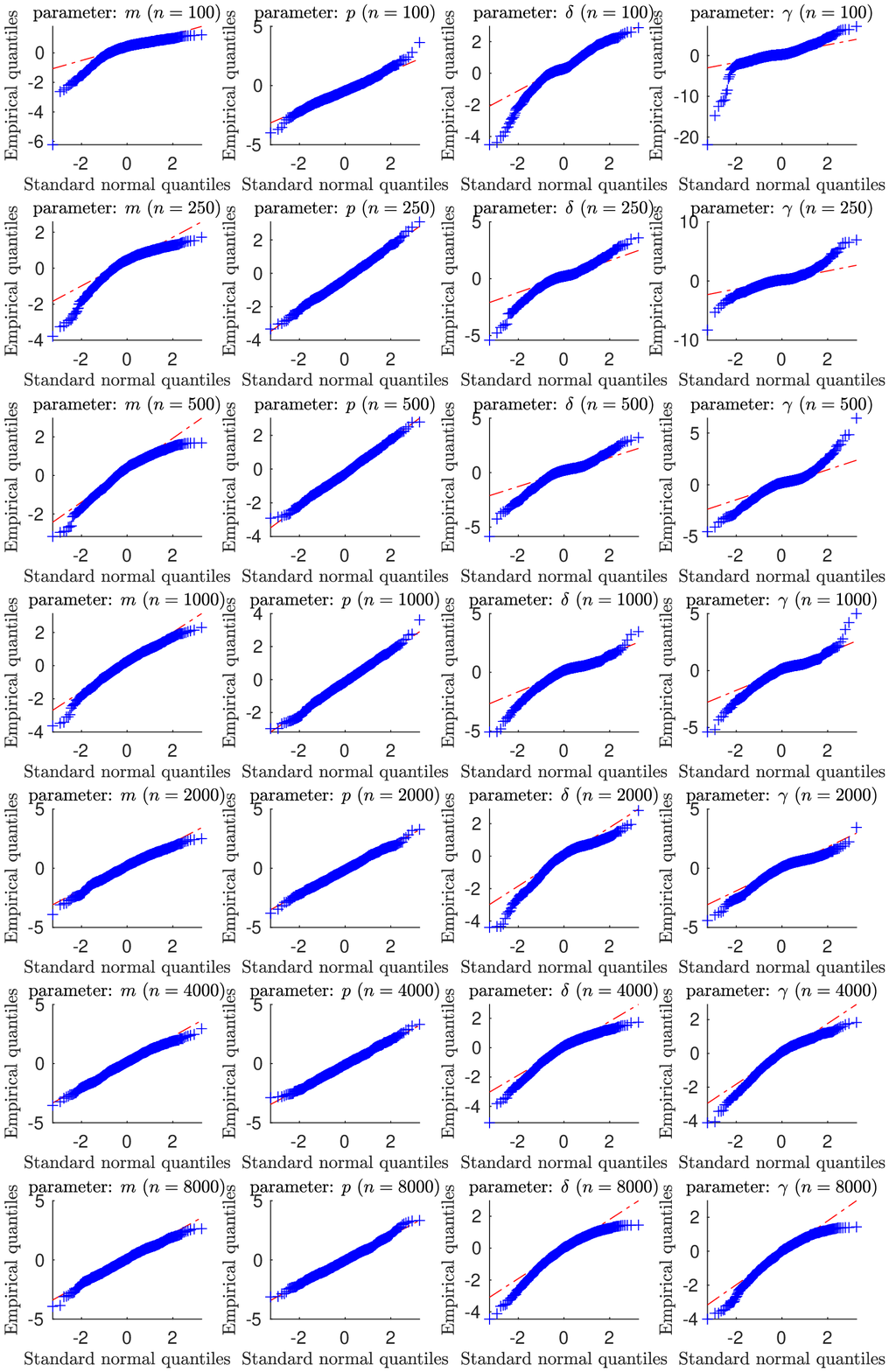} 
		\caption{\it QQ plot of $\{z_{i,m}\}_{m=1}^{M}$ of Equation \eqref{eq:z_std} for the NB-IG DGP.}
		\label{fig:QQ5}
	\end{figure}

			\begin{figure}[!t]
		\centering
		\includegraphics[scale=0.80]{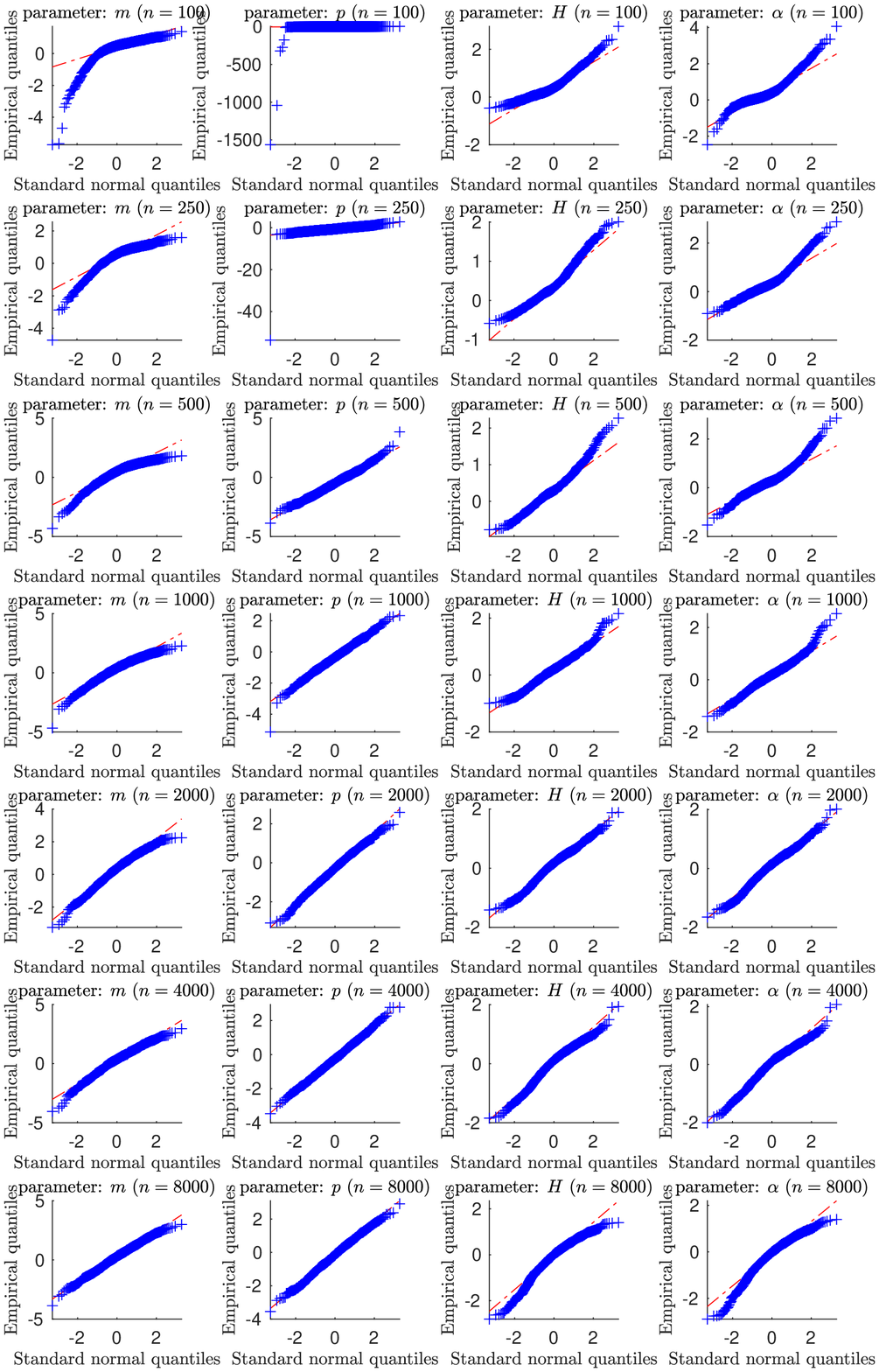} 
		\caption{\it QQ plot of $\{z_{i,m}\}_{m=1}^{M}$ of Equation \eqref{eq:z_std} for the NB-Gamma DGP.}
		\label{fig:QQ6}
	\end{figure}
	
%
	
	\clearpage
	
	\newpage
	
	\subsection{Finite sample properties of the model selection procedure}\label{sec:fs_ic}
	This section illustrates the use, and finite sample properties, of the model selection procedure introduced in Section 3.1.3 of the main paper.  Consider $n =  4000$ equidistant observations of an IVT process on a grid of $\Delta = 0.10$, i.e.~$X_{\Delta}, X_{2\Delta}, \ldots, X_{n\Delta}$.\footnote{The number of simulated observations, $n = 4000$, the space between observations, $\Delta = 0.10$, and the tuning parameter, $K = 10$, are chosen such as to be comparable to the data studied in the empirical section of the main paper.} For each of the six possible models, we then calculate the three goodness-of-fit measures, namely the value of the maximized composite likelihood function $CL$, the AIC-like composite likelihood information criteria $CLAIC$, and the BIC-like composite likelihood information criteria $CLBIC$. The model which has the maximum value of a criterion is ``selected'' by that criterion. We repeat this process for $100$ Monte Carlo replications and the six different DGPs using the parameters of Table \ref{tab:paramTab}. Figure \ref{fig:ic} reports the ``selection rates'' of the models, i.e.~the fraction of times that a model, given on the $x$-axis, is selected, for each of the three different criteria. Each panel in the figure corresponds to a particular DGP, as shown above, the respective panels.

	Consider, for instance, the case where the true DGP is the NB-Exp IVT model. The results from using this DGP are given in the upper right panel of Figure \ref{fig:ic}. In this case, when we estimate the six different models and calculate the three goodness-of-fit measures, the true model (i.e.~NB-Exp)  has the highest composite likelihood value in $70\%$ of simulations. In contrast, the CLAIC and CLBIC result in selecting the true model in $73\%$ and $81\%$ of the simulations, respectively. Note that since the models considered here are not nested, it is not necessarily the case that the maximized composite likelihood value CL will be larger for the more complicated models.
	
	Overall, Figure \ref{fig:ic} indicates that the model selection procedure is quite accurate when the marginal distribution of the DGP is the Negative Binomial distribution. Conversely, when the marginal distribution of the DGP is the Poisson distribution, the correct model is chosen less often. However, in these situations, it is often the case that although the Negative Binomial distribution is (incorrectly) preferred to the Poisson distribution, the correct trawl function (autocorrelation structure) is nonetheless selected.  
	
	Lastly, to examine the effect of the tuning parameter $K$ on the model selection procedure, we ran the same experiment but using both $K=5$ and $K=20$ (results not shown here, but available upon request). We find that the model selection procedure deteriorates when $K = 5$, while it performs similarly to that shown in  Figure \ref{fig:ic} when $K=20$, indicating that it is important to set $K$ sufficiently large value so that the selection criteria can properly distinguish between the models.

	

	

	\begin{figure}[!t]
		\centering
		\includegraphics[scale=0.95]{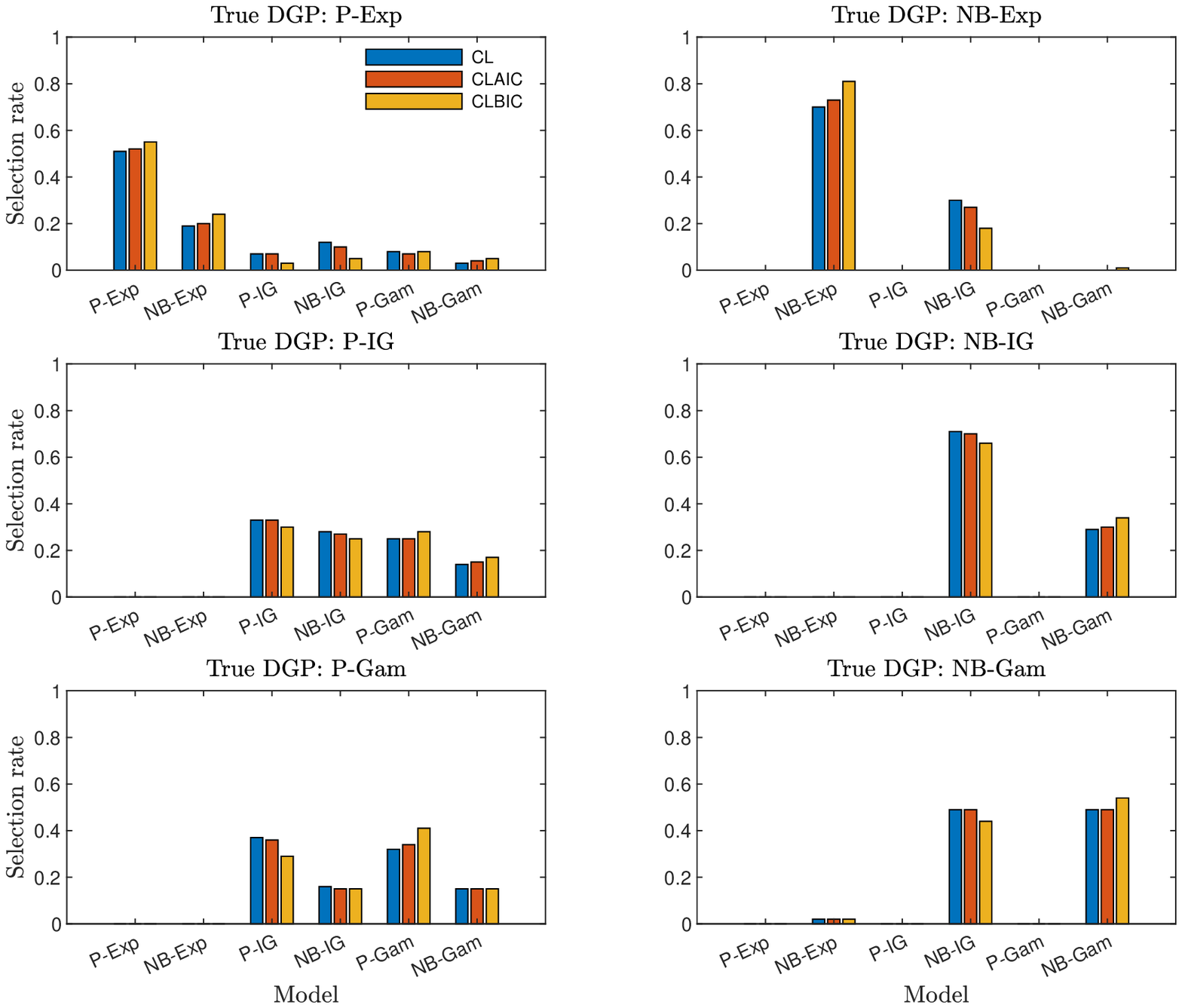} 
		\caption{\it Simulation study of model selection procedure. Each plot represents the outcome of a separate Monte Carlo study, where the true DGP in the study is given above the plot. The numbers plotted are the average selection rates of the models given on the $x$-axis, using a given criterion over $M = 100$ Monte Carlo simulations.  For each Monte Carlo replication, $n = 4000$ observations of the true DGP are simulated on a grid with step size $\Delta = 0.10$. The parameters used in the study are given in Table \ref{tab:paramTab} and we set $K = 10$.}
		\label{fig:ic}
	\end{figure}
	
	\clearpage
	
	\newpage

	\subsection{Alternative simulation setup} \label{sec:setup2}
	We perform simulation experiments similar to those in the main paper but with a different set of simulation settings. The parameter values used for the DGPs in this study are given in Table \ref{tab:paramTab2}; the associated implied marginal distributions of the underlying L\'evy bases and autocorrelations of the IVT processes are shown in Figure \ref{fig:simEx2}.
	
	In this simulation study, we simulate $n$ observations of an IVT process $X_t$ on an equidistant grid of size $\Delta = 0.10$. For the IVTs based on the exponential trawl function, we set $K = 1$, while we set $K = 3$ for the remaining IVTs. This should be contrasted to the setup of the main paper, where we set $K = 10$. The finite sample estimation results can be found in Tables \ref{tab:poi_exp2}--\ref{tab:nb_gam2}. Figure \ref{fig:CLvsMM2} plots the relative RMSE of the MCL estimator compared to the MM estimator; numbers smaller than one favour the MCL estimator.

	\clearpage
	
	\newpage

	\begin{table}
\caption{\it Parameter values used in simulation setup 2}
\begin{center}
\footnotesize
\begin{tabularx}{1.00\textwidth}{@{\extracolsep{\stretch{1}}}lcccccccc@{}} 
\toprule
DGP        &  $\nu$  & $m$           & $p$      & $\lambda$ & $\delta$ & $\gamma$ & $H$ & $\alpha$ \\
\cmidrule{2-9}
P-Exp      & $5.00$      &                  &              & $1.00$            &               & & & \\
P-IG        & $5.00$      &                  &              &                   &  $0.75$  & $0.50$ & & \\
P-Gamma  & $5.00$      &                  &              &              & &     &  $0.50$  & $0.75$ \\
NB-Exp   &            & $2.14$   & $0.70$  & $1.00$            &               & & & \\
NB-IG     &            &   $2.14$  & $0.70$  &                   &  $0.75$ & $0.50$ & & \\
NB-Gamma     &            &   $2.14$  & $0.70$  &          & &         &  $0.50$ & $0.75$ \\
\bottomrule 
\end{tabularx}
\end{center}
{\footnotesize \it Parameter values for the six different DGPs used in the simulation studies of the Supplementary Material. See the various Examples  of the main paper for details. The value $m = \nu (1-p)/p$ with $\nu = 5$ is chosen such that the first moment of the Poisson and Negative Binomial L\'evy bases are matched. Marginal distributions and autocorrelation function implied by these parameter values are shown in Figure \ref{fig:simEx2}.}
\label{tab:paramTab2}
\end{table}
	
	\begin{figure}[!t]
		\centering
		\includegraphics[scale=0.95]{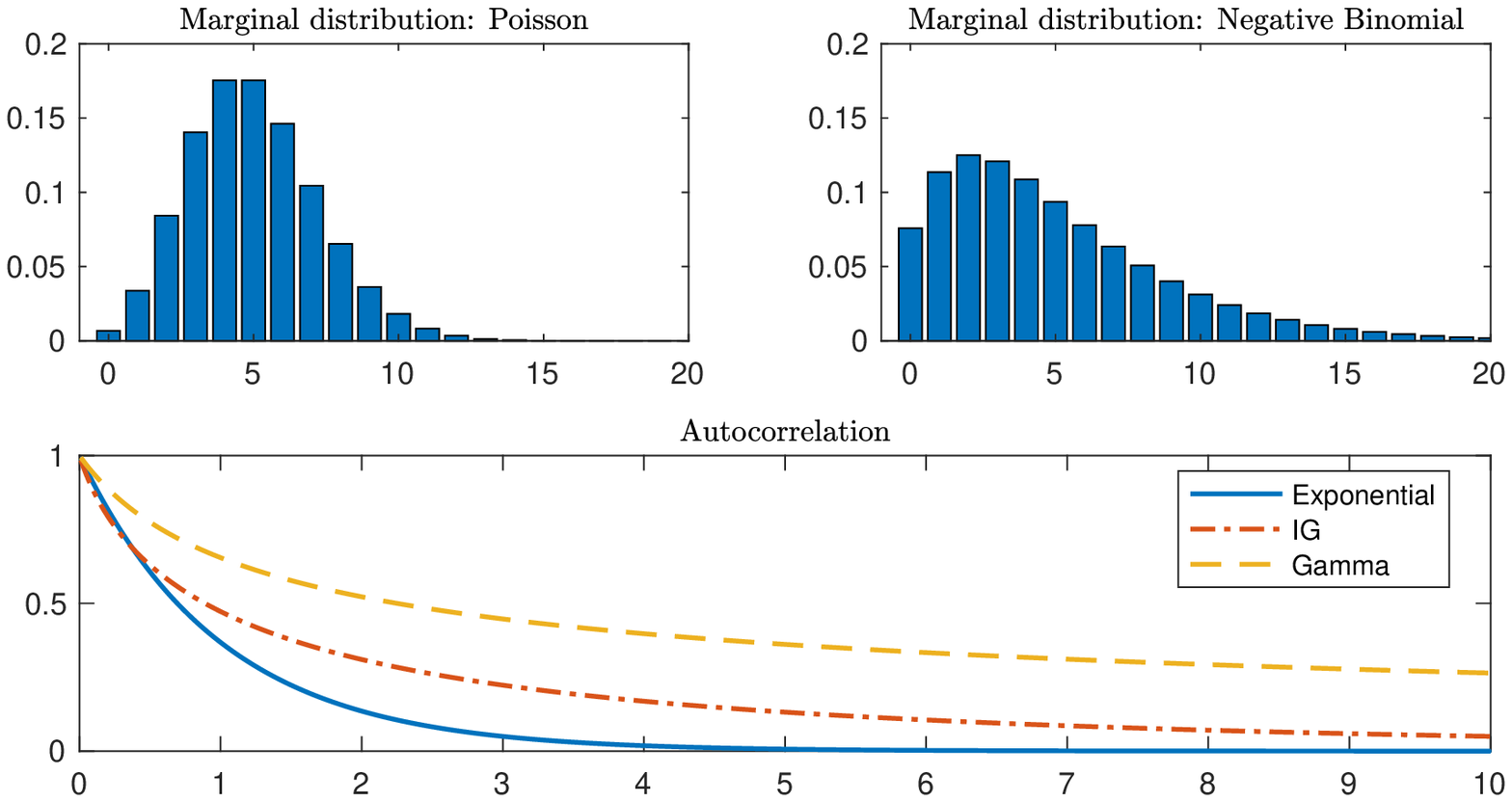} 
		\caption{\it  Marginal distributions of the L\'evy bases and autocorrelations of the DGPs used in the simulation studies of Section \ref{sec:setup2}. The marginal distribution and autocorrelation structure of IVT processes can be specified independently, resulting in six different DGPs in this setup (P-Exp, P-IG, P-Gamma, NB-Exp, NB-IG, NB-Gamma). The parameter values used to produce the plots are given in Table \ref{tab:paramTab2}. Note that the marginal distribution shown in the top plots is of the underlying L\'evy bases and not of the IVT process itself.}
		\label{fig:simEx2}
	\end{figure}

	\clearpage
	
	\newpage

	\begin{table}
\caption{\it CL estimation results: Poisson trawl process with exponential trawl function}
\begin{center}
\footnotesize
\begin{tabularx}{1.00\textwidth}{@{\extracolsep{\stretch{1}}}rcccccc@{}} 
\toprule
& \multicolumn{3}{c}{$\hat{\nu}$ ($\nu = 5$)} & \multicolumn{3}{c}{$\hat{\lambda}$ ($\lambda = 1$)}  \\  
\cmidrule{2-4} \cmidrule{5-7}  \\  
Nobs & Avg & Bias & RMSE & Avg & Bias & RMSE \\ 
\midrule
  $    100 $ & $ 4.9994 $ & $ -0.0006 $ & $ 0.6464 $ & $ 1.0131 $ & $ 0.0131 $ & $ 0.1220 $\\
  $    250 $ & $ 5.0345 $ & $ 0.0345 $ & $ 0.4055 $ & $ 1.0099 $ & $ 0.0099 $ & $ 0.0756 $\\
  $    500 $ & $ 5.0468 $ & $ 0.0468 $ & $ 0.2908 $ & $ 1.0074 $ & $ 0.0074 $ & $ 0.0623 $\\
  $   1000 $ & $ 5.0100 $ & $ 0.0100 $ & $ 0.2181 $ & $ 1.0067 $ & $ 0.0067 $ & $ 0.0441 $\\
  $   2000 $ & $ 5.0217 $ & $ 0.0217 $ & $ 0.1383 $ & $ 1.0026 $ & $ 0.0026 $ & $ 0.0288 $\\
  $   4000 $ & $ 5.0103 $ & $ 0.0103 $ & $ 0.1058 $ & $ 1.0009 $ & $ 0.0009 $ & $ 0.0203 $\\
  $   8000 $ & $ 5.0012 $ & $ 0.0012 $ & $ 0.0722 $ & $ 1.0005 $ & $ 0.0005 $ & $ 0.0150 $\\
\bottomrule 
\end{tabularx}
\end{center}
{\footnotesize \it Median (Med.), median bias (Bias) and root median squared error (RMSE) of the MCL estimator. DGP: Poisson-Exponential IVT process. The IVT process $X_t$ is simulated on the grid $t = \Delta, 2\Delta, \ldots, n\Delta$, with $\Delta = 0.10$, see Table \ref{tab:paramTab2} for  the values of the parameters used in the simulations. $K = 1$. Number of Monte Carlo simulations: $500$.} 
\label{tab:poi_exp2}
\end{table}

	\begin{table}
\caption{\it CL estimation results: Poisson trawl process with IG trawl function}
\begin{center}
\footnotesize
\begin{tabularx}{1.00\textwidth}{@{\extracolsep{\stretch{1}}}rccccccccc@{}} 
\toprule
& \multicolumn{3}{c}{$\hat{\nu}$ ($\nu = 5$)} & \multicolumn{3}{c}{$\hat{\delta}$ ($\delta = 0.75$)} & \multicolumn{3}{c}{$\hat{\gamma}$ ($\gamma = 0.5$)} \\  
\cmidrule{2-4} \cmidrule{5-7} \cmidrule{8-10}  \\  
Nobs & Avg & Bias & RMSE & Avg & Bias & RMSE & Avg & Bias & RMSE \\ 
\midrule
  $    250 $ & $ 4.8852 $ & $ -0.1148 $ & $ 0.6536 $ & $ 0.8535 $ & $ 0.1035 $ & $ 0.2654 $ & $ 5.6836 $ & $ 5.1836 $ & $ 31.0800 $\\
  $    500 $ & $ 4.8907 $ & $ -0.1093 $ & $ 0.4725 $ & $ 0.8241 $ & $ 0.0741 $ & $ 0.1833 $ & $ 1.1475 $ & $ 0.6475 $ & $ 6.2198 $\\
  $   1000 $ & $ 4.9770 $ & $ -0.0230 $ & $ 0.3128 $ & $ 0.7849 $ & $ 0.0349 $ & $ 0.1266 $ & $ 0.5699 $ & $ 0.0699 $ & $ 0.2729 $\\
  $   2000 $ & $ 4.9789 $ & $ -0.0211 $ & $ 0.2322 $ & $ 0.7641 $ & $ 0.0141 $ & $ 0.0890 $ & $ 0.5281 $ & $ 0.0281 $ & $ 0.1196 $\\
  $   4000 $ & $ 4.9921 $ & $ -0.0079 $ & $ 0.1516 $ & $ 0.7525 $ & $ 0.0025 $ & $ 0.0640 $ & $ 0.5118 $ & $ 0.0118 $ & $ 0.0773 $\\
  $   8000 $ & $ 4.9997 $ & $ -0.0003 $ & $ 0.1203 $ & $ 0.7533 $ & $ 0.0033 $ & $ 0.0421 $ & $ 0.5042 $ & $ 0.0042 $ & $ 0.0543 $\\
\bottomrule 
\end{tabularx}
\end{center}
{\footnotesize \it Median (Med.), median bias (Bias) and root median squared error (RMSE) of the MCL estimator. DGP: Poisson-IG IVT. The IVT process $X_t$ is simulated on the grid $t = \Delta, 2\Delta, \ldots, n\Delta$, with $\Delta = 0.10$, see Table \ref{tab:paramTab2} for  the values of the parameters used in the simulations. $K  = 3$. Number of Monte Carlo simulations: $500$.} 
\label{tab:poi_ig2}
\end{table}

	\begin{table}
\caption{\it CL estimation results: Poisson trawl process with $\Gamma$ trawl function}
\begin{center}
\footnotesize
\begin{tabularx}{1.00\textwidth}{@{\extracolsep{\stretch{1}}}rccccccccc@{}} 
\toprule
& \multicolumn{3}{c}{$\hat{\nu}$ ($\nu = 5$)} & \multicolumn{3}{c}{$\hat{H}$ ($H = 0.5$)} & \multicolumn{3}{c}{$\hat{\alpha}$ ($\alpha = 0.75$)} \\  
\cmidrule{2-4} \cmidrule{5-7} \cmidrule{8-10}  \\  
Nobs & Avg & Bias & RMSE & Avg & Bias & RMSE & Avg & Bias & RMSE \\ 
\midrule
  $    250 $ & $ 5.0958 $ & $ 0.0958 $ & $ 0.4897 $ & $ 0.6691 $ & $ 0.1691 $ & $ 0.3452 $ & $ 0.9620 $ & $ 0.2120 $ & $ 0.5434 $\\
  $    500 $ & $ 4.9864 $ & $ -0.0136 $ & $ 0.3323 $ & $ 0.6360 $ & $ 0.1360 $ & $ 0.2879 $ & $ 0.9710 $ & $ 0.2210 $ & $ 0.4283 $\\
  $   1000 $ & $ 4.9781 $ & $ -0.0219 $ & $ 0.2526 $ & $ 0.5957 $ & $ 0.0957 $ & $ 0.2109 $ & $ 0.8667 $ & $ 0.1167 $ & $ 0.3183 $\\
  $   2000 $ & $ 4.9866 $ & $ -0.0134 $ & $ 0.1776 $ & $ 0.5765 $ & $ 0.0765 $ & $ 0.1531 $ & $ 0.8480 $ & $ 0.0980 $ & $ 0.2443 $\\
  $   4000 $ & $ 4.9789 $ & $ -0.0211 $ & $ 0.1233 $ & $ 0.5519 $ & $ 0.0519 $ & $ 0.1156 $ & $ 0.8223 $ & $ 0.0723 $ & $ 0.1678 $\\
  $   8000 $ & $ 4.9720 $ & $ -0.0280 $ & $ 0.0945 $ & $ 0.5515 $ & $ 0.0515 $ & $ 0.0864 $ & $ 0.8147 $ & $ 0.0647 $ & $ 0.1234 $\\
\bottomrule 
\end{tabularx}
\end{center}
{\footnotesize \it Median (Med.), median bias (Bias) and root median squared error (RMSE) of the MCL estimator. DGP: Poisson-Gamma IVT. The IVT process $X_t$ is simulated on the grid $t = \Delta, 2\Delta, \ldots, n\Delta$, with $\Delta = 0.10$, see Table \ref{tab:paramTab2} for  the values of the parameters used in the simulations. $K = 3$. Number of Monte Carlo simulations: $500$.} 
\label{tab:poi_gam2}
\end{table}

	\begin{table}
\caption{\it CL estimation results: NB trawl process with exponential trawl function}
\begin{center}
\footnotesize
\begin{tabularx}{1.00\textwidth}{@{\extracolsep{\stretch{1}}}rccccccccc@{}} 
\toprule
& \multicolumn{3}{c}{$\hat{m}$ ($m = 2.1429$)} & \multicolumn{3}{c}{$\hat{p}$ ($p = 0.7$)}  & \multicolumn{3}{c}{$\hat{\lambda}$ ($\lambda = 1$)}  \\  
\cmidrule{2-4} \cmidrule{5-7} \cmidrule{8-10}  \\  
Nobs & Avg & Bias & RMSE & Avg & Bias & RMSE & Avg & Bias & RMSE \\ 
\midrule
  $    100 $ & $ 2.7614 $ & $ 0.6186 $ & $ 0.8445 $ & $ 0.6268 $ & $ -0.0732 $ & $ 0.1010 $ & $ 1.0021 $ & $ 0.0021 $ & $ 0.1902 $\\
  $    250 $ & $ 2.3957 $ & $ 0.2528 $ & $ 0.4848 $ & $ 0.6692 $ & $ -0.0308 $ & $ 0.0618 $ & $ 0.9921 $ & $ -0.0079 $ & $ 0.1206 $\\
  $    500 $ & $ 2.1957 $ & $ 0.0529 $ & $ 0.3009 $ & $ 0.6883 $ & $ -0.0117 $ & $ 0.0394 $ & $ 1.0038 $ & $ 0.0038 $ & $ 0.0791 $\\
  $   1000 $ & $ 2.2200 $ & $ 0.0771 $ & $ 0.2357 $ & $ 0.6912 $ & $ -0.0088 $ & $ 0.0309 $ & $ 1.0008 $ & $ 0.0008 $ & $ 0.0596 $\\
  $   2000 $ & $ 2.1887 $ & $ 0.0458 $ & $ 0.1632 $ & $ 0.6935 $ & $ -0.0065 $ & $ 0.0215 $ & $ 1.0030 $ & $ 0.0030 $ & $ 0.0410 $\\
  $   4000 $ & $ 2.1630 $ & $ 0.0201 $ & $ 0.1126 $ & $ 0.6967 $ & $ -0.0033 $ & $ 0.0144 $ & $ 0.9999 $ & $ -0.0001 $ & $ 0.0273 $\\
  $   8000 $ & $ 2.1557 $ & $ 0.0128 $ & $ 0.0805 $ & $ 0.6985 $ & $ -0.0015 $ & $ 0.0109 $ & $ 1.0012 $ & $ 0.0012 $ & $ 0.0210 $\\
\bottomrule 
\end{tabularx}
\end{center}
{\footnotesize \it Median (Med.), median bias (Bias) and root median squared error (RMSE) of the MCL estimator. DGP: Negative Binomial-Exponential IVT process. The IVT process $X_t$ is simulated on the grid $t = \Delta, 2\Delta, \ldots, n\Delta$, with $\Delta = 0.10$, see Table \ref{tab:paramTab2} for  the values of the parameters used in the simulations. $K = 1$. Number of Monte Carlo simulations: $500$.} 
\label{tab:nb_exp2}
\end{table}

	\begin{table}
\caption{\it CL estimation results: NB trawl process with IG trawl function}
\begin{center}
\footnotesize
\begin{tabularx}{1.00\textwidth}{@{\extracolsep{\stretch{1}}}rcccccccccccc@{}} 
\toprule
& \multicolumn{3}{c}{$\hat{m}$ ($m = 2.1429$)} & \multicolumn{3}{c}{$\hat{p}$ ($p = 0.7$)} & \multicolumn{3}{c}{$\hat{\delta}$ ($\delta = 0.75$)} & \multicolumn{3}{c}{$\hat{\gamma}$ ($\gamma = 0.5$)} \\  
\cmidrule{2-4} \cmidrule{5-7} \cmidrule{8-10}  \cmidrule{11-13}  \\  
Nobs & Avg & Bias & RMSE & Avg & Bias & RMSE & Avg & Bias & RMSE & Avg & Bias & RMSE \\ 
\midrule
  $    250 $ & $ 2.5049 $ & $ 0.3621 $ & $ 0.5605 $ & $ 0.6361 $ & $ -0.0639 $ & $ 0.0827 $ & $ 0.9671 $ & $ 0.2171 $ & $ 0.3323 $ & $ 0.6664 $ & $ 0.1664 $ & $ 0.2521 $\\
  $    500 $ & $ 2.3485 $ & $ 0.2056 $ & $ 0.4236 $ & $ 0.6688 $ & $ -0.0312 $ & $ 0.0632 $ & $ 0.8739 $ & $ 0.1239 $ & $ 0.2095 $ & $ 0.5738 $ & $ 0.0738 $ & $ 0.1402 $\\
  $   1000 $ & $ 2.2625 $ & $ 0.1197 $ & $ 0.3055 $ & $ 0.6767 $ & $ -0.0233 $ & $ 0.0450 $ & $ 0.7991 $ & $ 0.0491 $ & $ 0.1460 $ & $ 0.5280 $ & $ 0.0280 $ & $ 0.0986 $\\
  $   2000 $ & $ 2.1975 $ & $ 0.0547 $ & $ 0.2073 $ & $ 0.6937 $ & $ -0.0063 $ & $ 0.0303 $ & $ 0.7757 $ & $ 0.0257 $ & $ 0.1099 $ & $ 0.5147 $ & $ 0.0147 $ & $ 0.0769 $\\
  $   4000 $ & $ 2.1757 $ & $ 0.0328 $ & $ 0.1384 $ & $ 0.6936 $ & $ -0.0064 $ & $ 0.0199 $ & $ 0.7777 $ & $ 0.0277 $ & $ 0.0800 $ & $ 0.5203 $ & $ 0.0203 $ & $ 0.0573 $\\
  $   8000 $ & $ 2.1661 $ & $ 0.0232 $ & $ 0.0970 $ & $ 0.6978 $ & $ -0.0022 $ & $ 0.0152 $ & $ 0.7631 $ & $ 0.0131 $ & $ 0.0560 $ & $ 0.5127 $ & $ 0.0127 $ & $ 0.0380 $\\
\bottomrule 
\end{tabularx}
\end{center}
{\footnotesize \it Median (Med.), median bias (Bias) and root median squared error (RMSE) of the MCL estimator. DGP: Negative Binomial-IG IVT. The IVT process $X_t$ is simulated on the grid $t = \Delta, 2\Delta, \ldots, n\Delta$, with $\Delta = 0.10$, see Table \ref{tab:paramTab2} for  the values of the parameters used in the simulations. $K = 3$. Number of Monte Carlo simulations: $500$.} 
\label{tab:nb_ig2}
\end{table}

	\begin{table}
\caption{\it CL estimation results: NB trawl process with $\Gamma$ trawl function}
\begin{center}
\footnotesize
\begin{tabularx}{1.00\textwidth}{@{\extracolsep{\stretch{1}}}rcccccccccccc@{}} 
\toprule
& \multicolumn{3}{c}{$\hat{m}$ ($m = 2.1429$)} & \multicolumn{3}{c}{$\hat{p}$ ($p = 0.7$)} & \multicolumn{3}{c}{$\hat{H}$ ($H = 0.5$)} & \multicolumn{3}{c}{$\hat{\alpha}$ ($\alpha = 0.75$)} \\  
\cmidrule{2-4} \cmidrule{5-7} \cmidrule{8-10}  \cmidrule{11-13}  \\  
Nobs & Avg & Bias & RMSE & Avg & Bias & RMSE & Avg & Bias & RMSE & Avg & Bias & RMSE \\ 
\midrule
  $    500 $ & $ 2.6635 $ & $ 0.5206 $ & $ 0.5684 $ & $ 0.6240 $ & $ -0.0760 $ & $ 0.0873 $ & $ 3.0978 $ & $ 2.5978 $ & $ 2.5978 $ & $ 4.6910 $ & $ 3.9410 $ & $ 3.9414 $\\
  $   1000 $ & $ 2.4640 $ & $ 0.3212 $ & $ 0.4184 $ & $ 0.6510 $ & $ -0.0490 $ & $ 0.0608 $ & $ 1.1213 $ & $ 0.6213 $ & $ 0.6213 $ & $ 1.6475 $ & $ 0.8975 $ & $ 0.8975 $\\
  $   2000 $ & $ 2.3565 $ & $ 0.2136 $ & $ 0.3091 $ & $ 0.6659 $ & $ -0.0341 $ & $ 0.0461 $ & $ 0.7854 $ & $ 0.2854 $ & $ 0.2887 $ & $ 1.1460 $ & $ 0.3960 $ & $ 0.4207 $\\
  $   4000 $ & $ 2.2840 $ & $ 0.1411 $ & $ 0.1944 $ & $ 0.6739 $ & $ -0.0261 $ & $ 0.0327 $ & $ 0.7204 $ & $ 0.2204 $ & $ 0.2391 $ & $ 1.0852 $ & $ 0.3352 $ & $ 0.3518 $\\
  $   8000 $ & $ 2.2302 $ & $ 0.0873 $ & $ 0.1583 $ & $ 0.6843 $ & $ -0.0157 $ & $ 0.0256 $ & $ 0.7034 $ & $ 0.2034 $ & $ 0.2142 $ & $ 1.0659 $ & $ 0.3159 $ & $ 0.3259 $\\
\bottomrule 
\end{tabularx}
\end{center}
{\footnotesize \it Median (Med.), median bias (Bias) and root median squared error (RMSE) of the MCL estimator. DGP: Negative Binomial-Gamma IVT process. The IVT process $X_t$ is simulated on the grid $t = \Delta, 2\Delta, \ldots, n\Delta$, with $\Delta = 0.10$, see Table \ref{tab:paramTab2} for  the values of the parameters used in the simulations. $K = 3$. Number of Monte Carlo simulations: $500$.} 
\label{tab:nb_gam2}
\end{table}

	\begin{figure}[!t]
		\centering
		\includegraphics[scale=0.95]{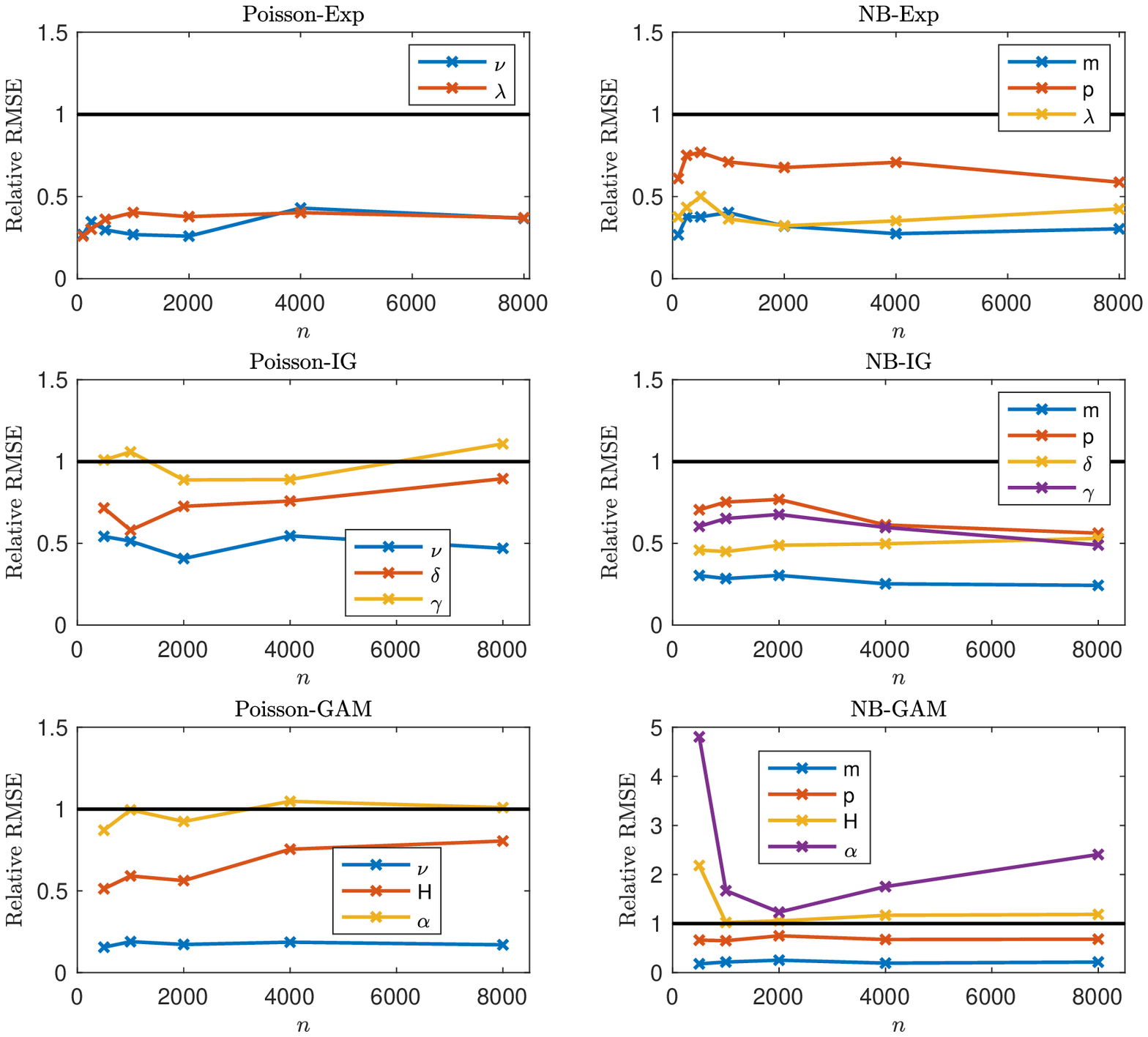} 
		\caption{\it  Root median square error (RMSE) of the MCL estimator divided by the RMSE of the GMM estimator. The underlying IVT process $X_t$ is simulated on the grid $t = \Delta, 2\Delta, \ldots, n\Delta$, with $\Delta = 0.10$, see Table \ref{tab:paramTab2} for  the values of the parameters used in the simulations. For the Poisson-Exp and NB-Exp we set $K = 1$; for the other DGPs we set $K = 3$. We also conducted the comparison with $K=5$, as suggested in \cite{BNLSV2014} with similar results (results not presented here but available from the authors upon request). }
		\label{fig:CLvsMM2}
	\end{figure}

	\clearpage

	\newpage
	\section{Empirical study}\label{sec:SuppEmp}
	This section contains additional details on the data pre-processing used in the empirical study, see Subsection \ref{sec:EmpSetUp}, and additional forecasting results, see Subsection \ref{sec:add_forec}. 
	\subsection{Details for the empirical study}\label{sec:EmpSetUp}
	We will now provide some additional details on the data preprocessing carried out for the empirical study. 
	In the article, we analyse
	the time series of the bid-ask spread, measured in U.S. dollar cents, of the Agilent Technologies Inc. stock (ticker: A) on a single day, May $4$, $2020$. 
	The A stock is traded on the New York Stock Exchange, which is open from $9$:$30$ AM to $4$ PM. To avoid opening effects, we consider the data from $10$:$30$ AM to $4$ PM, i.e.~we discard the first $60$ minutes of the day. Our data is gathered from the Trade and Quote database and cleaned using the approach proposed in \cite{BNHLS2009}. The data is available at a very high frequency but to obtain equidistant data, we sample the observations with $\Delta = \frac{1}{12} $ minutes (i.e.~$5$ seconds) time steps, using the previous tick approach, starting at $10$:$30$ AM, resulting in $n = 3961$ observations.  

	
	\subsection{Additional forecasting results}\label{sec:add_forec}
	Figures \ref{fig:ncif}--\ref{fig:ncif3} report forecasting results analogous to those of Section \ref{sec:fSpr} 
	in the main paper, but now using the conditional mode, instead of the conditional mean, as a point forecast. That is, using the notation of Section 6.1 in the main paper, we set  $\hat x_{i|i-h} = \arg \max_k \widehat \PP(X_{i|i-h} = k)$, where $\widehat \PP$ is the estimated predictive PMF of the IVT model. 
	
	
	\begin{figure}
		\captionsetup[subfigure]{aboveskip=-4pt,belowskip=-4pt}
		\subfloat[NB-Gamma versus Poissonian INAR(1)	\label{fig:ncif}]{	\includegraphics[width=\textwidth, height=0.2\textheight]{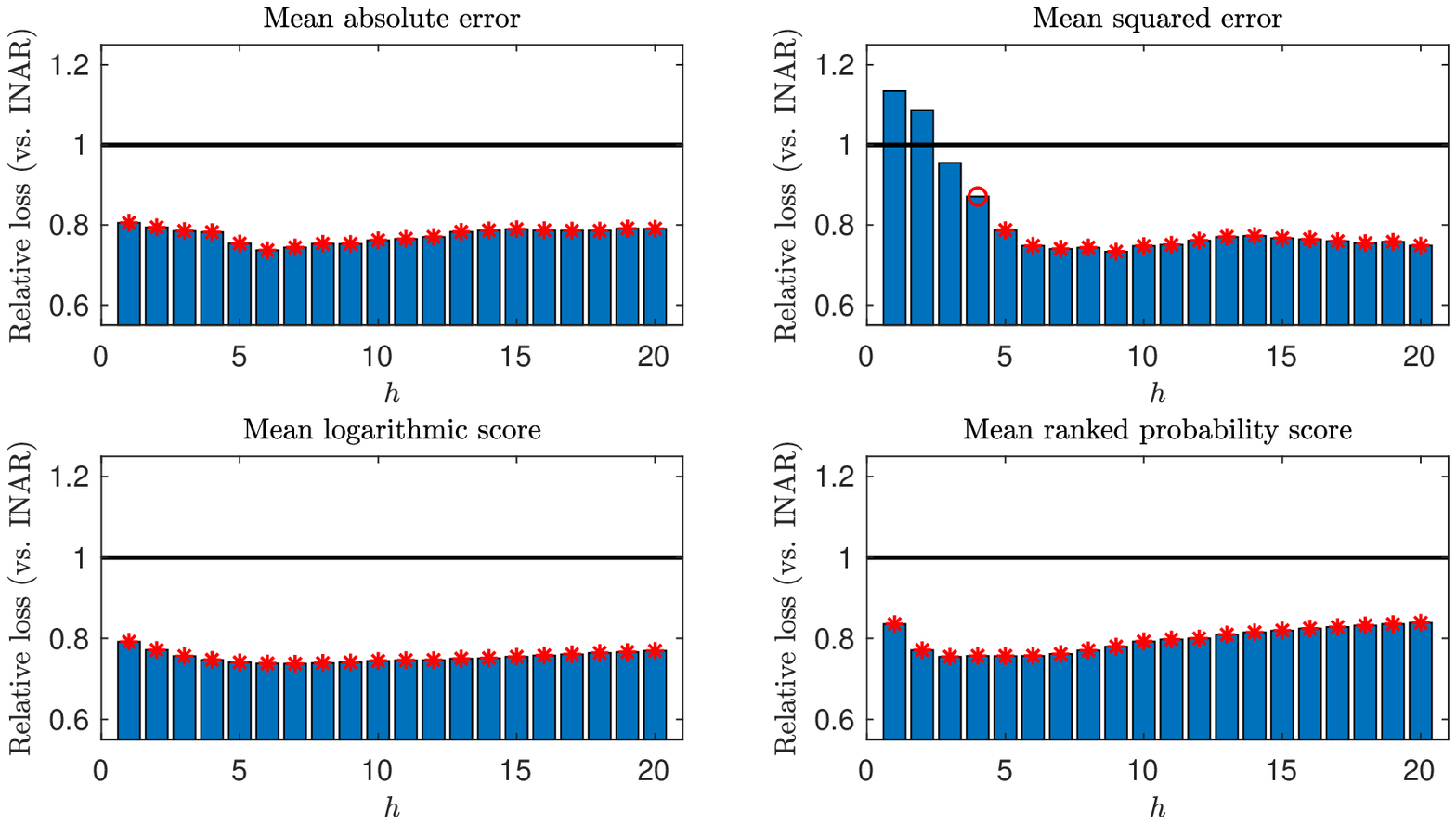} } \\
		\captionsetup[subfigure]{aboveskip=-4pt,belowskip=-4pt}
		\subfloat[NB-Gamma versus Poisson-Gamma model\label{fig:ncif2}]{	\includegraphics[width=\textwidth, height=0.2\textheight]{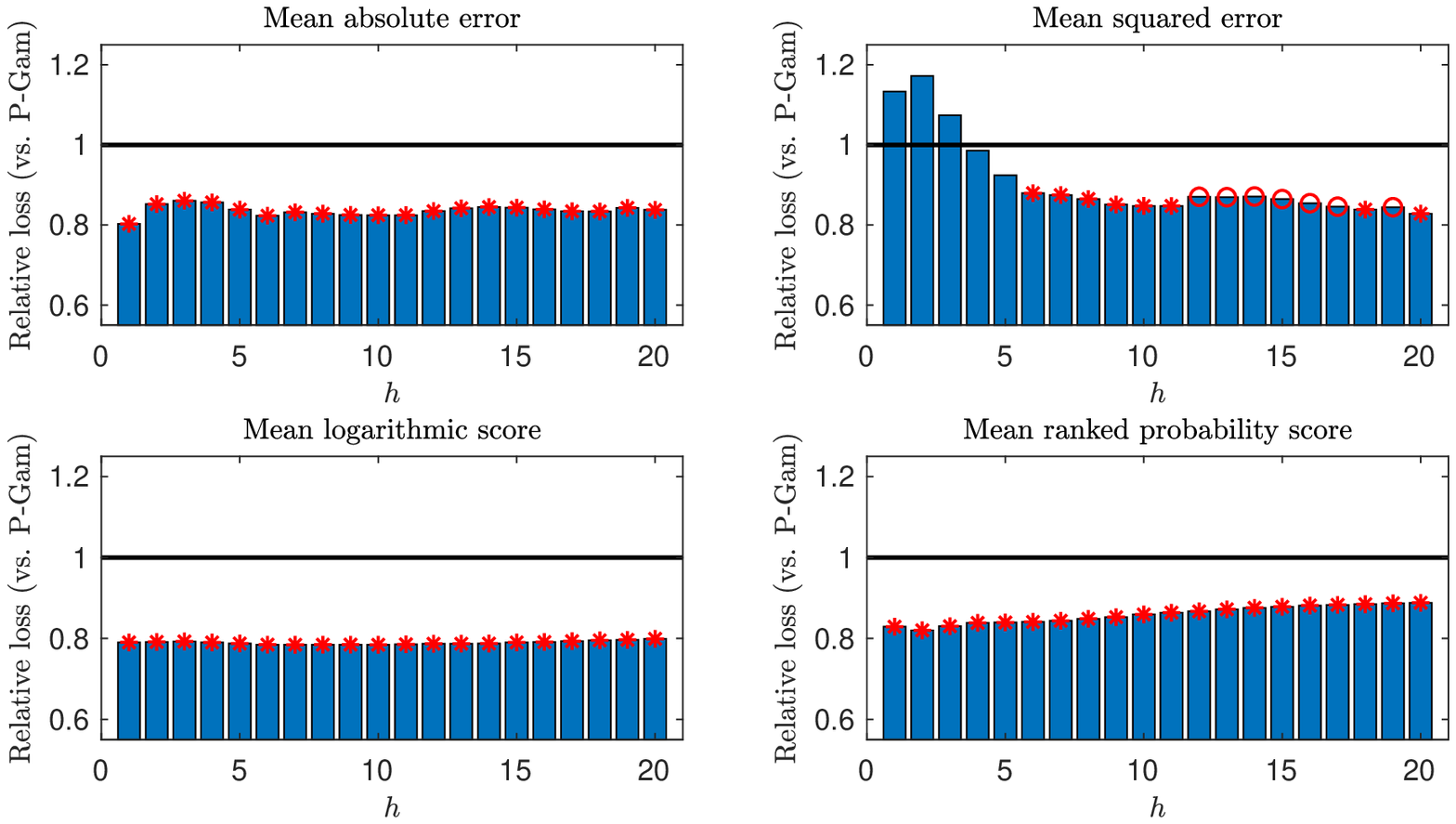}} 	\\
		\captionsetup[subfigure]{aboveskip=-4pt,belowskip=-4pt}
		\subfloat[NB-Gamma versus NB-Exponential\label{fig:ncif3}]{	\includegraphics[width=\textwidth, height=0.2\textheight]{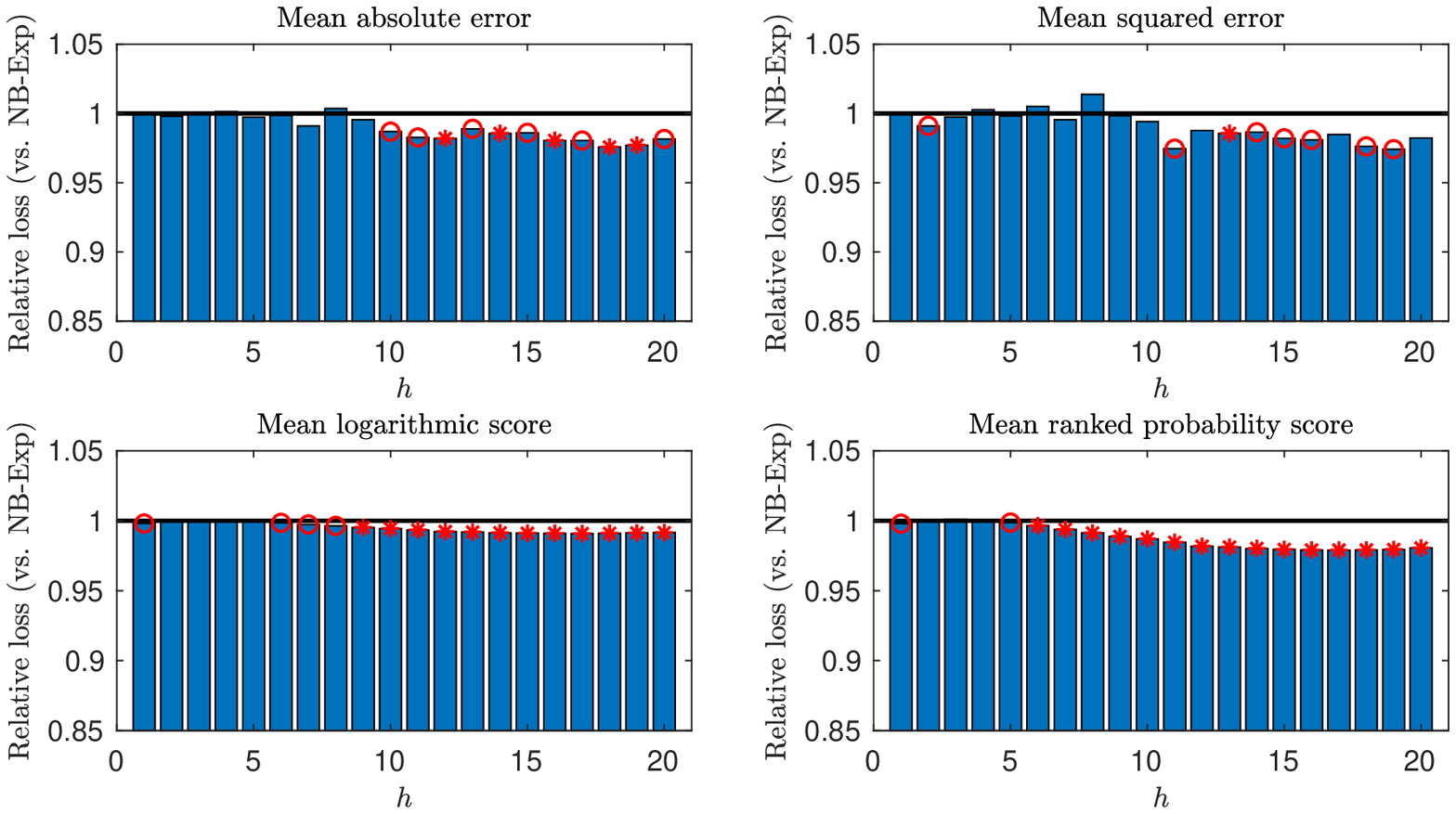}} 	
		\caption{\it Forecasting the spread level of the A stock on May 4, 2020. Four different loss metrics and twenty forecast horizons, $h = 1, 2, \ldots, 20$. The numbers plotted are relative average losses of the NB-Gamma forecasting model, 
			compared with the (a) Poissonian INAR(1) model, the (b) Poisson-Gamma model, and the (c) NB-Exponential, over $n_{oos} = 720$ out-of-sample forecasts. 
			A circle above the bars indicates rejection null of equal forecasting performance between the two models, against the alternative that the NB-Gamma model provides superior forecasts, using the Diebold-Mariano \citep{DM1995} test at a $5\%$ level; an asterisk denotes rejection at a $1\%$ level.}
		\label{fig:ncif_all_supp}
	\end{figure}

	
	\clearpage
	
	\newpage

	\clearpage
	
	\newpage
	
	\section{Details concerning integer-valued L\'evy bases}\label{sec:levyBases}
	
	\subsection{Poisson L\'evy basis}
	
	Consider the case where the L\'evy basis is Poisson, i.e.~$L' \sim \textnormal{Poi}(\nu)$ for some intensity $\nu > 0.$  For a bounded Borel set $B$ with $Leb(B)< \infty,$ we have
	\begin{align*}
		L(B) \sim \textnormal{Poi}(\nu Leb(B)).
	\end{align*}
	The cumulants, in this case, are $\kappa_j = \nu$ for all $j\geq 0$.
	
	\subsection{Negative binomial L\'evy basis}\label{sec:NB}
	We follow \cite{OlePollardNeil12,BNLSV2014} and denote by $NB(m,p)$ the negative binomial law with parameters $m \in \mathbb{N}$ and $p\in (0,1).$ Recall, that a negative binomial random variable is positively valued and can be interpreted as the number of successes, $k,$ until $m$ failures in a sequence of iid Bernoulli trials, each with the probability of success $p.$ Let $L' \sim NB(m,p);$ it holds that
	\begin{align*}
		P(L' = k) = \frac{\Gamma(m+k)}{k!\Gamma(m)} (1-p)^m p^k, \quad k = 0, 1, 2, \ldots.
	\end{align*}
	As is well known, we have that $L'_t \sim NB(mt,p)$ and therefore, for a Borel set $B$, it holds that $L(B) \sim NB(Leb(B)m,p)$, which implies
	\begin{align*}
		P(L(B) = k) =  \frac{\Gamma(Leb(B)m+k)}{k!\Gamma(Leb(B)m)} (1-p)^{Leb(B)m} p^k, \quad k = 0, 1, 2, \ldots.
	\end{align*}
	Here the relevant cumulants are $\kappa_1 = \frac{pm}{1-p}$, $\kappa_2 = \frac{pm}{(1-p)^2}$ and $\kappa_4 =m \frac{p + 4p^2 + p^3}{(1-p)^4}$.
	
	\subsection{Skellam L\'evy basis}\label{sec.:pairSkellam}
	The Skellam distribution is the distribution of the difference of two Poisson processes $N_t^+$ and $N_t^-$ and is therefore integer-valued. Let $N_t^{\pm} \sim \textnormal{Poi}(\psi^{\pm})$  with $\psi^{\pm}>0$; then $S := N_t^+ - N_t^- \sim \textnormal{Skellam}(\psi^+,\psi^-).$  Further, the Skellam L\'evy process  $(L_t')_{t\geq 0}$  with $L_1' \sim \textnormal{Skellam}(\psi^+,\psi^-)$  has the marginal distribution $L_t' \sim \textnormal{Skellam}(t\psi^+,t\psi^-)$  \citep{OlePollardNeil12}, meaning that for a Borel set $B,$ we have  $L(B) \sim \textnormal{Skellam}(Leb(B)\psi^+,Leb(B)\psi^-).$ The PMF of the random variable $X \sim \textnormal{Skellam}(\psi^+,\psi^-)$ is given by
	\begin{align*}
		g(k;\psi^+,\psi^-) := P(X = k) = e^{-(\psi^+ + \psi^-)} \left( \frac{\psi^+}{\psi^-}\right)^{k/2} I_k\left(2\sqrt{\psi^+ \psi^-}\right),
	\end{align*}
	where $I_{\nu}(x)$ is the modified Bessel function of the first kind (see e.g. \cite{AbramowitzStegun}) with parameter $\nu$ evaluated at $x.$ In the symmetric case, $\psi^+ = \psi^- = \psi,$ this reduces to $g(k;\psi) := e^{-2\psi} I_k(2\psi)$. The cumulants are easily seen to be $\kappa_j = \psi^+ - \psi^-$ for $j$ odd and $\kappa_j = \psi^+ + \psi^-$ for $j$ even.
	
	\subsection{$\Delta NB$ L\'evy basis}
	Analogous to the Skellam process, we can consider the difference of two L\'evy seeds which have negative binomials as their laws; \cite{OlePollardNeil12} call this a $\Delta NB$ L\'evy process. Let $L^{\pm} \sim NB(m^{\pm},p^{\pm})$ be independent L\'evy seeds with negative binomial laws. \cite{OlePollardNeil12} show that for $k \geq 0,$ the difference L\'evy seed $L' = L^+ - L^-$ has PMF
	\begin{align}\label{eq.:DNBpmf}
		P(L' = k) = (1-p^+)^{m^+} (1-p^-)^{m^-}\frac{(p^+)^k (m^+)_k}{k!} F(m^+ + k, m^-; k+1; p^+ p^-),
	\end{align}
	where 
	\begin{align*}
		F(\alpha,\beta;\gamma; z) = \sum_{n=0}^{\infty} \frac{(\alpha)_n (\beta)_n}{(\gamma)_n} \frac{z^n}{n!}, \quad z \in [0,1), \quad \alpha,\beta,\gamma >0,
	\end{align*}
	is the hypergeometric function, see e.g. \cite{AbramowitzStegun}, and $(\alpha)_n = \frac{\Gamma(\alpha + n)}{\Gamma(\alpha)}$ is the Pochhammer symbol. The PMF for $k \leq 0$ is, by symmetry, given as \ref{eq.:DNBpmf}, mutatis mutandis. The resulting distribution is denoted as $L' \sim \Delta NB(m^+,p^+,m^-,p^-)$ and it is easy to show that \citep{OlePollardNeil12} the L\'evy process corresponding to $L'$ has marginal distribution $L_t' \sim \Delta NB(tm^+,p^+,tm^-,p^-),$ meaning that we have for a Borel set $B$,
	\begin{align*}
		L(B) \sim NB(Leb(B) m^+,p^+,Leb(B) m^-,p^-).
	\end{align*}
	The cumulants for the $\Delta NB$ L\'evy seed are easily deduced from those of the negative binomial ones, recalling that the $\Delta NB$ law is the difference of two independent NB random variables.

	\clearpage 
	
	\newpage
	
	\section{Details concerning parametric trawl functions}\label{sec:trawls}
	
	The expressions for the likelihoods in the previous section reveal that we are interested in calculating expressions such as $Leb(A_t \backslash A)$ and $Leb(A_t \cap A)$ for different trawl functions. In this section we derive the required results for various trawls based on the superposition trawl function $d(s) = \int_0^{\infty} e^{\lambda s} \pi (d\lambda)$, $s \leq 0$, see also the main paper.
	
	\subsection{The exponential trawl}
	
	The case where the measure $\pi$ has an atom at $\lambda > 0,$  i.e.~$\pi(dx) =\delta_{\lambda}(dx),$ where $\delta_y(\cdot)$ is the Dirac delta function at $y \in \mathbb{R}_+,$ we get $d(s) = e^{\lambda s}.$ Consequently, for $t \geq 0,$
	\begin{align*}
		Leb(A) = \lambda^{-1}, \quad Leb(A_t \backslash A) = \lambda^{-1}(1-e^{-\lambda t}), \quad Leb(A_t \cap A) = \lambda^{-1}e^{-\lambda t}.
	\end{align*}
	This implies the correlation function
	\begin{align*}
		\rho(h) = \exp(-\lambda h), \quad h>0.
	\end{align*}
	
	\subsection{The finite superposition exponential trawl}\label{sec:supExp}
	Let $\pi$ have finitely many atoms, i.e.~$\pi(dx) = \sum_{i=1}^q w_i \delta_{\lambda_i}(dx)$ for $q \in \mathbb{N}.$ Then
	\begin{align*}
		d(s) = \sum_{i=1}^q w_i e^{\lambda_i s}, 
	\end{align*}
	and
	\begin{align*}
		Leb(A) = \sum_{i=1}^q w_i \lambda_i ^{-1}, \quad Leb(A_t \backslash A) = \sum_{i=1}^q w_i\lambda_i^{-1}(1-e^{-\lambda_i t}), \quad Leb(A_t \cap A) = \sum_{i=1}^q w_i \lambda_i^{-1}e^{-\lambda_i t}.
	\end{align*}
	This implies the correlation function
	\begin{align*}
		\rho(h) = \left( \sum_{i=1}^q \omega_i \lambda_i^{-1} \right)^{-1} \sum_{i=1}^q \omega_i \lambda_i^{-1} \exp(-\lambda_i h), \quad h>0.
	\end{align*}

	\subsection{The GIG trawl}\label{sec.:GIGtrawl}
	A flexible class of trawl functions can be specified through the \emph{generalized inverse Gaussian} (GIG) density function (see e.g. \cite{BNLSV2014}),
	\begin{align*}
		f_{\pi}(x) = \frac{(\gamma/\delta)^{\nu}}{2 K_{\nu}(\delta \gamma)} x^{\nu - 1} \exp \left( -\frac{1}{2} (\delta^2 x^{-1} + \gamma^2 x) \right),
	\end{align*}
	where $\nu \in \mathbb{R}$ and $\gamma, \delta \geq 0$ with both not equal to zero simultaneously. $K_{\nu}(x)$ is the modified Bessel function of the third kind with parameter $\nu,$ evaluated at $x$ (e.g. \cite{AbramowitzStegun}). Suppose now, that $\pi$ has density $f_{\pi},$ i.e.~$\pi(d\lambda) = f_{\pi}(\lambda) d\lambda.$ For $s \leq 0,$ the trawl function becomes
	\begin{align*}
		d(s) = \int_0^{\infty} e^{\lambda s} f_{\pi} (\lambda) d\lambda = \left( 1- \frac{2s}{\gamma^2}\right)^{-\nu/2} \frac{K_{\nu}\left(\delta \gamma \alpha_s \right)}{K_{\nu}(\delta \gamma)},
	\end{align*}
	whereas
	\begin{align*}
		& Leb(A) =\frac{\gamma}{\delta} \frac{K_{\nu-1}(\delta \gamma)}{K_{\nu}(\delta \gamma)}, \quad Leb(A_t  \cap A) = \frac{\gamma \alpha_t^{-\nu +1}}{\delta} \frac{K_{\nu-1}(\delta \gamma \alpha_t)}{K_{\nu}(\delta \gamma)},
	\end{align*}
	and
	\begin{align*}
		Leb(A_t \backslash A) =  \frac{\gamma}{\delta K_{\nu}(\delta \gamma)}\left( K_{\nu-1}(\delta \gamma) - \alpha_t^{-\nu + 1}K_{\nu-1}(\delta \gamma\alpha_t)\right),
	\end{align*}
	where $\alpha_t := \sqrt{\frac{2t}{\gamma^2} + 1}.$ This implies the correlation function
	\begin{align*}
		\rho(h) =  \alpha_h^{-\nu +1} \frac{K_{\nu-1}(\delta \gamma \alpha_h)}{K_{\nu-1}(\delta \gamma)}, \quad h>0.
	\end{align*}


	\subsection{The IG trawl}\label{sec:IG}
	The inverse Gaussian distribution is a special case of the GIG distributions, where $\nu = \frac{1}{2}.$ In this case, the trawl function simplifies to
	\begin{align*}
		d(s) =  \left( 1- \frac{2s}{\gamma^2}\right)^{-1/2} \exp\left( \delta \gamma \left( 1- \sqrt{ 1 - \frac{2s}{\gamma^2}} \right) \right), \quad s \leq 0,
	\end{align*}
	which means that
	\begin{align*}
		Leb(A) = \frac{\gamma}{\delta}, \quad Leb(A_t \cap A) =\frac{\gamma}{\delta} e^{\delta \gamma ( 1 - \alpha_t)}, \quad Leb(A_t  \backslash  A) =\frac{\gamma}{\delta} \left( 1- e^{\delta \gamma ( 1 - \alpha_t)}\right),
	\end{align*}
	where again $\alpha_t = \sqrt{\frac{2t}{\gamma^2} + 1}.$ This implies the correlation function
	\begin{align*}
		\rho(h) =\exp(\delta \gamma ( 1 - \alpha_h)), \quad h>0.
	\end{align*}
	
	\subsection{The $\Gamma$ trawl}
	An interesting case, capable of generating long memory in the trawl process, is given by the $\Gamma$ trawl. Suppose that $\pi$ has the $\Gamma(1+H,\alpha)$ density,
	\begin{align*}
		f_{\pi}(\lambda) = \frac{1}{\Gamma(1+H)} \alpha^{1+H} \lambda^{H} e^{-\lambda \alpha},
	\end{align*}
	where $\alpha > 0 $ and $H > 0.$ Now, 
	\begin{align*}
		d(s) = \left( 1- \frac{s}{\alpha}\right)^{-(H+1)}, \quad s \leq 0,
	\end{align*}
	which implies
	\begin{align*}
		Leb(A) =\frac{\alpha}{H}, \quad Leb(A_t  \cap  A) =\frac{\alpha}{H} \left(1 + \frac{t}{\alpha}\right)^{-H}, \quad Leb(A_t \backslash A) = \frac{\alpha}{H}\left(1- \left(1 + \frac{t}{\alpha}\right)^{-H}\right).
	\end{align*}
	
	This yields the correlation function
	\begin{align*}
		\rho(h) = Corr( L(A_{t+h}),L(A_t)) = \frac{ Leb( A_{h} \cap A)}{Leb(A)} =  \left(1 + \frac{h}{\alpha}\right)^{-H},
	\end{align*}
	so that
	\begin{align*}
		\int_0^{\infty} \rho(h) dh = \left\{
		\begin{array}{lr}
			\infty & \textnormal{if} \quad H \in (0,1],\\
			\frac{\alpha}{H-1} & \textnormal{if} \quad H > 1,
		\end{array}
		\right.
	\end{align*}
	from which we see, that the trawl process has long memory for  $H \in (0,1].$

	\clearpage
	
	\newpage

	\section{Details concerning gradients}\label{sec:gradients}
	Recall that we have the composite log-likelihood function
	\begin{align*}
		l_{CL}(\theta;x) := l_{CL}^{(K)}(\theta;x)  = \log L_{CL}^{(K)}(\theta;x) = \sum_{k=1}^K \sum_{i=1}^{n-k} \log f(x_{i+k},x_i;\theta).
	\end{align*}
	Let $\theta_i$ be an element of $\theta$. The derivative of $l_{CL}(\theta;x)$ wrt. $\theta_i$ is 
	\begin{align}\label{eq:grad1}
		\frac{\partial}{\partial \theta_i} l_{CL}(\theta;x)  = \frac{\partial}{\partial \theta_i} \log L_{CL}^{(K)}(\theta;x) = \sum_{k=1}^K \sum_{i=1}^{n-k} \frac{1}{f(x_{i+k},x_i;\theta)} \frac{\partial}{\partial \theta_i} f(x_{i+k},x_i;\theta).
	\end{align}
	
	Recall also that
	\begin{align*}
		f(x_{i+k},x_i;\theta) &= \sum_{c=-\infty}^{\infty} P_{1,i,k}^{(c)} \cdot P_{2,i,k}^{(c)} \cdot P_{3,k}^{(c)}
	\end{align*}
	with
	\begin{align*}
		P_{1,i,k}^{(c)} := \mathbb{P}\left(  L(A_{k\Delta} \setminus A) = x_{i+k}-c  \right), \quad P_{2,i,k}^{(c)} =  \mathbb{P}\left(  L(A_{k\Delta} \setminus A) = x_{i}-c  \right), \quad P_{3,k}^{(c)} = \mathbb{P}\left(  L(A_{k\Delta}\cap A) = c  \right),  
	\end{align*}
	implying that
	\begin{align}\label{eq:grad2}
		\frac{\partial}{\partial \theta_i}  f(x_{i+k},x_i;\theta) &= \sum_{c=-\infty}^{\infty} \left( \frac{\partial}{\partial \theta_i}  P_{1,i,k}^{(c)} \cdot P_{2,i,k}^{(c)} \cdot P_{3,k}^{(c)}  + P_{1,i,k}^{(c)} \cdot \frac{\partial}{\partial \theta_i} P_{2,i,k}^{(c)} \cdot P_{3,k}^{(c)} + P_{1,i,k}^{(c)} \cdot P_{2,i,k}^{(c)} \cdot\frac{\partial}{\partial \theta_i}  P_{3,k}^{(c)}\right).
	\end{align}
	The terms $P_{1,i,k}^{(c)}, P_{2,i,k}^{(c)}, P_{3,k}^{(c)}$ are calculated in the numerical maximization of the composite likelihood routine for all $c$. The aim of this section is to calculate $\frac{\partial}{\partial \theta_i}  P_{j,i,k}^{(c)}$ for $j = 1, 2, 3$, so that the gradient of the log-likelihood function is easily calculated using Equations \eqref{eq:grad1} and \eqref{eq:grad2}. It is clear that $\frac{\partial}{\partial \theta_i}  P_{j,i,k}^{(c)}$ will depend on both the L\'evy basis as well as the form of the trawl set (and hence the trawl function). We first supply the relevant derivations for the Poisson L\'evy basis (Section \ref{app:gradPoi}) and the Negative Binomial L\'evy basis (Section \ref{app:gradNB}), and then the trawl functions Exp, SupExp, IG, and $\Gamma$ (Sections \ref{app:gradExp}--\ref{app:gradGam}).

	\subsection{Some preliminary practical details}
	In our numerical implementation of the composite likelihood methods, we often have restrictions on some parameters. Most notably, we have positivity restriction, e.g. we require that the intensity $\nu>0$ for the Poisson L\'evy basis. One could impose such restrictions by using a constrained optimization procedure when performing the numerical optimization of the log composite likelihood function $ l_{CL}(\theta;x)$. We prefer to work with an unconstrained optimization procedure, by transforming the parameters such that they are fulfilling their restrictions automatically. That is if $\theta$ is a restricted parameter, we find an invertible transformation function $g$, such that $\tilde\theta = g^{-1}(\theta) \in \R$ is unrestricted. The unconstrained numerical optimizer is optimizing over the unrestricted parameter $\tilde\theta$ and arrives at, say, $\tilde\theta^*$. Our estimate of $\theta$ is thus $\hat \theta = g(\tilde \theta^*)$. Consequently, it is necessary to correct  for this when calculating standard errors (delta rule) as well as when supplying a gradient for our numerical optimization scheme. The reason is that the calculations concerning the gradient, detailed in the previous section, are with respect to $\theta$, and not $\tilde \theta$, which is the actual parameter being used in the numerical optimization procedure. In the case of a transformed variable, the gradient that should be supplied to the machine is therefore not the one given in \eqref{eq:grad2}, but rather
	\begin{align*}
		\frac{\partial}{\partial \tilde\theta_i}  f(x_{i+k},x_i;\theta) = \frac{\partial}{\partial \theta_i}  f(x_{i+k},x_i;\theta)  \frac{\partial \theta_i}{\partial \tilde\theta_i} =  \frac{\partial}{\partial \theta_i}  f(x_{i+k},x_i;\theta)  \frac{\partial }{\partial \tilde\theta_i}  g(\tilde \theta).
	\end{align*}
	
	In this paper two restrictions are encountered: many parameters are positive, while a few are restricted to be in the unit interval. If $\theta>0$ is a positive parameter, we use a log transformation by defining the new parameter $\tilde\theta$ through
	\begin{align*}
		\tilde\theta = g^{-1}(\theta) = \log \theta, \qquad \theta = g(\tilde\theta) = \exp (\tilde\theta).
	\end{align*}
	If $p \in (0,1)$ is a parameter, we use an inverse logistic (sigmoid) transformation,
	\begin{align*}
		\tilde p= g^{-1}(p) = \log \left( \frac{p}{1-p} \right), \qquad  p = g(\tilde p) =  \frac{1}{1+\exp(-\tilde p)}.
	\end{align*}

	\subsection{Poisson L\'evy basis} \label{app:gradPoi}
	Let $L' \sim Poi(\nu)$ and recall that  for a Borel set $B$, this implies
	\begin{align*}
		\mathbb{P}\left( L(B) = x \right) = \frac{[\nu Leb(B)]^x \exp(-\nu Leb(B))}{x!}.
	\end{align*}
	We deduce, that for a generic parameter $\theta \neq \nu$,
	\begin{align*}
		\frac{\partial}{\partial \theta}\mathbb{P}\left( L(B) = x \right) =  \left( x Leb(B)^{-1} - \nu \right) \mathbb{P}\left( L(B) = x \right)  \frac{\partial}{\partial \theta}Leb(B).
	\end{align*}
	The only ingredient left to calculate is
	\begin{align*}
		\frac{\partial}{\partial \nu} \mathbb{P}\left( L(B) = x \right) = \left(x\nu^{-1} - Leb(B)\right) \mathbb{P}\left( L(B) = x \right). 
	\end{align*}

	\subsection{Negative Binomial L\'evy basis}\label{app:gradNB}
	Recall that in the case where the L\'evy seed $L'$ is distributed as a Negative Binomial random variable with parameters $m>0$ and $p \in [0,1]$, we have $L(B) \sim NB(Leb(B)m,p)$, which implies
	\begin{align*}
		P(L(B) = x) =  \frac{\Gamma(Leb(B)m+x)}{x!\Gamma(Leb(B)m)} (1-p)^{Leb(B)m} p^x, \quad x = 0, 1, 2, \ldots.
	\end{align*}
	Using the well-known property of the $\Gamma$ function that $\Gamma(x+1) = x\Gamma(x)$ \citep[][p. 904]{integrals}, we can write
	\begin{align*}
		P(L(B) = x) =  \left(Leb(B)m + x-1\right) \left(Leb(B)m + x-2\right) \cdots \left(Leb(B)m\right)   \frac{1}{x!} (1-p)^{Leb(B)m} p^x,
	\end{align*}
	for $ k = 0, 1, 2, \ldots$. We deduce, that for a generic parameter $\theta \neq m, p$,
	\begin{align*}
		\frac{\partial}{\partial \theta}& P(L(B) = x) \\
		&= \left(  \frac{\partial}{\partial \theta} Leb(B)  \right)mP(L(B) = x)  \left( \log(1-p) +  \frac{1}{Leb(B)m} +   \frac{1}{Leb(B)m+1} + \cdots +\frac{1}{Leb(B)m + x-1}  \right) .
	\end{align*}
	The only ingredients left to calculate are
	\begin{align*}
		\frac{\partial}{\partial p}&P(L(B) = x) =    P(L(B) = x)  \left( \frac{x}{p} - \frac{Leb(B)m}{1-p}  \right), 
	\end{align*}
	and
	\begin{align*}
		&\frac{\partial}{\partial m}P(L(B) = x) \\
		&=    P(L(B) = x)  Leb(B) \left( \log(1-p) +  \frac{1}{Leb(B)m} +   \frac{1}{Leb(B)m+1} + \cdots +\frac{1}{Leb(B)m + x-1}  \right) .
	\end{align*}

	\subsection{Exponential trawl function}\label{app:gradExp}
	Let $L'$ be a generic L\'evy seed and $d(s) = \exp(\lambda s)$ for $s \leq 0$. Recall that for $t>0$,
	\begin{align*}
		Leb(A_t \setminus A) = \lambda^{-1} (1-\exp(-\lambda t)), \qquad Leb(A_t \cap A) = \lambda^{-1} \exp(-\lambda t).
	\end{align*}
	It is not difficult to show that
	\begin{align*}
		\frac{\partial}{\partial \lambda} Leb(A_t \setminus A) = \lambda^{-1}\left( t\exp(-\lambda t) - \lambda^{-1} (1-\exp(-\lambda t))  \right),
	\end{align*}
	while
	\begin{align*}
		\frac{\partial}{\partial \lambda} Leb(A_t \cap A) = -\lambda^{-1}\exp(-\lambda t)\left( \lambda^{-1} +  t  \right).
	\end{align*}

	\subsection{SupExp trawl function}\label{app:gradSupExp}
	Let $L'$ be a generic L\'evy seed and $d(s)$ be the supExp trawl function (see above). Recall that for $t>0$,
	\begin{align*}
		Leb(A_t \setminus A) = \sum_{i=1}^q w_i\lambda_i^{-1}(1-e^{-\lambda_i t}), \qquad Leb(A_t \cap A) = \sum_{i=1}^q w_i \lambda_i^{-1}e^{-\lambda_i t}.
	\end{align*}
	
	It is not difficult to show that for $j = 1, 2, \ldots, q,$
	\begin{align*}
		\frac{\partial}{\partial \lambda_j} Leb(A_t \setminus A) &= w_j \lambda_j^{-1}\left( t\exp(-\lambda_j t) - \lambda_j^{-1} (1-\exp(-\lambda_j t))  \right), \\
		\frac{\partial}{\partial \lambda_j} Leb(A_t \cap A) &= - w_j\lambda_j^{-1}\exp(-\lambda_j t)\left( \lambda_j^{-1} +  t  \right),
	\end{align*}
	while
	\begin{align*}
		\frac{\partial}{\partial w_j} Leb(A_t \setminus A) &= \lambda_j^{-1}(1-e^{-\lambda_j t}), \\
		\frac{\partial}{\partial w_j} Leb(A_t \cap A) &= \lambda_j^{-1}e^{-\lambda_j t}.
	\end{align*}

	\subsection{IG trawl function}\label{eq:gradIG}
	Let $L'$ be a generic L\'evy seed and $d(s)$ be the IG trawl (see above).  Recall that, for $t>0$,
	\begin{align*}
		Leb(A_t  \setminus  A) =\frac{\gamma}{\delta} \left( 1- \exp(\delta \gamma ( 1 - \alpha_t))\right), \quad Leb(A_t \cap A) =\frac{\gamma}{\delta} \exp(\delta \gamma ( 1 - \alpha_t)),
	\end{align*}
	where $\alpha_t = \sqrt{\frac{2t}{\gamma^2} + 1}$.
	
	We can show that
	\begin{align*}
		\frac{\partial}{\partial \delta} Leb(A_t \setminus A) &=  -\delta^{-1}Leb(A_t \setminus A)  -\gamma^2 \delta^{-1} (1-\alpha_t) \exp(\delta \gamma (1-\alpha_t)), \\
		\frac{\partial}{\partial \gamma} Leb(A_t \setminus A) &= -\gamma^{-1} Leb(A_t \setminus A)  - \gamma \exp(\delta \gamma (1-\alpha_t)) [1-\alpha_t + 2\gamma^{-2}\alpha_t^{-1} t], 
	\end{align*}
	and
	\begin{align*}
		\frac{\partial}{\partial \delta} Leb(A_t \cap A) &=  Leb(A_t \cap A)(\gamma(1-\alpha_t) - \delta^{-1}), \\
		\frac{\partial}{\partial \gamma} Leb(A_t \cap A) &=  Leb(A_t \cap A)(\gamma^{-1} + \delta (1-\alpha_t) + 2\delta \gamma^{-2} \alpha_t^{-1} t).
	\end{align*}

	\subsection{$\Gamma$ trawl function}\label{app:gradGam}
	Let $L'$ be a generic L\'evy seed and $d(s)$ be the $\Gamma$ trawl (see above). Recall that, for $t>0$,
	\begin{align*}
		Leb(A_t \setminus A) = \frac{\alpha}{H}\left(1- \left(1 + \frac{t}{\alpha}\right)^{-H}\right), \quad Leb(A_t  \cap  A) =\frac{\alpha}{H} \left(1 + \frac{t}{\alpha}\right)^{-H}.
	\end{align*}
	It is easy to show that
	\begin{align*}
		\frac{\partial}{\partial H} Leb(A_t \cap A) &=  - \frac{\alpha}{H} \left( 1 + \frac{t}{\alpha} \right)^{-H} \left( H^{-1}+ \log\left( 1+ \frac{t}{\alpha}\right) \right),  \\
		\frac{\partial}{\partial \alpha} Leb(A_t \cap A) &=   \left( 1 + \frac{t}{\alpha} \right)^{-(H+1)} \left( H^{-1} \left( 1+ \frac{t}{\alpha}\right) + \frac{t}{\alpha} \right).
	\end{align*}
	and
	\begin{align*}
		\frac{\partial}{\partial H} Leb(A_t \setminus A) &=   -\alpha H^{-2} - \frac{\partial}{\partial H} Leb(A_t \cap A), \\
		\frac{\partial}{\partial \alpha} Leb(A_t \setminus A) &= H^{-1} - \frac{\partial}{\partial \alpha} Leb(A_t \cap A).
	\end{align*}

	\clearpage
	
	\newpage

	\section{Additional calculations}\label{sec:add_calcs}
	
	\subsection{Calculations for the GIG trawl of Section \ref{sec.:GIGtrawl}}
	We have the trawl function
	\begin{align*}
		d(s) = \left( 1- \frac{2s}{\gamma^2}\right)^{-\nu/2} \frac{K_{\nu}\left(\delta \gamma \sqrt{1-\frac{2s}{\gamma^2}}\right)}{K_{\nu}(\delta \gamma)}.
	\end{align*}
	In the following we use the substitution $x = \sqrt{1+\frac{2s}{\gamma^2}}$ to get
	\begin{align*}
		Leb(A) = \int_0^{\infty} d(-s)ds &= \int_0^{\infty} \left( 1+ \frac{2s}{\gamma^2}\right)^{-\nu/2} \frac{K_{\nu}\left(\delta \gamma \sqrt{1+\frac{2s}{\gamma^2}}\right)}{K_{\nu}(\delta \gamma)} ds \\
		&=\int_1^{\infty} x^{-\nu + 1} \frac{K_{\nu}\left(\delta \gamma x\right)}{K_{\nu}(\delta \gamma)} \gamma^2 dx \\
		&= \frac{\gamma^2}{K_{\nu}(\delta \gamma)} \left(   \int_0^{\infty} x^{-\nu + 1}K_{\nu}\left(\delta \gamma x\right)   -   \int_0^1  x^{-\nu + 1}K_{\nu}\left(\delta \gamma x\right) \right).
	\end{align*}
	Now apply (6.561.12) and (6.561.16) in \cite{integrals} to get\footnote{Note, that we here need to impose $\nu < 1.$}
	\begin{align*}
		Leb(A) =\frac{\gamma}{\delta} \frac{K_{\nu -1}(\delta \gamma)}{K_{\nu}(\delta \gamma)}.
	\end{align*}
	Set $\alpha :=  \sqrt{\frac{2t}{\gamma^2} + 1}.$ Now, make the same substitution as above to get
	\begin{align*}
		Leb(A_t \cap A) = \int_t^{\infty} d(-s)ds &= \int_0^{\infty} \left( 1+ \frac{2s}{\gamma^2}\right)^{-\nu/2} \frac{K_{\nu}\left(\delta \gamma \sqrt{1+\frac{2s}{\gamma^2}}\right)}{K_{\nu}(\delta \gamma)} ds \\
		&=\int_{\alpha}^{\infty} x^{-\nu + 1} \frac{K_{\nu}\left(\delta \gamma x\right)}{K_{\nu}(\delta \gamma)} \gamma^2 dx.
	\end{align*}
	Set $y = \alpha^{-1} x$ to get
	\begin{align*}
		\int_{\alpha}^{\infty} x^{-\nu + 1} \frac{K_{\nu}\left(\delta \gamma x\right)}{K_{\nu}(\delta \gamma)} \gamma^2 dx &= \frac{\gamma^2}{K_{\nu}(\delta \gamma)} \int_{1}^{\infty} (\alpha y)^{-\nu + 1} K_{\nu}\left(\delta \gamma \alpha y\right) \alpha dy \\
		&= \frac{\gamma^2 \alpha^{-\nu+2}}{K_{\nu}(\delta \gamma)} \int_{1}^{\infty} y^{-\nu + 1} K_{\nu}\left(\delta \gamma \alpha y\right) dy.
	\end{align*}
	Now, splitting the integral as above and using the same formulae yields
	\begin{align*}
		Leb(A_t \cap A) =\frac{\gamma \alpha^{-\nu +1}}{ \delta} \frac{K_{\nu -1}(\delta \gamma \alpha)}{K_{\nu}(\delta \gamma)}.
	\end{align*}

	\subsection{Calculations for the IG trawl of Section \ref{sec:IG}}
	We have
	\begin{align*}
		d(s) =  \left( 1- \frac{2s}{\gamma^2}\right)^{-1/2} \exp\left( \delta \gamma \left( 1- \sqrt{ 1 - \frac{2s}{\gamma^2}} \right) \right),
	\end{align*}
	which means that
	\begin{align*}
		Leb(A) = \int_0^{\infty} d(-s) ds = \int_0^{\infty} \left( 1 + \frac{2s}{\gamma^2} \right)^{-1/2}  \exp\left( \delta \gamma \left( 1- \sqrt{ 1 + \frac{2s}{\gamma^2}} \right) \right) ds.
	\end{align*}
	So, after the change of variable $x =  \sqrt{ 1 + \frac{2s}{\gamma^2}}$ we have
	\begin{align*}
		Leb(A) = \int_0^{\infty} d(-s) ds &= \int_1^{\infty} x^{-1}  \exp\left( \delta \gamma ( 1- x) \right) \gamma^2 x dx \\
		&= \gamma^2 \int_1^{\infty}  \exp\left( \delta \gamma ( 1- x) \right)  dx \\
		&= \gamma^2 e^{\delta \gamma} \int_1^{\infty} e^{-\delta \gamma x}dx \\
		&= \frac{\gamma}{\delta}.
	\end{align*}
	Again, defining $\alpha :=  \sqrt{\frac{2t}{\gamma^2} + 1},$ we get by similar calculations
	\begin{align*}
		Leb(A_t \cap A) = \int_t^{\infty} d(-s) ds &= \gamma^2 e^{\delta \gamma} \int_{\alpha}^{\infty} e^{-\delta \gamma x}dx = \frac{\gamma}{\delta} e^{\delta \gamma ( 1 - \alpha)}.
	\end{align*}

	\clearpage
	
	\newpage
	
	\section{Weak dependence of trawl processes and asymptotic theory for the GMM approach}\label{sec:WeakDepGMM}
	
	In this section, we show that trawl processes, not necessarily restricted to the integer-valued case, are $\theta$-weakly dependent and we state and prove the asymptotic theory for the GMM approach to parameter estimation. 
	\subsection{Weak dependence of trawl processes}
	In this section, we show that (integer-valued) trawl processes are $\theta$-weakly dependent, see \citet[Definition 3.2]{CuratoStelzer2019}.
	
	Let us consider a (not necessarily integer-valued) L\'{e}vy seed $L'$ with characteristic triplet $(\gamma, a, \lev)$, i.e.~an infinitely divisible random variable with characteristic function given by 
	\begin{align}
		\label{eq:LevSeed}
		\Psi(\theta; L')=\mathbb{E}(\exp(i \theta L'))
		= \exp\left(i \theta \drift  - \frac{1}{2}\theta^2 a  + \int_{\mathbb{R}}\left(e^{i\theta \xi}-1-i\theta \xi \mathbb{I}_{[-1,1]}(\xi)\right) \lev(d\xi)\right),
	\end{align}
	for $\theta \in \R$.
	
	In the case of an integer-valued trawl process, we have
	\begin{align}
		\label{eq:LevSeedIV}
		\Psi(\theta; L')=\mathbb{E}(\exp(i \theta L'))
		= \exp\left( \int_{\mathbb{R}}\left(e^{i\theta \xi}-1\right) \lev(d\xi)\right)
		=\exp\left( \sum_{\xi}\left(e^{i\theta \xi}-1\right) \lev(\xi)\right),
	\end{align}
	for $\theta \in \R$. I.e.~in this case, the corresponding characteristic triplet is given by
	$(\gamma, 0, \lev)$, where $\gamma = \int_{\mathbb{R}}\xi \mathbb{I}_{[-1,1]}(\xi) \lev(d\xi)=\sum_{\xi=-1}^1\xi \lev(\xi)$.
	
	We set 
	$$
	A_t=\{(x, s): s\leq t, 0 \leq x \leq d(s-t)\},
	$$ for a function $d:(-\infty, 0]\to [0, \infty)$. 
	Let us also define a function $g:[0, \infty)\to [0, \infty)$ by $g(s):=d(-s)$, for all $s\geq 0$. 
	We note that the trawl process associated with the L\'{e}vy seed $L'$ can be expressed as 
	$X=(X_t)_{t\geq 0}$ with \begin{align}\label{eq:PT}
		X_t & =
		L(A_t) = \int_{(-\infty,t]\times \R}\ind_{(0,d(s-t))}(x)L(dx,ds)=\int_{(-\infty,t]\times \R}\ind_{(0,g(t-s))}(x)L(dx,ds)\\
		&=
		\int_{\R\times \R}f(x,t-s) L(dx,ds),
	\end{align}
	with  
	$f(x,t-s)=\ind_{(0,g(t-s))}(x)\ind_{[0,\infty)}(t-s)$, which is a special case of a causal mixed moving average processes as defined in 
	\citet[Definition 3.3]{CuratoStelzer2019}.
	Hence, using \citet[Corollary 3.4]{CuratoStelzer2019} and assuming that $\int_{|\xi|>1}|\xi|^2\lev(d\xi)<\infty$, we deduce that 
	the trawl process is $\theta$-weakly dependent in the sense of \citet[Definition 3.2]{CuratoStelzer2019} with coefficient, for $r\geq 0$, 
	\begin{align*}
		\theta_X(r)&=\left( \mathrm{Var}(L')\int_{(-\infty, -r)\times \R}\ind_{(0,g(-s))}^2(x)\ind_{[0,\infty)}^2(-s) dxds
		\right.\\
		&\left. 
		+
		\left| \E(L')\int_{(-\infty, -r)\times \R}\ind_{(0,g(-s))}(x)\ind_{[0,\infty)}(-s) dxds
		\right|^2
		\right)^{1/2}
		\\
		&=
		\left( \mathrm{Var}(L')\int_{-\infty}^{-r}g(-s)ds+ (\E(L'))^2\left(\int_{-\infty}^{-r}g(-s)ds\right)^2\right)^{1/2}
		\\
		&=
		\left( \mathrm{Var}(L')\int_{r}^{\infty}g(s)ds+ (\E(L'))^2\left(\int_{r}^{\infty}g(s)ds\right)^2\right)^{1/2}
		\\
		&=
		\left( \mathrm{Cov}(X_0, X_r)
		+\frac{(\E(L'))^2}{(\mathrm{Var}(L'))^2} (\mathrm{Cov}(X_0, X_r))^2\right)^{1/2},
	\end{align*}
	where   $\E(L')=\gamma+\int_{|\xi|>1}\lev(d\xi)$, and 
	$\mathrm{Var}(L')=a+\int_{\R}\xi^2\lev(d\xi)$.
	
	In the case when $L'$ is of finite variation, i.e.~when the characteristic triplet is given by $(\gamma, 0, \lev)$ with $\int_{\R}|\xi|\lev(\xi)< \infty$, which includes, in particular, integer-valued trawl processes, then the coefficient is, for $r\geq 0$,  given by
	\begin{align*}
		\theta_X(r)&=\int_{(-\infty, -r)\times \R}\int_{\R}|\ind_{(0,g(-s))}(x)\ind_{[0,\infty)}(-s)\xi|\lev(d\xi) dxds\\
		& +\int_{(-\infty, -r)\times \R}|\ind_{(0,g(-s))}(x)\ind_{[0,\infty)}(-s)\gamma_0| dxds\\
		&=\left(\int_{\R}|\xi|\lev(d\xi)+|\gamma_0|\right) \int_{r}^{\infty}g(s)ds=
		c\mathrm{Cov}(X_0, X_r),
	\end{align*}
	where $c=\left(\int_{\R}|\xi|\lev(d\xi)+|\gamma_0|\right)/\mathrm{Var}(L')$ and $\gamma_0=\gamma-\int_{|\xi|\leq 1}\xi\lev(d\xi)$.
	
	We note that, in the case of an integer-valued trawl, we have
	$
	\gamma_0=\gamma-\int_{|\xi|\leq 1}\xi\lev(d\xi)
	=\int_{\mathbb{R}}\xi \mathbb{I}_{[-1,1]}(\xi) \lev(d\xi)-\int_{|\xi|\leq 1}\xi\lev(d\xi)=0$ and, hence, $c=\int_{\R}|\xi|\lev(d\xi)=\sum_{\xi}|\xi|\lev(\xi)$.

	We note that, as pointed out in \citet[p.~324]{CuratoStelzer2019} and shown in the discrete-time case in \cite{Doukhanetal2012}, for integer-valued trawl processes, the fact that IVT processes are $\theta$-weakly dependent, implies that they are strongly mixing.

	\subsection{GMM estimation for trawl processes}
	In \cite{BNLSV2014}, the authors proposed estimating the trawl parameters via a (generalised) method of moments (G)MM. We shall now derive the corresponding asymptotic theory.
	
	Consider the equidistantly sampled process $X_{\Delta}, X_{2\Delta}, \ldots, X_{n\Delta}$, for $\Delta=T/n >0, T>0, n \in \N$. The GMM estimator is based on the sample mean, sample variance and sample autocovariances up to lag $m\geq 2$. Consider the vector
	$$
	Y_t^{(m)}=(X_{t\Delta}, X_{(t+1)\Delta}, \ldots, X_{(t+m)\Delta}),
	$$
	for $t=1, \ldots, n-m$. Let $\Theta$ denote the parameter space of the trawl process and set $\mu:=\mu(\theta)=\E(X_0)$ and $D(k):=D(k,\theta):=\E(X_0X_{k\Delta})$, for $k=0, \ldots, m$. For a given parametric model of $X$, $D(k)$ is just a function of the model parameter(s) $\theta$.
	
	Define the measurable function $h:\R^{m+1}\times \Theta\to \R^{m+2}$ by
	\begin{align*}
		h(Y_t^{(m)}, \theta)&=
		\left(
		\begin{array}{c}
			h_E(Y_t^{(m)}, \theta)\\
			h_0(Y_t^{(m)}, \theta)\\
			h_1(Y_t^{(m)}, \theta)\\
			\vdots
			\\
			h_m(Y_t^{(m)}, \theta)\\
		\end{array}
		\right)
		=
		\left(
		\begin{array}{c}X_{t\Delta}-\mu(\theta) \\
			X_{t\Delta}^2-D(0,\theta)\\
			X_{t\Delta}X_{(t+1)\Delta}-D(1,\theta)\\
			\vdots
			\\
			X_{t\Delta}X_{(t+m)\Delta}-
			D(m,\theta)
		\end{array}
		\right).
	\end{align*}
	The corresponding sample moments can be defined as 
	\begin{align*}
		g_{n,m}(\theta)&=
		\frac{1}{n-m}\sum_{t=1}^{n-m}h(Y_t^{(m)}, \theta)=
		\left(
		\begin{array}{c}
			\frac{1}{n-m}\sum_{t=1}^{n-m}h_E(Y_t^{(m)}, \theta)\\
			\frac{1}{n-m}\sum_{t=1}^{n-m}h_0(Y_t^{(m)}, \theta)\\
			\vdots
			\\
			\frac{1}{n-m}\sum_{t=1}^{n-m}h_m(Y_t^{(m)}, \theta)
		\end{array}
		\right).
	\end{align*}
	We can then estimate the true parameter $\theta_0$, say,  by minimising the objective function of the GMM, which leads to the estimator
	\begin{align}\label{eq:GMMest}
		\widehat\theta_{0,\mathrm{GMM}}^{n,m}
		=\mathrm{argmin} 
		g_{n,m}(\theta)^{\top} A_{n,m}
		g_{n,m}(\theta),
	\end{align}
	where $A_{n,m}$ is the positive-definite weight matrix of the $m+2$ moments considered.
	
	We would like to derive a central limit theorem for the GMM estimator. As a first step, as in \citep[Section 6.1 for supOU processes]{CuratoStelzer2019}, we derive a central limit theorem for the moment function $h(Y_t^{(m)},\theta_0)$:
	\begin{theorem}\label{prop:clt-momentfct}
		Consider a trawl process $X$ with characteristic triplet $(\drift, a, \lev)$ and suppose that
		$\int_{|\xi|>1}|\xi|^{4+\delta}\lev(d\xi)<\infty$, for some $\delta>0$ and suppose that the $\theta$-weakly dependence coefficient of the trawl process is given by $\theta_X(r)=O(r^{-\alpha})$, for $\alpha>\left( 1+\frac{1}{\delta}\right)\left(1+\frac{1}{2+\delta} \right)$.
		Set $
		Y_t^{(m)}=(X_{t\Delta}, X_{(t+1)\Delta}, \ldots, X_{(t+m)\Delta})$,
		for $t=1, \ldots, n-m$.
		Then $h(Y_t^{(m)}, \theta_0)$ is a $\theta$-weakly dependent process, the matrix
		$$
		\Sigma_{a}=\sum_{l\in \Z}\mathrm{Cov}(h(Y_0^{(m)}, \theta_0), h(Y_l^{(m)}, \theta_0))
		$$
		is finite, positive definite and, as $n\to \infty$,
		$$
		\sqrt{n}g_{n,m}(\theta_0)
		\stackrel{d}{\to}\mathrm{N}(0, \Sigma_a).
		$$
	\end{theorem}
	\begin{proof}[Proof of Theorem \ref{prop:clt-momentfct}]
		We note that $Y=(Y_t^{(m)})_{t\in \R}$ can be represented as a causal $(m+1)$-dimensional mixed moving average process given by
		\begin{align*}
			Y_t^{(m)}= \int_{(-\infty,t\Delta]\times \R}
			\left( \begin{array}{c}
				\ind_{(0,g(t\Delta-s))}(x)
				\\
				\ind_{(0,g((t-1)\Delta-s))}(x)\\
				\vdots\\
				\ind_{(0,g((t-m)\Delta-s))}(x)
			\end{array} \right)
			L(dx,ds)
			=\int_{(-\infty,t\Delta]\times \R}
			\left( \begin{array}{c}
				\ind_{(0,g(t\Delta-s))}(x)
				\\
				\ind_{(0,g(t\Delta-s-\Delta))}(x)\\
				\vdots\\
				\ind_{(0,g(t\Delta-s-m\Delta))}(x)
			\end{array} \right)
			L(dx,ds),
		\end{align*}
		which is $\theta$-weak dependent with coefficient $\theta(r)=\mathcal{D}\theta_X(r-m\Delta)$, for $r\geq m \Delta$, where $\mathcal{D}=(m\Delta +1)^{1/2}$ for general trawl processes and $\mathcal{D}=(m\Delta +1)$ in the finite variation case, see \citet[Proposition 4.1]{CuratoStelzer2019}
		
		Note that the condition $\int_{|\xi|>1}|\xi|^{4+\delta}\lev(d\xi)<\infty$  implies the existence of the $(4+\delta)$-moment of the trawl process.
		Define a function $H: \R^{m+1} \to \R^{m+2}$ such that
		\begin{align*}
			H(Y_t^{(m)})&=h(Y_t^{(m)}, \theta_0)+\left(
			\begin{array}{c}
				\mu(\theta_0)\\
				D(0, \theta_0)\\
				\vdots\\
				D(m, \theta_0)
			\end{array}
			\right)
			%
			=\left(
			\begin{array}{c}X_{t\Delta}-\mu(\theta_0) \\
				X_{t\Delta}^2-D(0, \theta_0)\\
				X_{t\Delta}X_{(t+1)\Delta}-D(1, \theta_0)\\
				\vdots
				\\
				X_{t\Delta}X_{(t+m)\Delta}-
				D(m, \theta_0)
			\end{array}
			\right)+\left(
			\begin{array}{c}
				\mu(\theta_0)\\
				D(0, \theta_0)\\
				D(1, \theta_0)\\
				\vdots\\
				D(m, \theta_0)
			\end{array}
			\right)\\
			&=\left(
			\begin{array}{c}X_{t\Delta} \\
				X_{t\Delta}^2\\
				X_{t\Delta}X_{(t+1)\Delta}\\
				\vdots
				\\
				X_{t\Delta}X_{(t+m)\Delta}
			\end{array}
			\right).
		\end{align*}
		\cite{CuratoStelzer2019} showed that the function $H$ satisfies the  conditions of \citet[Proposition 3.4]{CuratoStelzer2019}  for $p=4+\delta, c=1, a=2$.
	Hence, according to \citet[Proposition 3.4]{CuratoStelzer2019},  $H(Y_t^{(m)})$ is a $\theta$-weakly dependent process with coefficient 
	$\mathcal{C}(\mathcal{D}\theta_X(r-m\Delta))^{\frac{2+\delta}{3+\delta}}$, for $r\geq m \Delta$, 
	for a constant $\mathcal{C}>0$ independent of $r$.
	We can now deduce that $h(Y_t^{(m)}, \theta_0)$ is a zero-mean,  $\theta$-weakly dependent process with the same coefficient.
	Applying the Cramer-Wold device and 
	\citet[Theorem 1]{Dedecker2000} allows us to conclude as in the proof of \citet[Theorem 6.1]{CuratoStelzer2019}, where we note that the moment condition appearing in \citet[Theorem 1]{Dedecker2000} is implied for a weakly $\theta$-dependent process with coefficient $\theta(r)=O(r^{-\alpha^*})$, for $\alpha^*>1+1/\delta$, see e.g.~\cite{CSS2022}.
	I.e.~in our setting we require that $\alpha^*=\alpha\frac{2+\delta}{3+\delta} >1+1/\delta
	\Leftrightarrow \alpha>(1+\frac{1}{\delta})(1+\frac{1}{2+\delta})$.
\end{proof}

Let us now formulate the technical assumptions for the weak consistency and the central limit theorem of the GMM estimator.

We start off with the assumptions which guarantee weak consistency, cf.~\citet[Assumptions 1.1-1.3]{Matyas1999}
\begin{assumption}\label{as:1.1}
\begin{enumerate}
	\item[(i)] Suppose that $\E(h(Y_t^{(m)}, \theta))$ exists and is finite for all $\theta\in \Theta$ and for all $t$.
	\item[(ii)] Set $h_t^{(m)}(\theta)=\E(h(Y_t^{(m)}, \theta))$. There exists a $\theta_0\in \Theta$ such that $h_t^{(m)}(\theta)=0$ for all $t$ if and only if $\theta=\theta_0$. 
\end{enumerate}
\end{assumption}
We note that by construction, Assumption \ref{as:1.1} (i) is satisfied in our setting under suitable moment conditions on $L'$, whereas (ii)  needs to be verified for the specific parametric case of interest.

Next, we impose an assumption on the convergence of the sample moments to the population moments. To this end, let $h^{(m)}(\theta)=\sum_{t=1}^{n-m}h_t^{(m)}(\theta)$. We denote the $j$th component of the $m+2$-dimensional vectors $h^{(m)}(\theta)$ and $g_{n,m}(\theta)$ by $h^{(m)}_j(\theta)$ and $g_{n,m; j}(\theta)$, respectively.
\begin{assumption}\label{as:1.2}
Suppose that, for $j=1, \ldots, m+2$, as $n\to \infty$, 
$$
\sup_{\theta \in \Theta}|h^{(m)}_j(\theta) -g_{n,m; j}(\theta)|\stackrel{\mathbb{P}}{\rightarrow} 0,
$$
\end{assumption}
The next assumption concerns the convergence of the weighting matrix:
\begin{assumption}\label{as:1.3}
There exists a sequence of non-random, positive definite matrices $\overline{A}_{n,m}$ such that, as $n \to \infty$, 
$|A_{n,m}-\overline{A}_{n,m}|\stackrel{\mathbb{P}}{\rightarrow} 0$.
\end{assumption}

\begin{theorem}\label{thm:consistency}
Assume that Assumptions \ref{as:1.1}, \ref{as:1.2}, \ref{as:1.3} hold. Then the GMM estimator $\widehat\theta_{0,\mathrm{GMM}}^{n,m}$ defined in \eqref{eq:GMMest}
is weakly consistent.
\end{theorem}
\begin{proof}
This is an immediate consequence of \citet[Theorem 1.1]{Matyas1999}.
\end{proof}
\begin{remark}
We note that \citet[p.~14--17]{Matyas1999} discusses alternative (sufficient) assumptions which might be easier to check in practice.
\end{remark}
\begin{assumption}\label{as:1}
$\Theta$ is a compact  parameter space which includes the true parameter $\theta_0$. 
\end{assumption}
\begin{remark}
We note that, in practice, we would often impose bounds on the parameter space $\Theta$, so even if the true parameter constraints on $\Theta$ might not necessarily imply a compact space, it can typically be chosen to be compact when imposing suitable constraints in the optimisation. 
\end{remark}
\begin{assumption}\label{as:2}
The weight matrix $A_{n,m}$ converges in probability to a positive definite matrix $A$.
\end{assumption}
\begin{assumption}\label{as:3}
The covariance matrix $\Sigma_a$ is positive definite.
\end{assumption}

\begin{theorem}\label{prop:clt-gmm}
Consider a trawl process $X$ with characteristic triplet $(\drift, a, \lev)$ and suppose that
$\int_{|\xi|>1}|\xi|^{4+\delta}\lev(d\xi)<\infty$, for some $\delta>0$ and and suppose that the $\theta$-weakly dependence coefficient of the trawl process is given by $\theta_X(r)=O(r^{-\alpha})$, for $\alpha>\left( 1+\frac{1}{\delta}\right)\left(1+\frac{1}{2+\delta} \right)$.
Suppose that Assumptions \ref{as:1.1},  \ref{as:1}, \ref{as:2}, \ref{as:3} hold.
Then, as $n\to \infty$, 
$$
\sqrt{n}(\widehat\theta_{0,\mathrm{GMM}}^{n,m}-\theta_0)
\stackrel{d}{\to}\mathrm{N}(0, M\Sigma_aM^{\top}),
$$
where
\begin{align*}
	\Sigma_{a}&=\sum_{l\in \Z}\mathrm{Cov}(h(Y_0^{(m)}, \theta_0), h(Y_l^{(m)}, \theta_0)),\\
	M&=(G_0^{\top}AG_0)^{-1}G_0^{\top}A, \quad \mathrm{where} \; \;\;
	G_0=\E\left[\frac{\partial h(Y_t^{(m)}, \theta)}{\partial \theta^{\top}} \right]_{\theta=\theta_0}.
\end{align*}
\end{theorem}
\begin{proof}[Proof of Theorem \ref{prop:clt-gmm}]
The proof follows the strategy of the proof of \citet[Theorem 1.2]{Matyas1999}, see also \citet[Proof of Theorem 6.2]{CuratoStelzer2019} for the case of a supOU process.
Hence, we verify the Assumptions 1.1-1.3 and 1.7-1.9 in \citet[Chapter 1]{Matyas1999}.

We note that Assumption 1.1 in \cite{Matyas1999} is implied by Assumption \ref{as:1.1}. 

For Assumption 1.2 in \cite{Matyas1999}, see our Assumption \ref{as:1.2}, we verify the corresponding sufficient conditions Assumption 1.4, 1.5 and 1.6 in \cite{Matyas1999}. 

Assumption 1.3 in \cite{Matyas1999} holds due to Assumption \ref{as:2}.

Assumption 1.4 in \cite{Matyas1999} is implied by our Assumption \ref{as:1}.

Assumption 1.5 in \cite{Matyas1999} is satisfied since the trawl process is a special case of a mixed moving average process and hence mixing and ergodic, see \cite{FS2013}.

Assumption 1.6 in \cite{Matyas1999} implies that we need to show that each component of the function $h$ satisfies a (stochastic) Lipschitz condition. 
For the first component, we have 
$$
|h_E(Y_t^{(m)},\theta_1)-h_E(Y_t^{(m)},\theta_2)|=
|-\mu(\theta_1)+\mu(\theta_2)|,
$$
and similar expressions hold for the other components. In all cases, we observe that the random components cancel out and, hence, we only require a Lipschitz condition for the deterministic parts. 
We note that, when taking partial derivatives with respect to the model parameters, 
we get that these partial derivatives are bounded and, hence, the components are Lipschitz continuous and Assumption 1.6 in \cite{Matyas1999} holds.

Assumption 1.7 in \cite{Matyas1999} holds by construction. Then $h(Y_t^{(m)}, \theta)$ is continuously differentiable with respect to $\theta \in \Theta$.
We set 
\begin{align*}
	G_{n,m}(\theta):=\frac{1}{n-m}\sum_{t=1}^{n-m}\frac{\partial h(Y_t^{(m)}, \theta)}{\partial \theta^{\top}}.
\end{align*}
Note that Assumption 1.8 in \cite{Matyas1999} holds if we can show that $\frac{\partial h(Y_t^{(m)}, \theta)}{\partial \theta^{\top}}$ satisfies a weak law of large numbers in a neighbourhood of $\theta_0$. I.e.~we need to show that, for any sequence $(\theta_n^*)$ such that $\theta_n^* \stackrel{\mathbb{P}}{\to} \theta_0$, we have $G_{n,m}(\theta_n^*)\stackrel{\mathbb{P}}{\to} G_0$.
From the definition of $h(Y_t^{(m)}, \theta)$, we can read off that the partial derivative matrix $\frac{\partial h(Y_t^{(m)}, \theta)}{\partial \theta^{\top}}$ does not depend on $Y_t^{(m)}$, which implies that $G_{n,m}(\theta):=\frac{\partial h(Y_t^{(m)}, \theta)}{\partial \theta^{\top}}$ and $G_{0}=\E\left[\frac{\partial h(Y_t^{(m)}, \theta)}{\partial \theta^{\top}} \right]_{\theta=\theta_0}
=\left.\frac{\partial h(Y_t^{(m)}, \theta)}{\partial \theta^{\top}}\right|_{\theta=\theta_0}$.
We can now apply the continuous mapping theorem, to deduce that Assumption 1.8 in \cite{Matyas1999} holds.

Finally, we need to justify that Assumption 1.9 in \cite{Matyas1999} holds. However, this is a direct consequence of Theorem \ref{prop:clt-momentfct}.

Hence, the same steps as in the proof of \citet[Theorem 2.1]{Matyas1999} can be applied since Assumption \ref{as:3} holds, where $f_T$ and $F_T$ need to be replaced by $g_{n,m}$ and $G_{n,m}$.
\end{proof}

Let us now study some examples when the condition that the $\theta$-weakly dependence coefficient of the trawl process is given by $\theta_X(r)=O(r^{-\alpha})$, for $\alpha>\left( 1+\frac{1}{\delta}\right)\left(1+\frac{1}{2+\delta} \right)$.
Suppose in the following two examples that we consider an integer-valued trawl process, then $\theta_X(r)\propto \mathrm{Cov}(X_0, X_r)$.
\begin{example}
In the case of an exponential trawl function with parameter $\lambda>0$, we have that $\theta_X(r) \propto e^{-\lambda r}$, which decays even faster than any polynomial decay, so the condition that $\theta_X(r)=O(r^{-\alpha})$, for $\alpha>\left( 1+\frac{1}{\delta}\right)\left(1+\frac{1}{2+\delta} \right)$ holds. 
\end{example}
\begin{example}
In the case of a Gamma trawl function with parameters $\alpha>0, H>0$, we have that $\theta_X(r) \propto r^{-H}$. So the condition in the two theorems implies that we need that
$H>\left( 1+\frac{1}{\delta}\right)\left(1+\frac{1}{2+\delta} \right)>1$, which excludes the long-memory setting.
\end{example}

\subsection{Comparison of the asymptotic variances of the MCL and GMM estimators: The case of the Poisson-Exponential IVT process}\label{sec:AVARcomp} 
In Section \ref{sec:MC}, the GMM and MCL estimators were compared in finite samples. By considering the CLT of the MCL estimator (Theorem  \ref{th:CLT}) and that of the GMM estimator (Theorem  \ref{prop:clt-gmm}), we may also compare the two estimators analytically by comparing the values of the asymptotic variance (AVAR) of the two estimators. Indeed, in the short-memory case, we have results to the effect that (Theorem  \ref{th:CLT} and Theorem \ref{prop:clt-gmm})
\begin{align}\label{eq:AVAR}
\sqrt{n}( \hat \theta_x - \theta_0) \stackrel{(d)}{\rightarrow} N(0,AVAR_x), \qquad x = MCL, GMM,
\end{align}
where $AVAR_x$ is the asymptotic variance matrix for $x = MCL, GMM$. Both approaches to calculating $AVAR_x$ rely on terms for which we do not have closed-form expressions, but that we instead have to estimate using simulations.\footnote{This holds in particular for the $V$ matrix in the MCL CLT, see Theorem  \ref{th:CLT}  and the discussion following it. It also holds for the $\Sigma_a$ matrix in the GMM CLT, see Theorem  \ref{prop:clt-gmm}.} We illustrate this in the case of the Poisson-Exponential IVT model with $\nu = 17.50$ and $\lambda = 1.80$ (same setup as in the paper, cf. Table \ref{tab:paramTab}). We set $K = 10$ for the MCL estimator and $m=10$ for the GMM estimator.  We use the estimation-based method with $B = N = 500$ (Section \ref{app:B1}) to calculate the $V$ matrix of the MCL estimator of Theorem  \ref{th:CLT}\footnote{We also need to estimate the $H$ matrix of  Theorem  \ref{th:CLT}, which we would normally get as standard output from the numerical MCL estimation procedure. Here, we simulate one instance of a very long ($n = 20,000$) Poisson-Exponential IVT process and use this to estimate $H$ via standard output from the numerical optimizer.}; to calculate $\Sigma_a$ of Theorem \ref{prop:clt-gmm}, we use a similar simulation-based approach. To be precise, for $b = 1, 2, \ldots, B$ with $B = 500$, we simulate $N = 500$ observations of a Poisson-Exponential IVT process (with $\nu = 17.5$ and $\lambda = 1.8$) $Y$, which we use to calculate $\Sigma_a^{(b)} = \mathrm{Cov}(h(Y_0^{(m)}, \theta_0), h(Y_0^{(m)}, \theta_0)) + 2\sum_{l=1}^{50} \mathrm{Cov}(h(Y_0^{(m)}, \theta_0), h(Y_l^{(m)}, \theta_0))$. Then we estimate $\Sigma_a \approx \frac{1}{B} \sum_{b=1}^B \Sigma_a^{(b)}$.



We run the above simulation-based procedure $100$ times. The simulation-based estimated values for the diagonal element of $\sqrt{AVAR_x}$, for the runs $i = 1,2, \ldots M$, are shown in Figure \ref{fig:AVAR}. The diagonal elements of $AVAR_{MCL}$ (red crosses) are always smaller than those of  $AVAR_{GMM}$ (blue circles), and the simulation-based estimates of the latter are much more volatile. The ratio of the average values of the diagonal elements of $\sqrt{AVAR_{MCL}}$ to those of $\sqrt{AVAR_{GMM}}$ are $0.50$ and $0.51$ for $\nu$ and $\lambda$, respectively. Incidentally, these numbers are close to those found in the finite sample comparison between the two methods, cf. Figure \ref{fig:CLvsMM}. 

\begin{figure}[!t]
		\centering
		\includegraphics[scale=0.80]{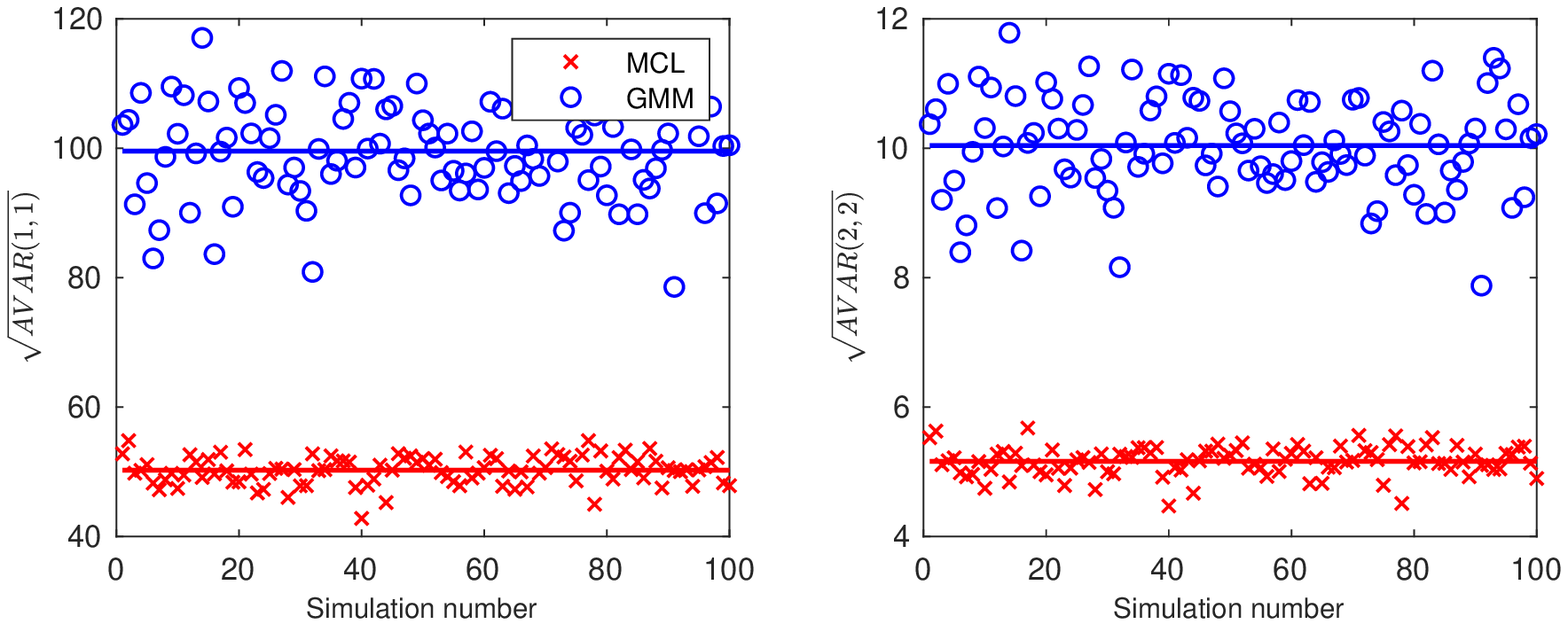} 
		\caption{\it Left plot: Simulation-based estimates of the asymptotic variance of $\hat \nu$, i.e.~$\sqrt{AVAR_x(1,1)}$ for $x = MCL, GMM$, where $AVAR_x$ is the asymptotic covariance matrix of $\hat \theta$, see Equation \eqref{eq:AVAR}. Right plot: Simulation-based estimates of the asymptotic variance of $\hat \lambda$, i.e.~$\sqrt{AVAR_x(2,2)}$ for $x = MCL, GMM$. Horizontal lines denote the average over the $100$ simulations.}
		\label{fig:AVAR}
		\end{figure}

\newpage

	\section{Software (MATLAB)}\label{sec:code}
	The following functions are available in the MATLAB software language. We give a very brief description of the functions here but refer to the extensive documentation in the code for further details. The code can be freely downloaded from \if1\blind{GitHub.} \fi 
	
	\if0\blind{\noindent  \url{https://github.com/mbennedsen/Likelihood-based-IVT}.} \fi
	
	\begin{itemize}
		\item
		\texttt{simulateIVT}: 
		\begin{itemize}
			\item	
			Simulates equidistant observations of a parametric IVT process, specified by a L\'evy basis and a trawl function. The L\'evy basis and trawl function can be specified independently of each other using the framework described in this Supplementary Material.
		\end{itemize}
		\item
		\texttt{estimateIVT}: 
		\begin{itemize}
			\item	
			Takes as input a vector of equidistantly spaced observations and a parametric specification (L\'evy basis and trawl function) and outputs estimates of the corresponding parameters using the maximum composite likelihood approach developed in the main paper.
		\end{itemize}
		\item
		\texttt{modelselectIVT}: 
		\begin{itemize}
			\item	
			This function estimates six parametric IVT models (Poisson-Exponential, Poisson-IG, Poisson-Gamma, NB-Exponential, NB-IG, NB-Gamma) and calculates the composite likelihood function when evaluated in the optimized parameters, as well as the CLAIC and CLBIC criteria given in the main paper. These three criteria can be used for model selection, with larger values indicating a better fit.
		\end{itemize}
		\item
		\texttt{forecastIVT}: 
		\begin{itemize}
			\item	
			Takes as input a parametric IVT model (L\'evy basis and trawl function), a forecast horizon (which can be a vector of several forecast horizons), as well as historical observations; the output is the predictive probability distribution for the given forecast horizons. The parameters underlying the predictive distribution are estimated using the maximum composite likelihood approach presented in the main paper.
		\end{itemize}
		\item
		\texttt{analyze\_stock\_A} and \texttt{analyze\_simulated\_data}: 
		\begin{itemize}
			\item	
			These files illustrate the use of the functions \texttt{simulateIVT}, \texttt{estimateIVT}, \texttt{modelselectIVT}, and \texttt{forecastIVT}. The file \texttt{analyze\_stock\_A} reproduces the output of the main paper, while  \texttt{analyze\_simulated\_data} simulates a user-specified IVT process and then conducts analyses similar to those considered in the main paper on these simulated data.
		\end{itemize}
	\end{itemize}

\end{document}